\documentclass[acmsmall]{acmart}

\usepackage{booktabs}
\usepackage{hyperref}
\usepackage{url}

\usepackage{multirow}
\usepackage{color}
\newtheorem{strategy}{Strategy}

\usepackage{appendix}
\usepackage{lineno}
\usepackage{diagbox}
\usepackage{amsmath}
\usepackage{arydshln}
\usepackage{color}

\usepackage{algorithmic}

\usepackage{tabularx}
\usepackage{makecell}

\usepackage[linesnumbered,ruled]{algorithm2e}

\SetAlFnt{\small}
\SetAlCapFnt{\small}
\SetAlCapNameFnt{\small}
\SetAlCapHSkip{0pt}
\IncMargin{-\parindent}
\usepackage{longtable}

\usepackage{soul}
\usepackage{color, xcolor}
\sethlcolor{yellow}
\soulregister{\cite}7
\soulregister{\citep}7
\soulregister{\citet}7
\soulregister{\ref}7
\soulregister{\pageref}7
\soulregister{\underline}7

\acmJournal{JACM}

\setcopyright{rightsretained}

\begin{document}


\title{Targeted Sequential Pattern Mining with High Average Utility}

\author{Kai Cao}
\affiliation{ 
	\institution{Hainan University}
	\city{Haikou 570228}
	\country{China}
}
\email{caokai.pds@gmail.com}

\author{Yucong Duan}
\authornote{This is the corresponding author.}
\affiliation{
	\institution{Hainan University}
	\city{Haikou 570228}
	\country{China}
}
\email{duanyucong@hotmail.com}

\author{Wensheng Gan}
\affiliation{
	\institution{Jinan University}
	\city{Guangzhou}
	\country{China}
}
\email{wsgan001@gmail.com}

\begin{abstract}
  Incorporating utility into targeted pattern mining can address the practical limitations of traditional frequency-based approaches. However, utility-based methods often suffer from generating a large number of long and complicated sequences. To improve pattern relevance and interpretability, average utility provides a more balanced metric by considering both utility and sequence length. Moreover, incorporating user-defined query targets into the mining process enhances usability and interactivity by retaining only patterns containing user-specified goals. To address challenges related to mining efficiency in large-scale, long-sequence datasets, this study introduces average utility into targeted sequential pattern mining. A novel algorithm, TAUSQ-PG, is designed to find targeted high average utility sequential patterns. It incorporates efficient filtering and pruning strategies, tighter upper bound models, as well as novel specialized evaluation metrics and query flags tailored to this task. Extensive comparative experiments on different datasets demonstrate that TAUSQ-PG effectively controls the candidate set size, thereby reducing redundant sequence generation and significantly improving runtime and memory efficiency.
\end{abstract}

\ccsdesc[500]{Information systems~Information systems applications~Data mining}

\keywords{targeted pattern mining, sequence data, average utility, upper bound, pruning strategies}

\maketitle

\renewcommand{\shortauthors}{K. Cao \textit{et al.}}

\section{Introduction}  \label{sec: introduction}

The proliferation of large-scale sensors and smart devices has significantly enhanced the collection of diverse real-world data, thereby intensifying the need for more efficient data mining and analysis techniques. Among the various frequency-based methods used to discover interesting patterns in transactional databases, sequential pattern mining (SPM) \cite{agrawal1995mining, husak2017sequential} and frequent pattern mining (FPM) \cite{agrawal1996fast, han2004mining} are two representative approaches. The pioneering FPM method was proposed by Agrawal et al. in 1993 \cite{han2007frequent}, while SPM focuses on uncovering frequent sequential patterns from a sequential database. The introduction of high utility pattern mining (HUPM) and high utility itemset mining (HUIM) \cite{gan2019survey, nguyen2019efficient, liu2018opportunistic} marked a departure from the early assumption that high frequency directly correlates with high relevance. In practical scenarios, alternative measures of interestingness, such as utility, are often more critical than simple frequency. For instance, in the retail industry, profit (utility) frequently takes precedence over sales volume. HUPM/HUIM aimed at identifying patterns with higher utility \cite{gan2018survey}. 

However, in HUPM/HUIM, longer patterns tend to accumulate higher utility, which can lead to overly complex results \cite{lee2022efficient}. This contradicts the original intent of pattern mining—to reveal actionable and insightful knowledge. To overcome this issue, average utility was introduced \cite{hong2009mining}, refining traditional utility-based evaluations by jointly considering both the length of a pattern and its utility. Based on this, the high average utility pattern mining (HAUPM) and the high average utility itemset mining (HAUIM) \cite{kim2017efficient, lan2012efficiently, lan2012projection} are proposed to extract more compact and practically meaningful patterns in real-world applications. Utility-based pattern mining techniques have demonstrated wide applicability across various domains, including e-commerce \cite{shie2013mining} (for identifying profitable product combinations and supporting cross-selling strategies), internet of things analytics \cite{srivastava2020large} (for detecting event sequences that critically affect system performance), bioinformatics \cite{zihayat2017mining} (for revealing significant gene expression patterns). In addition, the average utility provides a more balanced and practical evaluation metric than the total utility in bioinformatics \cite{segura2022mining} or in spatial data analysis \cite{tran2024efficiently}.

Nevertheless, even HAUPM and HAUIM may still generate a vast collection of patterns that meet the specified threshold, making the final results difficult to interpret and apply. To reduce redundancy, techniques such as top-\textit{k} pattern mining and closed pattern mining have been proposed. However, these methods typically focus on structural properties or utility ranking and may not necessarily align with specific user interests or intentions. In contrast, targeted pattern mining (TPM), also known as targeted pattern query, emphasizes user-driven discovery by filtering out irrelevant results and extracting only patterns that contain user-defined target subsequences. Compared to conventional pattern mining, TPM provides a more concise, interactive, and user-centric framework. However, identifying subsets of the potential search space in TPM poses substantial computational challenges, particularly when attempting to estimate the utility of candidate patterns that do not yet meet the specified parameters \cite{fournier2022pattern}. Although naive post-processing can also achieve targeted querying, it suffers from excessive time and memory consumption, making it impractical for real-time applications \cite{zhang2022tusq}.

To tackle the aforementioned limitations, we introduce a novel TPM model, termed targeted high average utility sequential pattern mining (TAUSPM). By integrating average utility with targeted sequential pattern queries, TAUSPM offers notable advantages. The average utility reduces the pattern-length bias common in utility-based mining, leading to more meaningful and representative results \cite{zhang2022tusq}. Meanwhile, targeted querying narrows the search to the given sequence, optimizing both the performance and the relevance of pattern discovery. From the perspective of TPM, target patterns are not merely used for fast queries—they serve a deeper analytical role. In sequential databases, it is preferable that the sequences containing these target patterns also exhibit relatively high average utility in the corresponding segments. Instead of relying on total utility, the objective is to identify patterns whose average utility meets user-defined thresholds, as these are often more indicative of meaningful or critical insights. The use of average utility further enhances the typicality or representativeness of discovered patterns. For example, in gene expression analysis, identifying sequences that contain specific nucleotide subsequences associated with genetic disorders can support therapeutic development. More importantly, verifying whether these sequences are typical representations of such associations provides deeper insight into the molecular mechanisms of the disorders and a stronger theoretical foundation for precise gene therapy strategies. 

This study makes the following primary contributions:
\begin{itemize}
    \item  This study incorporates the notion of average utility into the TPM framework for sequential data and formally defines a new problem that focuses on identifying a complete yet compact set of patterns evaluated by average utility.
	
    \item  This work designs two efficient variants of upper bound models (UBs) and corresponding pruning strategies based on the characteristics of TPM tasks. Two specialized flags combined with the position comparison method are proposed to enhance query efficiency.
	
    \item  A novel and efficient algorithm, called TAUSQ-PG, is proposed and is extensively evaluated on various datasets. The results of comparative experiments demonstrate its remarkable advantages in both effectiveness and efficiency when contrasted with baseline methods.
\end{itemize}

The remainder of this paper is organized as follows. A concise review of related work is stated in Section \ref{sec: relatedWork}. Section \ref{sec: preliminary} delineates the formulation of TAUSPM problem and introduces essential definitions. The algorithm TAUSQ-PG is described in detail in Section \ref{sec: algorithm}, including several optimization strategies and supporting data structures. Then, the performance evaluation of TAUSQ-PG is conducted through comparative experiments in Section \ref{sec: experiment}. Finally, the contributions and outcomes of this research are summarized in Section \ref{sec: conclusion}.

\section{Related Work} \label{sec: relatedWork}

This section provides an overview of three major elements relevant to our research: HUSPM, HAUSPM, and TPM.

\subsection{High-Utility Sequential Pattern Mining} \label{sec: huspm}

SPM was initially proposed in 1995 by Agrawal and Srikant \cite{agrawal1995mining} for the analysis of customer purchase records. Ahmed et al. extended SPM to incorporate the concept of utility and formally introduced the problem of HUSPM. Their work proposed two two-phase algorithms, including Utility Span, which employed a pattern growth approach to control candidate generation. Subsequent HUSPM algorithms focused on designing efficient data structures for utility computation and pruning. USpan \cite{yin2012uspan} introduced the lexicographic quantitative sequence tree (LQS-tree); ProUM \cite{gan2020proum} utilized a data structure named utility-array; HUSP-ULL \cite{gan2020fast} adopted the UL-list; and HUSP-SP \cite{zhang2023husp} developed the seqPro structure. Efficient indexing strategies \cite{lan2014applying} further enhanced the performance of projected databases. Another major focus in HUSPM is the design of UBs to prune unpromising candidates. PHUS \cite{lan2014applying} used maximum utility as a measure to simplify the evaluation of HUSPM and defined the sequence utility upper bound (SUUB). HuspExt \cite{alkan2015crom} designed a tighter upper bound named CRoM. Two tighter utility UBs, PEU and RSU, were proposed in HUS-Span \cite{wang2016efficiently}. ProUM \cite{gan2020proum} designed an upper bound called SEU. Based on the upper bound PEU, Gan et al. \cite{gan2020fast} proposed pruning strategies to quickly eliminate unpromising candidates, namely irrelevant item pruning (IIP) and lookahead pruning (LAR). More details of these advances can be found in the review literature \cite{gan2019survey}.

\subsection{High Average Utility Sequential Pattern Mining} \label{sec: hauspm}

Some preliminary studies have confirmed that existing methods and strategies for capturing HAUIM are not capable of handling sequential databases. However, several challenging issues are shared across different types of datasets, including traditional transaction databases and quantitative sequential databases. For example, in both HUIM and HUSPM, the utility of patterns fails to comply with the downward closure property. Moreover, in HAUIM and HAUSPM, unlike support or utility, the average utility exhibits neither anti-monotonic nor monotonic behavior, which makes the discovery processing more challenging.

Hong et al. \cite{hong2009mining, hong2011effective} proposed the first two-phase HAUI algorithms, TPAU, which introduces an upper bound, referred to as auub, based on utility overestimation to retain the downward closure property. Subsequent studies introduced tighter and more diverse upper bounds, such as transaction maximum utility in HAUI-Miner \cite{lin2018maintenance} and maximum average utility in MHAI \cite{yun2017mining}. EHAUPM \cite{lin2017ehaupm} proposed a revised tighter upper bound and a looser upper bound, referred to as rtub and lub, respectively. The top-\textit{k} revised transaction maximum utility upper bound (krtuub) and mfuub, focusing on maximum following utility, were introduced in TUB-HAUPM \cite{wu2018tub}. LMHAUP \cite{kim2021efficient} designed two tighter upper bounds: the tight maximum average utility upper bound and the maximum remaining average utility upper bound. EHAUSM \cite{truong2020ehausm} introduced a weak upper bound, twaub, along with two other upper bounds—\( \textit{\rm AMUB}_1 \) and BiUB—to identify HAUSPs in a quantitative sequential database, and incorporated four pruning strategies to enhance mining efficiency. FLCHUSPM \cite{truong2022mining} proposed a cost lower bound (FLB) and two utility upper bounds, AMUB and FUB, for the FLCHUSM. C-FHAUSPM \cite{tin2022frequent} employed three upper bounds, l\_aub, t\_aub, and AM\_aub (\( \textit{\rm AMUB}_1 \) from EHAUSM), and one weak upper bound t\_waub, to find frequent sequences with high minimum average utility and constraints. U-HPAUSM \cite{duong2025u} introduced a tighter upper bound (AMUBau) and a weak upper bound (TWUBau) to handle the task of finding the high average utility and high probability patterns in uncertain quantitative sequential databases.

In addition to the design of upper bounds, efforts have also been made to develop efficient data structures. TPAU \cite{hong2011effective} follows a level-wise approach. This hierarchical approach suffers from two key limitations: the necessity of multiple database scans and the excessive generation of candidate patterns. To overcome these limitations, PBAU \cite{lan2012projection} adopted a projection-based method by designing a tree structure and index tables. Building upon the projection technique and the prefix concept, an improved strategy called PAI was proposed \cite{lan2012efficiently}. HAUI-Growth utilized a HAUI-tree structure to maintain the average utility and avoid repeated database scans. Besides the aforementioned tree structure, MHAI \cite{yun2017mining}, HAUI-Miner \cite{lin2018maintenance}, and EHAUPM \cite{lin2017ehaupm} employed a list-based structure. EHAUPM introduced a MAU-list structure. FLCHUSPM \cite{truong2022mining} proposed a list of cost-utility (CUL) to efficiently store and update utility and cost information. C-FHAUSPM \cite{tin2022frequent} adopted a list of extended utility (EUL), which was originally introduced in EHAUSM \cite{alkan2015crom}, for discovering frequent sequences with high minimum average utility and constraints. Additionally, EHAUSM \cite{alkan2015crom} designed a list of sums of items (SL) to work alongside EUL. A similar utility-list structure, nUL, was employed in U-HPAUSM \cite{duong2025u}. Furthermore, a utility list of the diffset (IDUL) \cite{zaki2003fast} was developed for vertical database representation in VMHAUI \cite{truong2019efficient}.

\subsection{Targeted Pattern Mining} \label{sec: tpm}

Conventional pattern mining typically aims to discover all patterns that meet specified thresholds. However, this approach often yields an overwhelming number of results, many of which are not of interest to users. To address this issue, target pattern querying (TPQ) was proposed \cite{kubat2003itemset}, enabling users to focus the mining process on patterns that contain specific target substructures and facilitating targeted exploratory analysis \cite{fournier2022pattern}. Kubat et al. \cite{kubat2003itemset} introduced an optimized approach tailored for TPQ/TPM task in transaction databases. They implemented an incremental updating approach by leveraging a novel data structure called the itemset tree. To enhance the approach efficiency, they devised the memory-efficient itemset tree (MEIT) \cite{fournier2013meit}, which reduces memory consumption compared to the traditional structure IT. GFP-growth \cite{shabtay2021guided} was developed to compute the support of a larger list of itemsets. In the context of constraint-based target queries for sequence data, a solution was proposed to address specific analytical needs. For the target patterns defined at the end of sequences, a method was proposed to mine target sequential patterns that satisfy monetary and recency constraints \cite{chand2012target}. TargetUM \cite{miao2022targeted} utilizes a utility-based trie tree structure and introduces a utility-driven target-querying method tailored for quantitative transaction database mining. Additionally, TUSQ was designed to support target queries on sequence datasets, employing two novel upper bounds and the targeted utility chain to achieve targeted and efficient discovery of high-utility sequences. A general definition of targeted sequential pattern mining (TSPM) was provided, along with the introduction of an efficient algorithm, TaSPM \cite{huang2024taspm}. To facilitate the identification of abnormal behaviors and periodic patterns, TCSPM \cite{hu2024targeted} was developed for querying patterns with strict continuity, integrating the concept of targeted querying into the mining of contiguous sequential patterns. However, these approaches rely on total utility overestimation, and little attention has been paid to target queries under average utility semantics. This work aims to fill this gap by introducing a targeted high-average-utility pattern mining framework for quantitative sequence data.

\section{Preliminaries} \label{sec: preliminary}

This section outlines the notations and definitions used in this study to clearly characterize the research problem and proposed methodology. The remainder of this section shows some examples.

Let \( I \) = \{\( i_1 \), \( i_2 \), \( \cdots \), \( i_M \)\} be a set of distinct items, and let \( X \subseteq I \) represent a nonempty subset of these items, where \( |X| \) denotes the quantity of items in \( X \). A sequence \( S \) is defined as an ordered list of itemsets, where each itemset's items are sorted alphabetically. The size of \( S \) is the total quantity of itemsets it contains, while the length of \( S \) is the total count of individual items across all itemsets in this sequence. We refer to \( S \) as an \textit{l}-sequence if its length is l.

A sequence \( S \): \(\langle X_1 , X_2 , \cdots , X_n \rangle\) contains the subsequence \( {s'} \): \(\langle {X_v}', {{X_{v+1}}'}, \cdots , {X_m}' \rangle\), if there exist integers \( 1 \leqslant k_1 < k_2 < \cdots < k_m \leqslant n \) such that \( {X_v'} \subseteq X_{k_v},(1 \leqslant v \leqslant m) \), denoted by \( {s'} \subseteq S \). For example, consider the sequence \( s \) = \(\langle\{a\},\{a,b\},\{c,d,e\},\{f,g\}\rangle\), which consists of 4 itemsets or 7 distinct items. The size of \( s \) is 4, and its length is 8. The sequence \( {s'} \)=\(\langle {b}, {cd}, {f}\rangle\) is a subsequence of \( s \), or \( s \) contains the subsequence \( {s'} \), meaning \( {s'} \subseteq s \).

\begin{table}[ht]
    \caption{Quantitative sequential database} 
    \small
	\centering
	\begin{tabular}{cc}
	\toprule
	\textbf{SID} & \textbf{\textit{Q}-sequence} \\
	\midrule
	\( \textit{QS}_1 \) & \(\langle \{(b,4)(d,1)\},\{(b,2)(c,1)(d,4)\}, \{(a,1)(e,2)(i,1)\} \rangle\) \\
	\( \textit{QS}_2 \) & \(\langle \{(a,5)(c,2)(d,4)\}, \{(b,5)(c,1)(d,3)\}, \{(a,1)(e,2)\}, \{(f,4)\} \rangle\) \\
	\( \textit{QS}_3 \) & \(\langle \{(a,1)(b,1)(g,1)\}, \{(b,6)(c,4)(d,4)\}, \{(a,1)(i,3)\}, \{(a,1)(b,1)(d,4)(e,3)\} \rangle\) \\
	\( \textit{QS}_4 \) & \(\langle \{(c,1)(f,1)\}, \{(a,1)(c,5)(d,4)(e,1)\}, \{(b,1)(g,3)(i,1)\} \rangle\) \\
	\( \textit{QS}_5 \) & \(\langle \{(h,2)\}, \{(c,1)(d,3)(g,2)\}, \{(a,1)(e,1)(i,1)\} \rangle\) \\
	\bottomrule
	\end{tabular}
	\label{tb: database}
\end{table}

\begin{table}[!ht]
    \centering
    \caption{Utility table}
    \begin{tabular}{|l|l|l|l|l|l|l|l|l|l|}
    \hline
        \textbf{Item} & \( a \) & \( b \) & \( c \) & \( d \) & \( e \) & \( f \) & \( g \) & \( h \) & \( i \) \\ \hline
        \textbf{Profit} & 2 & 3 & 8 & 1 & 7 & 9 & 4 & 15 & 5 \\ \hline
    \end{tabular}
    \label{tb: utilityTable}
\end{table}

\begin{definition}[quantitative item, quantitative itemset, quantitative sequence, quantitative sequential database]
  \rm A quantitative sequential database consists of a quantitative sequence (\textit{q}-sequence) and the corresponding unique identifier (SID). Each quantitative sequence (\textit{q}-sequence) is an ordered list of the quantitative itemsets (\textit{q}-itemsets). In a certain quantitative sequential database, each distinct item \( i \) corresponds with its external utility \( \textit{eu}(i) \). The quantitative item (\textit{q}-item) in the \textit{q}-itemset is a pair \( (item,quantity) \), and the internal utility of each \textit{q}-item is its quantity, which is denoted as \( q(i,j,\textit{QS}) \), where \( i \) is the label of the item, and \( j \) is the numerical order of the quantitative itemset that contains this item in the quantitative sequence \( \textit{QS} \). 
\end{definition} 

In Table \ref{tb: database}, for instance, the \textit{q}-items \((a, 1)\) and \((e, 2)\) are ordered alphabetically in the last \textit{q}-itemset of the \textit{q}-sequence \( \textit{QS}_1 \), and we have \( q(a,3,\textit{QS}_1) \) = 1, \( q(e,3,\textit{QS}_1) \) = 2. The external utilities for items \( a \) and \( e \) are presented in Table \ref{tb: utilityTable}, where their values are 2 and 7, respectively.

\begin{definition}[utility of quantitative item, quantitative itemset, quantitative sequence]
  \rm Let \( \textit{QS} \): \( \langle Y_1, Y_2, \cdots, Y_n \rangle \) denote a \textit{q}-sequence, and \( Y_j \) is the \( j^{th} \) \textit{q}-itemset in \( \textit{QS} \). The \( (i,q) \) denotes one of the \textit{q}-items within \( Y_j \). The internal utility of the \textit(q)-item \( i \) is \( q(i,j,s) \) and its external utility is \( \textit{eu}(i) \). The utility of \textit{q}-item \( (i,q) \) is defined as \( u(i, j, \textit{QS}) \) = \( q(i,j,\textit{QS}) \times \textit{eu}(i) \). The utility of a \textit{q}-itemset is defined as the sum of \( q(i,j,\textit{QS}) \times \textit{eu}(i) \) for all \textit{q}-items \( i \) contained in it, denoted by \( u(Y_j, j, \textit{QS}) \) = \( \sum_{\forall (i,q) \in {Y_j}}{q(i,j,\textit{QS}) \times \textit{eu}(i)} \). The utility of the quantitative sequence \( \textit{QS} \) is defined as \( u(\textit{QS}) \) = \( \sum_{\forall {Y_j}\in \textit{QS}}{u(Y_j, j, \textit{QS})} \).
\end{definition} 

For example, the item \( a \), which is in the last \textit{q}-itemset of \( \textit{QS}_1 \) in Table \ref{tb: database}, has its utility calculated as: \( u(a, 3, \textit{QS}_1) \) = \( q(a,3,\textit{QS}_1) \times \textit{eu}(a) \) = 1 \(\times\) 2 = 2. Furthermore, \( u(\{ae\},3,\textit{QS}_1) \) = \( u(a, 3, \textit{QS}_1) + u(e, 3, \textit{QS}_1) \) = 2 + 14 = 16. As shown in Table \ref{tb: database}, we have \( u(\textit{QS}_1) \) = \( u(\{bd\}, 1, \textit{QS}_1) + u(\{bcd\},2,\textit{QS}_1) + u(\{aei\}, 3, \textit{QS}_1) \) = 13 + 18 + 21 = 52.

\begin{definition}[average utility of \textit{q}-item, \textit{q}-itemset, \textit{q}-sequence]
  \rm Let \( \textit{QS} \): \( \langle Y_1, Y_2, \cdots, Y_n \rangle \) denote a \textit{q}-sequence within the given quantitative sequential database, which we denote by \( \mathcal{D} \). Let \( (i,q) \) denote one of the \textit{q}-items in the \( j^{th} \) \textit{q}-itemset \( Y_j \) in \( \textit{QS} \). The size of \( Y_j \) is the entire count of \textit{q}-items in \( Y_j \), denoted as \( |Y_j| \). The size of \textit{QS} is \( |\textit{QS}| \) = \( n \). The length of \( \textit{QS} \) is \( |\textit{QS}| \) = \( \sum_{\forall Y_j\in \textit{QS}}{|Y_j|} \). The average utility of \textit{q}-itemset \( Y_j \) is defined as \( au(Y_j,j,\textit{QS}) \) = \( \frac {u(Y_j,j,\textit{QS})}{|Y_j|} \). The average utility of \textit{q}-item \( (i,q) \) is defined as \( au(i,j,\textit{QS}) \) = \( u(i,j,\textit{QS}) \). The average utility of \textit{q}-sequence \( \textit{QS} \) is \( au(\textit{QS}) \) = \( \frac {u(\textit{QS})}{|\textit{QS}|} \).
\end{definition} 

For example, the average utility of the last \textit{q}-itemset of \( \textit{QS}_1 \) in Table \ref{tb: database} is calculated as: \( au(\{ae\}, 3, \textit{QS}_1) \) = \( \frac {u(\{ae\}, 3, \textit{QS}_1)}{|\{ae\}|} \) = \( \frac{16}{2} \) = 8. Moreover, we have \( au(\textit{QS}_1) \) = \( \frac{52}{8} \) = 6.5.

\begin{definition}[match and contain] 
    \label{def: match&contain}
  \rm We say that the itemset \( X \): \( \{i_1, i_2, \cdots, i_m\} \) \textit{matches} the \textit{q}-itemset \( Y \): \{\( ({i'}_1, q_1) \), \( ({i'}_2, q_2) \), \( \cdots \), \( ({i'}_n, q_n) \)\} if and only if \( m \) = \( n \) such that \( i_k \) = \( {i'}_k, (1 \leqslant k \leqslant n) \). It could be notated as \( X \sim Y \). Let \( {X'} \) denote a subset of \( X \). We could say that \( Y \) \textit{contains} \( {X'} \), it is notated as \( {X'} \sqsubseteq Y \).
\end{definition}

\begin{definition}[instance] 
    \label{def: instance}
  \rm Consider the \textit{q}-sequence \( \textit{QS} \): $\langle$$Y_1$, $Y_2$, $\cdots$, $Y_n$$\rangle$ and the sequence \( S \): \( \langle X_1, X_2, \cdots, X_m \rangle \), where \( m \leqslant n \). Assume that there exists integer \( j_v \), if and only if 1 \( \leqslant j_1 < j_2 < \cdots < j_m \leqslant n \) and \( X_v \sqsubseteq Y_{j_v} \), where \(1 \leqslant v \leqslant m \). We say that in \( \textit{QS} \), there is an \textit{instance} of \( S \) at position \( p \): \( \langle j_1, j_2, \cdots, j_m \rangle \). Then, the sum of all \textit{q}-items utilities is the \textit{instance utility}. It is defined as \( u(S, p, \textit{QS}) \) = \( \sum_{\forall {Y_{j_v}}\in \textit{QS}}{u(Y_{j_v}, j_v, \textit{QS})} \). The \textit{instance average utility} is defined as \( au(S, p, \textit{QS}) \) = \( \frac{u(S, p, \textit{QS})}{|S|} \).
\end{definition}

For example, \( \langle \{(a,1)(e,2)\} \rangle \) \textit{contains} \( \{e\} \), and \{\( bd \)\} has two \textit{matches}, $\langle$\{$(b,4)$ $(d,1)$\} $\rangle$ and $\langle$ \{$(b,2)$ $(d,4)$\} $\rangle$, in \( \textit{QS}_1 \). The \textit{q}-sequences \( \langle \{(b,4)(d,1)\}, \{(e,2)\} \rangle \) and \( \langle \{(b,2)(d,4)\}, \{(e,2)\} \rangle \) are two \textit{instances} of \( \langle \{bd\}, \{e\} \rangle \) in \( \textit{QS}_1 \). Moreover, a \textit{k}-itemset (also referred to as a \textit{k}-\textit{q}-itemset) is defined as an itemset with a cardinality of exactly \( k \) items. Similarly, a \textit{k}-sequence (or \textit{k}-\textit{q}-sequence) denotes a sequence comprising precisely \( k \) items. For example, in Table \ref{tb: database}, the \textit{q}-sequence \( \textit{QS}_1 \) is a \textit{8}-\textit{q}-sequence, and its last \textit{q}-itemset is a \textit{3}-\textit{q}-itemset.

\begin{definition}[sequence average utility] 
    \label{def: SAU}
  \rm If the sequence \( S \): \( \langle X_1, X_2, \cdots, X_m \rangle \) appears at different positions in the \textit{q}-sequence \( \textit{QS} \): \( \langle Y_1, Y_2, \cdots, Y_n \rangle \). Let \( P(S, \textit{QS}) \) denote the set of all the positions of \( S \) in \( \textit{QS} \), the utility of the sequence \( S \) in \( \textit{QS} \) is the maximum \( u(S, p, \textit{QS}) \), and is denoted as \( u(S, \textit{QS}) \) = \( \max\limits_{p \in P(S, QS)}{u(S, p, \textit{QS})} \). The average utility of the sequence \( S \) in \( \textit{QS} \) is defined as \( au(S, \textit{QS}) \) = \( \frac{\max\limits_{p\in P(S, \textit{QS})}{u(S, p, \textit{QS})}}{|S|} \) = \( \max\limits_{p \in P(S, \textit{QS})}{\frac{u(S, p, \textit{QS})}{|S|}} \).
\end{definition} 

For example, in Table \ref{tb: database}, the utility of \( \langle \{bd\}, \{e\} \rangle \) in \( \textit{QS}_1 \) is determined by taking the maximum value from the utilities of its two \textit{instances} at different positions. That is , \( u(\langle \{bd\}, \{e\} \rangle,\textit{QS}_1) \) = \( \max{\{u(\langle \{(b,4)(d,1)\}, \{(e,2)\} \rangle, \textit{QS}_1), u(\langle \{(b,2)(d,4)\}, \{(e,2)\} \rangle, \textit{QS}_1)\}} \) = \( \max{\{27, 24\}} \) =27.

\textbf{Problem definition:} Given a query sequence \( T \) and a quantitative sequential database \( \mathcal{D} \), let \( \mathcal{D_T} \) denote the filtered database consisting of all sequences from \( \mathcal{D} \) that contain \( T \) as a subsequence. The total utility of the filtered database is denoted as \( u(\mathcal{D_T}) \). Let \( \xi \) be a user-specified parameter where \( 0 \leqslant \xi \leqslant 1 \). The minimum acceptable average utility is thus defined as \( \xi \times u(\mathcal{D_T}) \). Based on this, the targeted high average utility sequence querying (TAUSQ) or targeted high average utility sequential pattern mining (TAUSPM) problem is defined as the task of finding all targeted sequential patterns (TSPs) in the original database \( \mathcal{D} \) that both contain the query sequence \( T \) and have an average utility greater than the threshold \( \xi \times u(\mathcal{D_T}) \).

For example, in Table \ref{tb: database}, all sequences except \( \textit{QS}_4 \) contain the given query sequential pattern \( \langle \{d\}, \{e\} \rangle \). Therefore, the sequence \( \textit{QS}_4 \) is filtered out, resulting in a filtered database \( D_T \) = \{\( \textit{QS}_1, \textit{QS}_2, \textit{QS}_3, \textit{QS}_5 \)\}, with a total utility of \( u(D_T) \) = 333. If \( \xi \) = 0.1, then the minimum acceptable average utility becomes \( \xi \times u(\mathcal{D_T}) \) = 33.3. The sequence \( \langle \{cd\}, \{e\} \rangle \) is a targeted high average utility sequential pattern (TAUSP) since its average utility is \( au(\langle \{cd\}, \{e\} \rangle) \) = \( \frac{135}{3} \) = 45, which exceeds the threshold of 33.3. In summary, the formal problem studied in this paper is defined as follows: Given a quantitative sequential database, a query sequence, and a user-defined minimum average utility threshold, the task of TAUSPM is to enumerate all TAUSPs that contain the query sequence and whose average utility within the filtered database is greater than or equal to the specified threshold.

\section{Algorithm} \label{sec: algorithm}

In SPM, a typical approach begins by constructing a reasonable and compact projection database. To avoid multiple scanning and a combinatorial explosion, we adopt the classical pattern growth method \cite{pel2001prefixspan}. Moreover, various novel upper bounds and their variants are designed to effectively reduce the search space and enhance mining efficiency. The following sections provide a detailed description of the proposed algorithm.

\subsection{Pruning Strategies and Upper Bound Models} \label{sec: pS&ubs}

The proposed algorithm first identifies all the \textit{1}-sequences, whose average utility is equal to their utility, as the starting point. From these, candidate patterns are progressively extended. During this process, some pruning strategies and efficient data structures are employed to eliminate unpromising candidates and improve computational performance.

\begin{definition}[\textit{S}-Extension and \textit{I}-Extension \cite{yin2012uspan, pel2001prefixspan}] \label{def: extension}
  \rm Consider the last itemset \( X_k \): \( \{i_1,i_2,\cdots ,i_l\} \) in the sequence \( S \): \( \langle X_1,X_2,\cdots ,X_k\rangle \). Let \( X_k \) be the position for the \textit{extension} operation. For an appending item \( i \), if \( i \) is appended to \( X_k \) as \( i_{l + 1} \), the length of the sequence is increased by one, but the size of \( S \) remains static. However, if \( i \) is appended to \( S \) as \( X_{k + 1} \), both the size of \( S \) and the length of the sequence are increased by one. The former case is defined as \textit{I}-Extension and is notated as \( S \oplus i \). The latter case is defined as \textit{S}-Extension and is notated \( S \otimes i \).
\end{definition} 

For example, consider the sequence \( \textit{QS}_3 \) in Table \ref{tb: database}. Given a sequence \( s \) = \( \langle \{cd\} \rangle \) and a new appending item \( e \), the results of appending \( e \) are as follows: \( S \oplus i \) = \( \langle \{cde\} \rangle \), and \( S \otimes i \) = \( \langle \{cd\},\{e\} \rangle \). The newly generated sequences, after the \textit{extension} process, are treated as candidate patterns and form child nodes under the current node, \( \langle \{cd\} \rangle \), in the LQS tree. This process is analogous to the search procedure described in Ref. \cite{gan2020fast}. To guarantee the completeness and correctness of discovering TAUSPs, we also define an order for processing sequences based on the conditions outlined in Ref. \cite{gan2020fast}. For example, in the following cases, sequences on the left are always processed first. Consider the sequences \( \langle \{b\} \rangle \) and \( \langle \{bd\} \rangle \), where the sequence \( \langle \{bd\} \rangle \) is processed after \( \langle \{b\} \rangle \) due to its longer length. Similarly, for the sequences \( \langle \{bd\} \rangle \) and \( \langle \{b\},\{d\} \rangle \), the left sequence is obtained by performing an \textit{I}-Extension on \( \langle \{b\} \rangle \), while the right sequence results from an \textit{S}-Extension on \( \langle \{b\} \rangle \). Lastly, when comparing sequences such as \( \langle \{bc\} \rangle \) and \( \langle \{bd\} \rangle \), or \( \langle \{b\},\{c\} \rangle \) and \( \langle \{b\},\{d\} \rangle \), the left sequence is processed first because the item added to the left sequence in either a \textit{S}-Extension or an \textit{I}-Extension operation is lexicographically smaller than the item added to the right sequence.

\begin{definition}[extension item \cite{wang2016efficiently} and remaining \textit{q}-sequence \cite{wang2016efficiently, yin2012uspan}] \label{def: EI&RS}
  \rm Consider the \textit{instances} of \( S \): \( \langle X_1, X_2, \cdots, X_m \rangle \) and a \textit{q}-sequence \( \textit{QS} \): \( \langle Y_1, Y_2, \cdots, Y_n \rangle \), the \textit{instances} generally appear at several positions in \( \textit{QS} \). The set of positions is notated as \( P(S,\textit{QS}) \): \{\( p_1, p_2, \cdots, p_w \)\}. Let \( p_k \): \( \langle j_1, j_2, \cdots, j_m \rangle \) be one of positions, the \textit{extension position} \( j_m \) is the sequence number of \textit{q}-itemset in \( \textit{QS} \) which \textit{contains} \( X_m \). The \textit{extension item} is the \textit{q}-item which corresponds to the last item within \( X_m \). All the items that are behind the \textit{extension item} form a subsequence. We define the subsequence as the \textit{remaining sequence} of \( \textit{QS} \), designated as \( \textit{rs} \). The utility of \( \textit{rs} \) is notated as \( \textit{ru}(S, j_m, \textit{QS}) \).
\end{definition} 

\begin{definition}[longest query prefix and query suffix \cite{zhang2022tusq}] 
    \label{def: qPre &qSuf}
  \rm Consider a sequence \( S \): \( \langle X_1, X_2, \cdots, X_m \rangle \) is a prefix sequence in pattern growth. Let \( t \) be a prefix of the query sequence \( T \): \( \langle x_1, x_2, \cdots, x_v \rangle \) and \( q \) be an \textit{instance} of \( S \), then we have the \( q \) \textit{contains} \( t \). If and only if there exists no other subsequence of \( T \) which has \textit{instance} in \( q \) and whose length exceeds the length of \( t \), then \( t \) is defined as the longest query prefix of the query sequence \( T \), denoted as \( \textit{qPre}(T, S) \). The remaining part of \( T \), after removing \( \textit{qPre}(T, S) \), is referred to as the query suffix of the query sequence \( T \), denoted \( \textit{qSuf}(T, S) \).
\end{definition}

\begin{definition}[post-processing technique \cite{huang2024taspm} and pre-processing technique] \label{def: pp}
  \rm Since the introduction of target sequential pattern mining, most algorithms adopt two fundamental processing methods: preprocessing and postprocessing \cite{huang2024taspm}. In the preprocessing phase, an initial scan of the original database is carried out to determine whether a sequence includes the query sequence. Any sequence in the initial dataset that lacks the query sequence is filtered out. During the data preprocessing stage, after generating candidate patterns, each pattern is checked to verify if it includes the given query sequence. Patterns containing the query sequence are retained, while those that do not are discarded. When the pattern growth method is used to generate new patterns, both techniques significantly influence the efficiency. The combination of these methods helps control the search space and improve efficiency.
\end{definition} 

\begin{strategy}[sequence filter pruning strategy] \label{str: STR_1}
  \rm Given a specified query sequence \( T \), if the current sequence records in the database do not contain \( T \), filtering of the current sequence records is necessary. It is evident that sequence records not containing the query sequence will not generate targeted sequential patterns. By eliminating these irrelevant sequence records, memory consumption is reduced, thus enhancing efficiency. Moreover, for high average utility pattern mining, if the utility of the filtered database lies below a predefined utility threshold, targeted high average utility sequential patterns cannot be discovered from this database. Note that the specified minimum utility threshold is expressed as \( |T| \times \xi \times u(\mathcal{D_T}) \), where \( |T| \) is the length of the given target pattern \( T \).
\end{strategy}

Consider the database in Table \ref{tb: database} together with a query pattern \( T \) = \( \langle \{cd\}, \{e\} \rangle \). To achieve a more concise filtered database, it is clear that sequence \( \textit{QS}_4 \) should be filtered out. Without this filtering strategy, if \( \xi \) = 0.2 and the current sequence is \( \langle \{c\} \rangle \), the utility of \( \langle \{c\} \rangle \) would be 104, which exceeds the threshold \( \xi \times u(\mathcal{D}) \) = \( 0.2 \times 423\) = 84.6, potentially leading to further recursive growth. As a result, invalid operations would accumulate because the utility of \( \langle \{c\} \rangle \) in the target sequential pattern can reach at most 64, which does not exceed the threshold \( \xi \times u(\mathcal{D_T}) \) = \( 0.2 \times 333\) = 66.6. As another instance, consider the query sequence \( T \) = \( \langle \{d\},\{bcd\},\{ai\} \rangle \), with \( \xi \) also set to 0.2. After filtering the database \( \mathcal{D} \), the resulting filtered database \( \mathcal{D_T} \) will contain only the sequence \( \textit{QS}_1 \). Under this scenario, \( \mathcal{D_T} \) exhibits a utility of 52, falling below the minimum acceptable utility threshold calculated as \( \xi \times u(\mathcal{D_T}) \times |T| \) = \( 0.2 \times 52 \times 6 \) = 62.4. Therefore, no target sequences with high average utility will be discovered.

\begin{strategy}[prefix pattern pruning strategy] \label{str: STR_2}
  \rm In SPM algorithms, the pattern growth method \cite{pel2001prefixspan} is widely used. This approach generates new patterns by extending existing patterns through appending a new item to the tail of the prefix. Consider that, given a prefix, it is possible to filter the remaining part of the sequence accordingly. In targeted pattern mining, however, it calculates that the remaining part of the sequence contains the query suffix of the query sequence. Given the absence of the query suffix in the remaining part of the sequence corresponding to the given prefix, the pattern growth approach cannot be applied to generate patterns that contain the full query sequence. In line with Strategy 1, the current sequence should be excluded from further consideration. Moreover, if, based on the current prefix, the total utility of the filtered database falls below the specified utility threshold, it becomes impossible to discover any targeted high average utility sequential patterns.
\end{strategy} 

As exemplified in Table \ref{tb: database}, consider the query sequence \( T \) = \( \langle \{cd\} \rangle \) and the \( \xi \) set at 0.2; it is clear that the item \( c \) qualifies as a frequent item since its utility value of 104 surpasses the acceptable minimum utility threshold. Next, we compare the utility values by calculating the cumulative utility of sequences in the filtered database that include the given prefix. For this query sequence \( T \), the filter database encompasses all the sequences presented in Table \ref{tb: database}. When considering the prefix item \( h \), the filtered database is narrowed down to only \( \textit{QS}_5 \). It is evident that, under this strategy, the utility of the filtered database with the specified prefix is 63. Therefore, the item \( h \) will not be considered for further pattern extension.

\begin{strategy}[unpromising \textit{S}-Extension item pruning strategy] \label{str: STR_3}
  \rm This strategy is used to identify which sequential patterns can undergo recursive growth following an \textit{S}-Extension operation. It compares the current extension item with the current query item. If the current query item matches the current extension item, the longest query prefix and query suffix are updated accordingly. The remaining utility of all sequences with the current prefix \( s \) with the corresponding sequence's remaining portion containing the query suffix, is designated as \( \textit{ru}_\textit{suf}(s) \). Let \( u \) represent the utility of the prefix sequence. The current extension item cannot be used as an extension item for the pattern growth process, when the sum of \( u(s) + \textit{ru}_\textit{suf}(s) \) is below the minimum utility threshold.
\end{strategy} 

\begin{strategy}[unpromising \textit{I}-Extension item pruning strategy] \label{str: STR_4}
  \rm This strategy is used to determine which sequential patterns can continue recursive growth methods after undergoing \textit{I}-Extension operations. It compares the current extension item with the current query item. If the current query item matches the current extension item, it updates the longest query prefix and query suffix. The remaining utility of all sequences containing the current prefix \( s \) is computed by considering the rest of the sequence that contains the query suffix, denoted as \( \textit{ru}_\textit{suf}(s) \). Let \( u \) represent the prefix sequence utility. The current extension item cannot be used as an extension item for pattern growth, when the value of \( u(s) + \textit{ru}_\textit{suf}(s) \) falls below the prespecified minimum utility threshold. Note that for \textit{I}-Extension expansions, if the current query item appears before the current extension item, the longest query prefix and query suffix need to be reset.
\end{strategy}

For instance, referring to Table \ref{tb: database}, consider a current sequence \( s \) = \( \langle \{a\} \rangle \), the query sequence \( T \) = \( \langle \{cd\}, \{e\} \rangle \), and \( \xi \) = 0.1. It is evident that the item \( c \) can be extended through \textit{S}-Extension, and the resulting extended sequence \( {s'} \) is \( \langle \{a\}, \{c\} \rangle \). The query suffix for this extension is \( \textit{qSuf}(T, {s'}) \) = \( \langle \{d\}, \{e\} \rangle \). The utility of the prefix is given by \( u(s') \) = \( u(\langle \{a\}, \{c\} \rangle,\textit{QS}_2) + u(\langle \{a\}, \{c\} \rangle,\textit{QS}_3) \) = 18 + 34 = 52. The remaining utility is \( \textit{ru}_\textit{suf}(s') \) = \( \textit{ru}_\textit{suf}(\langle \{a\}, \{c\} \rangle,\textit{QS}_2) + \textit{ru}_\textit{suf}(\langle \{a\}, \{c\} \rangle,\textit{QS}_3) \) = 55 + 51 = 106. Therefore, we have \( u(s') + \textit{ru}_\textit{suf}(s') \) = 52 + 106 = 158, which exceeds the threshold \( \xi \times u(\mathcal{D_T}) \times (|{s'}| + |\textit{qSuf}|) \) =  0.1 $\times$ 333 $\times$ 4  = 133.2. Thus, the extended sequence \( {s'} \) can continue to grow a target sequential pattern with high average utility.

Similarly, in the case in Table \ref{tb: database}, for the current sequence \( s \) = \( \langle \{a\} \rangle \), the query sequence \( T \) = \( \langle \{cd\}, \{e\} \rangle \), and \( \xi \) = 0.1, the item \( c \) can also be extended through \textit{I}-Extension, forming the sequence \( {s'} \) = \( \langle \{ac\} \rangle \). The corresponding query suffix is \( \textit{qSuf}(T, {s'}) \) = \( \langle \{d\}, \{e\} \rangle \). In this case, the utility of the prefix is \( u(s') \) = \( u(\langle \{ac\} \rangle, \textit{QS}_2) \) = 26, and the remaining utility is \( \textit{ru}_\textit{suf}(s') \) = \( \textit{ru}_\textit{suf}(\langle \{ac\} \rangle,\textit{QS}_2) \) = 82. The total utility \( u(s') + \textit{ru}_\textit{suf}(s') \) = 26 + 82 = 108, which is below the threshold $\xi \times u(\mathcal{D_T})$ $\times$ $(|{s'}| + |\textit{qSuf}|)$ =  0.1 $\times$ 333 $\times$ 4 = 133.2. Therefore, the extended sequence \( {s'} \) cannot be further extended.

\begin{definition}[item match position (\textit{IIMatch}) and itemset match position (\textit{IMatch}) \cite{huang2024taspm}] 
\label{def: IIM&IM}
  \rm Throughout the entire pattern growth process, the item currently being extended is referred to as the current query item, and the itemset containing this item is referred to as the current query itemset. In TPM, it is crucial to track the matching status between the generated pattern and the query sequence, as this not only determines the pattern prefix but also plays a key role in the pruning strategy for suffix judgments. To record the match positions effectively, two flags are introduced: \textit{IMatch} and \textit{IIMatch}. The \textit{IMatch} flag stores the position of the current query itemset, while the \textit{IIMatch} flag stores the position of the current query item. These flags help monitor the progress of query matching. Once a generated pattern fully matches the query sequence, both flags are no longer updated.
\end{definition} 

These flags are initialized to 0. During pattern growth, if the current extension item matches the current query item, \textit{IIMatch} is updated to 1. Continuing to expand within the current itemset, each occurrence of a matched item increments \textit{IIMatch} by 1. Once the current query itemset is fully expanded, further extensions within the same itemset no longer update \textit{IIMatch}. The \textit{IMatch} remains unchanged unless \textit{IIMatch} equals the size of the current query itemset, at which point \textit{IMatch} is incremented by 1, and \textit{IIMatch} is reset to 0. If the extension position changes, meaning the current extension itemset changes, then \textit{IIMatch} is reset to 0. This mechanism allows efficient tracking of query matches without maintaining costly arrays or structures, which is particularly beneficial for long query sequences.

For example, suppose the query sequence is \( \langle \{cd\},\{ae\} \rangle \) and the current sequence is \( \langle \{ab\},\{c\} \rangle \). Upon an item \( d \) is extended to \( \langle \{ab\},\{c\} \rangle \) via an \textit{I}-Extension, the sequence transforms into \( \langle \{ab\},\{cd\} \rangle \), allowing for further \textit{I}-Extensions. Since the appending item \( d \) is the next \textit{extension item} in the query, and the \textit{extension position} is within the itemset containing \( c \), this operation updates the \textit{IIMatch} from 1 to 2. Once all items within the current itemset of \( \langle \{cd\},\{ae\} \rangle \) appear in the pattern, the \textit{IIMatch} is reset to 0, and \textit{IMatch} is updated from 0 to 1. Next, we proceed with another extension. If item \( a \) is extended, which is also the next part of the query and positioned in a different itemset from \( c \), this extension is performed via an \textit{S}-Extension. Consequently, the sequence becomes \( \langle \{ab\},\{cd\},\{a\} \rangle \) and \textit{IIMatch} is updated from 0 to 1. A special case arises when \textit{IMatch} reaches the size of the query sequence, indicating a complete match between a subsequence of the pattern and the query sequence. At this stage, further updates to \textit{IIMatch} and \textit{IMatch} are unnecessary.

From strategies \ref{str: STR_3} and \ref{str: STR_4}, it can be observed that when using a pattern growth approach to estimate the UBs of pattern utility for a query sequence, it is necessary to repeatedly verify whether the \textit{remaining sequence} \textit{contains} the corresponding \textit{qSuf}. To improve efficiency in this process, a novel data structure called the \textit{LI}-Table was introduced in Ref. \cite{zhang2022tusq}. Specifically, it stores the position of the final \textit{instance} of each itemset in \( T \) within the current \( \textit{QS} \). This transforms the complex problem of sequence matching into a simple numerical comparison, enabling what we call the position comparison method for more efficient sequence evaluation.

For example, consider a \textit{q}-sequence \( \textit{QS} \): \( \langle Y_1, Y_2, \cdots, Y_n \rangle \) of size \( n \), which contains the query sequence \( T \): \( \langle x_1, x_2, \cdots, x_v \rangle \). Starting from the last itemset of the sequence, \( Y_n \), we traverse the \textit{q}-sequence in reverse order and check whether each current itemset is the last itemset \( x_v \) of the query sequence \( T \). The \textit{LI}-Table documents the first match position. Next, we continue the search for the second-to-last itemset of \( T \), \( x_{v-1} \), within \( \textit{QS} \). Importantly, the search for each preceding itemset in \( T \) does not restart from the end of \( \textit{QS} \), but rather from the position previously found. This process continues until the positions of all itemsets in \( T \) have been recorded in the \textit{LI}-Table. By comparing the current \textit{extension position} with the position of the last \textit{instance} in the \textit{LI}-Table, we can efficiently determine whether \textit{qSuf} is present in the \textit{remaining sequence}, thus avoiding frequent scans of \( \textit{QS} \). If the current \textit{extension position} precedes the recorded position, the extension is promising. Otherwise, it is unpromising, as the \textit{remaining sequence} can not \textit{contain} \textit{qSuf}. This method offers faster querying than traditional subsequence checking.

In the context of HUSPM, pruning strategies are commonly combined with utility upper bounds to narrow the search space and boost efficiency. In some HAUPM algorithms, such as those in \cite{hong2009mining, hong2011effective}, the auub model was employed to estimate the sequence average utility \cite{hong2009mining, hong2011effective}. In these algorithms, high-utility items in transactions are used to replace the average utility of patterns. However, these approaches often perform poorly on datasets with uneven distributions in utility. To tackle this challenge, Lin et al. \cite{lin2017ehaupm} introduced the looser upper bound utility (lub) for discovering high average utility itemsets (HAUIs). The lub assumes that the utility of itemsets that may be extended in the \textit{remaining sequence} is equivalent to \( \textit{remu} \), the maximum utility of any item within the \textit{remaining sequence}. Its anti-monotonicity was formally proven in \cite{lin2017ehaupm}. Later, EHAUSM \cite{truong2020ehausm} extended the discussion on upper bounds by proposing the use of BiUB or \( \textit{\rm AMUB}_1 \) as the tighter bounds, depending on the scenario, offering more flexibility.

However, several challenges remain. Specifically, the utility estimation of the itemsets with the maximum utility is complicated by the fact that the \textit{remaining sequence} continually changes during the recursive mining process. In many cases, it is difficult to determine whether the current appending item is the one with maximal utility in the \textit{remaining sequence}, or to quickly identify the rank of any specific item in the \textit{remaining sequence} when it is sorted in descending order of utility. This requires multiple scans of the updated \textit{remaining sequence} to find the item with the maximum utility. Although using itemsets with higher utility based on utility ranking introduces less bias compared to estimating utility using the item with the maximum utility, the ranking process is both time-consuming and memory-intensive. These performance costs become especially problematic on datasets with certain data distributions. 

The method proposed in this paper involves two main components: filtering the original items and processing the filtered items. Notably, during the pattern growth process, it is unnecessary to determine the utility-based order of items in the \textit{remaining sequence}. The item with the maximum utility and the total utility of the sequence are also not required. For the TAUSPM, the process begins by checking whether the rest of the sequence containing the current prefix includes the query suffix corresponding to the longest query prefix. Then, for any \textit{extension item}, if the item utility is below the average utility of the prefix, the average utility of the generated pattern cannot increase. Moreover, any item in the \textit{remaining sequence} with utility inferior to the user-specified minimum acceptable average utility cannot help generate patterns with higher average utility from previous patterns with non-high average utility. To guide this evaluation, we propose a new measure that estimates the ability of the \textit{remaining sequence} to increase the average utility of the generated pattern. This evaluation metric computes the maximum additional utility increment provided by items in the \textit{remaining sequence} that meet the prespecified average utility threshold and support the growth of the average utility of the generated pattern.

\begin{definition}[remaining rising sequence] \label{def: RRS}
    \rm Consider a \textit{q}-sequence \( \textit{QS} \) \textit{containing} the query sequence \( T \), at the \textit{extension position} \( j_m \) within \( \textit{QS} \), sequence\( S \) has an \textit{instance}. The \textit{remaining sequence} is the rest after position \( p \): \( \langle j, j_2, \cdots, j_m \rangle \) to the end, denoted as \( \textit{rs} \), and \( \textit{qSuf}(T, S) \sqsubseteq \textit{rs} \). Its subsequence, consisting of items with utility values at least a predefined minimum threshold, is the \textit{remaining rising sequence} for this threshold, denoted as \( \textit{rrs} \). The utilities of \( \textit{rs} \) and \( \textit{rrs} \) are denoted as \( \textit{ru}(S, j_m, \textit{QS}) \) and \( u_\textit{rrs}(S, T, j_m, \textit{QS}) \), respectively.
\end{definition}

\begin{definition}[suffix remaining average utility] \label{def: SRAU}
    \rm Consider a \textit{q}-sequence \( \textit{QS} \) and query sequence \( T \), at position \( p \): \( \langle j, j_2, \cdots, j_m \rangle \), sequence \( S \) has an \textit{instance}. Its \textit{extension position} is \( j_m \), and the corresponding \textit{remaining sequence} is \( \textit{rs} \) that spans from \( p \) to the end, and \( \textit{qSuf}(T, S) \sqsubseteq rs \). Let \( \textit{rrs} \) be a subsequence in \( \textit{rs} \) containing only items with utility exceeding \( \xi \times u(\mathcal{D_T}) \). The suffix remaining average utility of the sequence \( S \) at position \( p \) for \( T \), denoted \( \textit{SRAU}(S, T, p, \textit{QS}) \), is formulated as
\begin{displaymath}
\begin{split}
\textit{SRAU}(S, T, p, \textit{QS}) 
= \left\{\begin{array}{rcl} \frac{u(S, T, p, \textit{QS}) + u_\textit{rrs}(S, j_m, \textit{QS})}{|S|}, & \textit{rs} \neq \emptyset \wedge \textit{qSuf}(T, S) \sqsubseteq \textit{rs} \\ 0, & otherwise \end{array}\right. \nonumber
\end{split}
\end{displaymath}

Let \( p_i \) denote a specific position of \( S \) with respect to \( T \) in \( \textit{QS} \). Then, we define \( \textit{SRAU}(S, T, \textit{QS}) \) = \( \max{\{\textit{SRAU}(S, T, p_i, \textit{QS})\}} \) as the \( \textit{SRAU} \) of \( S \) with respect to \( T \) in the \textit{q}-sequence \( \textit{QS} \). Finally, the \( \textit{SRAU} \) value of \( S \) with respect to \( T \) in the database \( \mathcal{D} \), denoted as \( \textit{SRAU}(S, T) \) = \( \sum_{S \sqsubseteq \textit{QS} \wedge \textit{QS} \in \mathcal{D_T}}{\textit{SRAU}(S, T, \textit{QS})} \), is defined as the upper bound of average utility.
\end{definition} 

For example, referring to Table \ref{tb: database}, consider a current pattern \( s \) = \( \langle \{b\} \rangle \), a query sequence \( T \) = \( \langle \{cd\}, \{e\} \rangle \), and \( \xi \) = 0.15. It is clear that the item \( c \) can be extended through \textit{S}-Extension or \textit{I}-Extension, resulting in the extended sequences \( {s'}_1 \) = \( \langle \{bc\} \rangle \) and \( {s'}_2 \) = \( \langle \{b\}, \{c\} \rangle \), respectively. For \( {s'}_1 \), we have \( u({s'}_1) + \textit{ru}_\textit{suf}({s'}_1) \) = \( u({s'}_1, \textit{QS}_1) + u({s'}_1, \textit{QS}_2) + u({s'}_1, \textit{QS}_3) + \textit{ru}({s'}_1, \textit{QS}_1) + \textit{ru}({s'}_1, \textit{QS}_2) + \textit{ru}({s'}_1, \textit{QS}_3) \) = 14 + 23 + 50 + 25 + 55 + 51 = 218. Similarly, for \( {s'}_2 \), we have \( u({s'}_2) + \textit{ru}_\textit{suf}({s'}_2) \) = \( u({s'}_2, \textit{QS}_1) + u({s'}_2,\textit{QS}_3) + \textit{ru}({s'}_2, \textit{QS}_1) + \textit{ru}({s'}_2, \textit{QS}_3) \) = 20 + 35 + 25 + 51 = 131. Both extended sequences appear to satisfy the threshold \( \xi \times u(\mathcal{D_T}) \times |{s'}_1| \) = \( \xi \times u(\mathcal{D_T}) \times |{s'}_2| \) = \( 0.15 \times 333 \times 2 \) = 99.9. However, when using the evaluation metric \( \textit{rrs} \) to measure the \textit{remaining utility}, the results change. For \( {s'}_1 \), we have \( u({s'}_1) + \textit{rrs}_\textit{suf}({s'}_1) \) = \( u({s'}_1, \textit{QS}_1) + u({s'}_1, \textit{QS}_2) + u({s'}_1, \textit{QS}_3) + u_\textit{rrs}({s'}_1, \textit{QS}_1) + u_\textit{rrs}({s'}_1, \textit{QS}_2) + u_\textit{rrs}({s'}_1, \textit{QS}_3) \) = 14 + 23 + 50 + 14 + 50 + 21 = 172. For \( {s'}_2 \), \( u({s'}_2) + \textit{rrs}_\textit{suf}({s'}_2) \) = \( u({s'}_2, \textit{QS}_1) + u({s'}_2, \textit{QS}_3) + u_\textit{rrs}({s'}_2, \textit{QS}_1) + u_\textit{rrs}({s'}_2, \textit{QS}_3) \) = 20 + 35 + 14 + 21 = 90. The utility value of the latter falls below the threshold of 99.9. Consequently, only \( {s'}_1 \) can grow into a valid target sequential pattern with high average utility.

\begin{theorem} \label{theo: THEO_1}
  \rm Consider the query sequence \( T \), a sequence \( S \) \( \neq \) \( \langle \rangle \) and its extension \( {S'} \) in database \( \mathcal{D} \). Both sequences satisfy the conditions \( \textit{qSuf}(T, S) \sqsubseteq \textit{rs}(S, T) \) and \( \textit{qSuf}(T, S') \sqsubseteq \textit{rs}(S', T) \). If \( \textit{SRAU}(S, T) \leqslant \xi \times u(\mathcal{D_T}) \) then \( au({S'}, T) \leqslant \xi \times u(\mathcal{D_T}) \).
\end{theorem}

\begin{proof} \label{proof: PROOF_1}
  \rm Assume that a sequence \( {S'} \) can be extended from its prefix \( S \) with an extension sequence \( s \). \( s \) is a subsequence of the \textit{remaining sequence} \( \textit{rs} \), and \( \textit{qSuf}(T, S) \sqsubseteq \textit{rs} \). Let \( \textit{exu} \) represent the excess part of utilities exceeding the threshold. Then, in \( \textit{QS} \), the excess utility of \( s \) can be notated as \( \textit{exu}(s, T) \). Let \( \xi \times u(\mathcal{D_T}) \) be the threshold, we obtain \( \textit{exu}(s, T) \) = \( u(s, T) - (|s| \times \xi \times u(\mathcal{D_T})) \). Subsequently, we have the excess utility of \textit{remaining rising sequence} of \( S \) is \( {exu}_\textit{rrs}(S, T) \) = \( u_\textit{rrs}(S, T) - (|\textit{rrs}| \times \xi \times u(\mathcal{D_T})) \). It is obviously that \( \textit{exu}(s, T) \leqslant {exu}_\textit{rrs}(S, T) \). As the problem statement of TAUSPM/TAUSQ, we have \( |\textit{rs}| \geqslant |s| \geqslant 1 \). Then, we derive
	\begin{displaymath}
	\begin{split}
	au({S'}, T) &= \frac{\sum{u({S'}, T, \textit{QS})}}{|{S'}|} \leqslant \frac{\sum{u(S, T, \textit{QS})} + \sum{u(s, T, \textit{QS})}}{|S| + |s|} \\
	&= \xi \times u(\mathcal{D_T}) + \frac{\textit{exu}(S, T) + \textit{exu}(s, T)}{|S| + |s|} \\
	&\leqslant \xi \times u(\mathcal{D_T}) + \frac{\textit{exu}(S, T) + {exu}_\textit{rrs}(S, T)}{|S|} \\
	&= \frac{[\textit{exu}(S, T) + |S| \times \xi \times u(\mathcal{D_T})] + {exu}_\textit{rrs}(S, T)}{|S|} \\
	&\leqslant \frac{u(S, T) + u_\textit{rrs}(S, T)}{|S|}. \nonumber
	\end{split}
	\end{displaymath}
\end{proof}
	
Thus, for \( S \sqsubseteq {S'} \), we have \( au({S'}, T) \leqslant \textit{SRAU}(S, T) \) in \( \mathcal{D} \). It can be shown that this \( \textit{SRAU}(S, T) \) is one of the average utility UBs and enables removing unpromising items in the \textit{remaining sequence}.

\begin{strategy}[depth pruning strategy] \label{str: STR_5}
  \rm Consider a query sequence \( T \) and a \textit{q}-sequence \( S \), if \(\textit{SRAU}(S, T) \) falls below the prespecified minimum acceptable average utility, \( \xi \times u(\mathcal{D_T}) \), then there is no need to check any descendant sequences extending from \( S \). In other words, TAUSQ can terminate the extension of the \textit{q}-sequence \( S \).
\end{strategy} 
	
\begin{definition}[terminated descendants' average utility] \label{def: TDAU}
  \rm In \textit{q}-sequence \( \textit{QS} \), let \( \textit{SRAU}(S, T, \textit{QS}) \) be the suffix remaining average utility of \( S \), where \( T \) is the query sequence. Through one \textit{extension} operation, the sequence \( S \) is expanded to a sequence \( {S'} \). This entails that a node in the LQS-tree denotes \textit{S}, with the node for \( {S'} \) serving as its child. The \( \textit{TDAU}({S'}, T, \textit{QS}) \) is terminated descendants' average utility of \( {S'} \) for \( T \) in \( \textit{QS} \), and is formulated as
	\begin{displaymath}
		\begin{split}
			&\textit{TDAU}({S'}, T, \textit{QS}) 
			= \left\{\begin{array}{rcl} \textit{SRAU}(S, T, \textit{QS}), & {S \sqsubseteq \textit{QS}} \wedge {{S'} \sqsubseteq \textit{QS}} \wedge \textit{qSuf}(T, S) \sqsubseteq \textit{rs} \\ 0, & otherwise \end{array}\right. \nonumber
		\end{split}
	\end{displaymath}
\end{definition} 
Then, the \textit{TDAU} of a \textit{q}-sequence \( S \) with respect to \( T \) the database\( \mathcal{D} \), denoted as \( \textit{TDAU}(S, T) \) = \( \sum_{{S \sqsubseteq \textit{QS}} \wedge {\textit{QS} \in \mathcal{D}}}{\textit{TDAU}(S, T, \textit{QS})} \), is defined as another upper bound of average utility.

As an example, take the database presented in Table \ref{tb: database}. Let a sequence be \( s \) = \( \langle \{b\}, \{cd\} \rangle \), with the query sequence \( T \) = \( \langle \{cd\}, \{e\} \rangle \) and the parameter \( \xi \) = 0.1. The sequence \( {s'} \) = \( \langle \{b\}, \{cd\}, \{i\} \rangle \) is generated form \( s \) by an \textit{S}-Extension. It is evident that both \textit{q}-sequences \( \textit{QS}_1 \) and \( \textit{QS}_3 \) \textit{contain} \( s \) and \( {s'} \). Next, we get \( \textit{TDAU}({s'}, T, \textit{QS}_1) \) = 0, since the corresponding \textit{rs} is null and does not \textit{contain} \textit{qSuf}(T, {s'}). Therefore, \( \textit{TDAU}({s'}, T) \) = \( \textit{TDAU}({s'}, T, \textit{QS}_1) \) + \( \textit{TDAU}({s'}, T, \textit{QS}_3) \) = \( 0 + \textit{SRAU}({s}, T, \textit{QS}_3) \) = \( \frac{0 + 60}{3} \) = 20, which does not exceed the minimum acceptable average utility threshold, \( \xi \times u(\mathcal{D}) \) = 33.3.

\begin{theorem} \label{theo: THEO_2}
  \rm Consider the query sequence \( T \) and a sequence \( {S'} \neq \langle \rangle \) in \( \mathcal{D} \), assume a sequence \( {S''} \) = \( {S'} \) or is extended from the sequence \( {S'} \), and both sequences satisfy the conditions \( \textit{qSuf}({S'}, T) \subset \textit{rs}({S'}, T) \) and \( \textit{qSuf}({S''}, T) \subset \textit{rs}({S''}, T) \). If \( \textit{TDAU}({S'}, T) \leqslant \xi \times u(\mathcal{D_T}) \) then \( \textit{au}({S''}, T) \leqslant \xi \times u(\mathcal{D_T}) \).
\end{theorem}

\begin{proof} \label{proof: PROOF_2}
  \rm Consider sequences \( {S'} \) and \( S \) in \textit{q}-sequence \textit{QS}, both of them satisfy the conditions \( \textit{qSuf}({S'}, T) \subset \textit{rs}({S'}, T) \) and \( \textit{qSuf}({S''}, T) \subset \textit{rs}({S''}, T) \). Let \( {S'} \) be generated from a sequence \( S \) via a single \textit{extension} step. Then, based on Definition \ref{def: TDAU}, we have \( \textit{TDAU}({S'}, T, \textit{QS}) \) = \( \textit{SRAU}(S, T, \textit{QS}) \). Consider any sequence \( {S''} \) that is extended from \( {S'} \) or \( {S''} \) = \( {S'} \), we have \( S \) also a prefix of \( {S''} \). By comparison with Definition \ref{def: SRAU}, we see that \( \textit{au}({S''}, T, \textit{QS}) \leqslant \textit{SRAU}(S, T, \textit{QS}) \). So that we also have \( \textit{au}({S''}, T) \leqslant \textit{TDAU}({S'}, T) \). Therefore, the proposed reduced sequence average utility is one of UBs of the sequence average utility.
\end{proof}

\begin{strategy}[width pruning strategy] \label{str: STR_6}
  \rm Consider a query sequence \( T \) and a \textit{q}-sequence \( {S'} \), if \(\textit{TDAU}({S'}, T) \) falls below the prespecified minimum acceptable average utility, \( \xi \times u(\mathcal{D_T}) \), then there is no need to further explore \( {S'} \) or any of its descendant sequences. In other words, the exploration of the \textit{q}-sequence \( S \) can be terminated at this point in TAUSQ.
\end{strategy}

To improve pruning strategy efficiency, a variant of \textit{SRAU} has been proposed. It no longer strictly adheres to the definition of the upper bound but can be shown to be effective for the pruning strategy. Assume that \( \frac{u(S, T) + u_\textit{rrs}(S, T)}{|S| + |\textit{rrs}|_d} \leqslant \xi \times u(\mathcal{D_T}) \), where $|rrs|_d$ denotes the count of distinct items within all $rrs$ in the database. Then, we derive
	\begin{displaymath}
	\begin{split}
	&u(S, T) + u_\textit{rrs}(S, T) \leqslant \xi \times u(\mathcal{D_T}) \times (|S| + |\textit{rrs}|_d) \\
	&u(S, T) + u_\textit{rrs}(S, T) - \xi \times u(\mathcal{D_T}) \times |\textit{rrs}|_d \leqslant \xi \times u(\mathcal{D_T}) \times |S| \\
    &u(S, T) + u_\textit{rrs}(S, T) - \xi \times u(\mathcal{D_T}) \times |\textit{rrs}| \leqslant \xi \times u(\mathcal{D_T}) \times |S| \\
	&u(S, T) + {exu}_\textit{rrs}(S, T) \leqslant \xi \times u(\mathcal{D_T}) \times |S| \\
	&\frac{u(S, T) + {exu}_\textit{rrs}(S, T)}{|S|} \leqslant \xi \times u(\mathcal{D_T}). \nonumber
	\end{split}
	\end{displaymath}

Based on the aforementioned theorems and proofs for \textit{SRAU}, we have
	\begin{displaymath}
	\begin{split}
	\textit{au}({S'}, T) &= \sum{\textit{au}({S'}, T, \textit{QS})} = \frac{\sum{u({S'}, T, \textit{QS})}}{|{S'}|}\leqslant \frac{\sum{u(S, T, \textit{QS})} + \sum{u(s, T, \textit{QS})}}{|S|} \\
	&= \xi \times u(\mathcal{D}) + \frac{\textit{exu}(S, T) + \textit{exu}(s, T)}{|S|} \\
	&\leqslant \xi \times u(\mathcal{D}) + \frac{\textit{exu}(S, T) + {exu}_\textit{rrs}(S, T)}{|S|} \\
	&= \frac{u(S, T) + {exu}_\textit{rrs}(S, T)}{|S|}. \nonumber
	\end{split}
	\end{displaymath}

Thus, if \( \frac{u(S, T) + u_\textit{rrs}(S, T)}{|S| + |\textit{rrs}|_d} \leqslant \xi \times u(\mathcal{D_T}) \), then we can derive \( \textit{au}({S'}, T) \leqslant \xi \times u(\mathcal{D_T}) \). 	
	
In fact, both the \textit{remaining sequence} and the query sequence \( T \) serve as crucial aspects in TAUSQ. They can be regarded as the starting points for improving algorithmic efficiency. According to the problem definition of TAUSQ, all TAUSPs must contain the query sequence. Considering the length of the query suffix |\textit{qSuf}|, a variant of \textit{SRAU} model, denoted as \textit{vSRAU}, is defined as follows:
\begin{displaymath}
\begin{split}
&\textit{vSRAU}(S, T, p, \textit{QS}) \\
&= \left\{\begin{array}{rcl} \frac{u(S, T, p, \textit{QS}) + u_{\textit{rrs} \cup \textit{qSuf}}(S, j_m, \textit{QS})}{|S| + |\textit{qSuf}|}, & \textit{rs} \neq \emptyset \wedge \textit{qSuf}(T, S) \sqsubseteq rs \wedge |\textit{rrs}|_d \leqslant |\textit{qSuf}(T, S)| \\ \frac{u(S, T, p, \textit{QS}) + u_\textit{rrs}(S, j_m, \textit{QS})}{|S| + |\textit{rrs}|_d}, & \textit{rs} \neq \emptyset \wedge \textit{qSuf}(T, S) \sqsubseteq \textit{rs} \wedge |\textit{rrs}|_d > |\textit{qSuf}(T, S)| \\ 0, & otherwise \end{array}\right. \nonumber
\end{split}
\end{displaymath}
where \( u_{\textit{rrs} \cup \textit{qSuf}}(S, j_m, \textit{QS}) \) is the utility of the subsequence formed by merging the \textit{remaining rising sequence} and the query suffix of the query pattern \( T \) for a prefix sequence \( S \).

Based on this variant model, we also have a variant of \textit{TDAU}, called \textit{vTDAU}:
	\begin{displaymath}
		\begin{split}
			\textit{vTDAU}({S'}, T, \textit{QS}) 
            = \left\{\begin{array}{rcl} \textit{vSRAU}(S, T, \textit{QS}), & {S \sqsubseteq \textit{QS}} \wedge {S' \sqsubseteq \textit{QS}} \wedge \textit{qSuf}(T, S) \sqsubseteq \textit{rs} \\ 0, & otherwise \end{array}\right. \nonumber
		\end{split}
	\end{displaymath}
    
where \( \textit{vSRAU} = \frac{u(S, T, p, \textit{QS}) + u_{\textit{rrs} \cup \textit{qSuf}}(S, j_m, \textit{QS})}{|S| + |\textit{qSuf}|} \). Note that achieving a tighter estimation of \( |\textit{rrs}| \) typically requires additional data structures to record relevant information, which may lead to increased memory usage and computational overhead. Considering the characteristics of TPM task, it is more practical to retain only the \textit{vSRAU} entries that satisfy the condition \( |\textit{rrs}| \leqslant |\textit{qSuf}(T, S)| \). Although this design choice may slightly weaken the pruning effectiveness, it simplifies the data structures and reduces the complexity of the width-based pruning strategy. Moreover, the experimental results presented later confirm the feasibility and effectiveness of this simplified design.

Most utility-based mining algorithms assume all utility values are positive. In pattern growth mining approaches, extending patterns with positive utility items increases utility, while extending with negative utility items decreases it. Hence, certain preprocessing pruning strategies, which rely on the total utility of sequences, limit the general applicability of these algorithms. Additionally, many utility-based mining methods rely on defining strict and tight upper bounds to overestimate potential pattern utility. However, this becomes challenging when utilities can be negative. This negative utility scenario complicates the process of accurately estimating the upper bound, as traditional models developed under the positive utility assumptions no longer apply directly. 

When using the newly defined \textit{remaining rising sequence} for evaluation, whether utilities in the \textit{remaining sequence} are positive or negative does not affect the evaluation result. Because the evaluation metric depends on relative numerical magnitudes rather than absolute values. In this work, strategies \ref{str: STR_2}, \ref{str: STR_3}, and \ref{str: STR_4}, which estimate utility and filter sequences based on the total utility of a sequence and a sequence remaining utility, are not applicable to datasets containing negative utility items. To address this, one modification to strategy \ref{str: STR_2} is to consider only items with positive utility within sequences, which also requires adjusting the input data format accordingly. Alternatively, strategy \ref{str: STR_2} can be discarded entirely. For strategies \ref{str: STR_3} and \ref{str: STR_4}, replacing the parameter \( ru \) with the newly defined \( u_\textit{rrs} \) extends the applicability of the algorithm to datasets with negative utility items.

\subsection{Data Structure for TAUSPM}
\label{sec: dataStruc}

The following part of this section discusses the data structures employed in the proposed strategies and the associated calculations. In HUSPM, utilizing a projected database rather than the initial database for multiple scanning is a common and effective method. The challenge is to efficiently record necessary information in the compact data structure. By filtering with various strategies, the projected database size is maintained at an acceptable level, thus enhancing the overall efficiency of the algorithm.

A \textit{q}-matrix structure is used to represent \textit{q}-sequences in the original database \cite{zhang2022tusq}. This structure is indexed by the identities of items and itemsets. Each \textit{q}-sequence is mapped into two parts: item utility and the corresponding remaining sequence utility, recorded in the utility and \textit{rs} utility matrices, respectively.

The targeted chain \cite{zhang2022tusq} is introduced for recording essential information for utility and upper bound calculations, offering a compact representation of the projected database. Unlike the projected data structure used in the HUSP problem \cite{zhang2021tkus}, this targeted chain adds the length of the longest query prefix matching the pattern's prefix as indispensable information. This addition is sufficient for discovering targeted high-utility sequential patterns. However, when average utility serves as the evaluation metric, further length information is required—not only for the pattern prefix but also for the \textit{remaining sequence} within the projected database. In the method proposed here, this necessary information also includes the length and utility of the \textit{rrs}.

Moreover, differing from Ref. \cite{zhang2022tusq}, the method employs two flags to record the information about \textit{qPre}. These flags enable directly querying of relevant information within the \textit{q}-matrix structure, further enhancing the efficiency afforded by the projection database approach. As noted earlier in the strategy discussion, the length of the longest query prefix can vary even within the same head table, depending on different pattern extensions. The method employing two flags allows this variation to be intuitively and efficiently recorded and reflected during mining. In the projected database example presented in Table \ref{fig: tChain}, the head table includes four fields. First, \textit{QSID} denotes the identifier of the \textit{q}-sequence. Second, \textit{SRAU} serves as the proposed evaluation metric for average utility with respect to the specific query sequence. Third, \textit{IMatch} and fourth, \textit{IIMatch} are two flags of \textit{qPre}. The size of the targeted list aligns with the count of \textit{extension positions} of the \textit{instance} in the \textit{q}-sequence. In the first entry of the targeted list, the unique identifier of the itemset associated with the \textit{extension position} is termed \textit{EID}. The remaining two fields, \textit{Util} and \textit{RrsUtil}, record the value of the \textit{instance} utility and the corresponding \textit{remaining rising sequence} utility at this \textit{extension position}, respectively. It should be noted that the \textit{rrs} used in the proposed method is a global parameter. Accordingly, the utility of the \textit{rrs} recorded in \textit{RrsUtil} is actually an estimated value obtained based on the previous expansion result. When calculating the upper bound for the current expansion, this estimated utility allows easy adjustment—by subtracting the utility of items excluded from \textit{rrs}—to derive the accurate utility of the \textit{rrs}. Moreover, during this process, another important parameter can also be derived, the length of the \textit{remaining rising sequence}. Although it is not directly computed from local parameters and therefore is not explicitly recorded in the targeted List, this parameter remains instrumental in supporting evaluation and pruning strategies within the mining task.

\begin{figure}[ht]
    \centering
    \includegraphics[clip,scale=0.98]{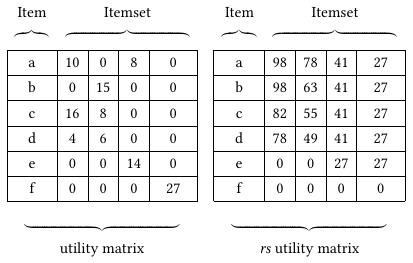}
    \caption{An example of \textit{q}-matrix structure in \( \textit{QS}_2 \).}
    \label{fig: qMatrix}
\end{figure}

\begin{figure}[ht]
    \centering
    \includegraphics[clip,scale=0.98]{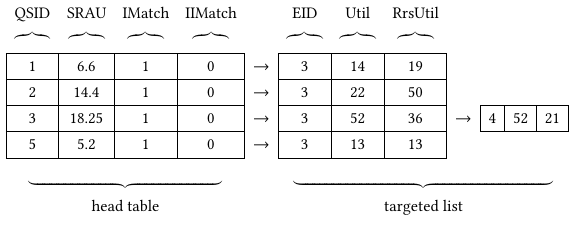}
    \caption{An example of targeted chain structure of \( \langle \{cd\}, \{a\} \rangle \) with \( \xi \) = 0.07.}
    \label{fig: tChain}
\end{figure}

Constructing the targeted chain incurs a time complexity of $\mathcal{O}$(|$D$| $\times$ (|$M$ + 1)|. For more details on the complexity analysis of projected database construction, refer to Ref. \cite{zhang2022tusq}.

\subsection{Proposed TAUSQ Algorithm}
\label{sec: tausq-pg}

The proposed algorithm, TAUSQ-PG, takes a minimum utility threshold, a target sequence, and a pair of databases and corresponding external utility tables as the inputs. It is composed of three main procedures that work together to find TAUSPs in databases. The algorithm is structured as follows.

\begin{algorithm}[ht]
    \small
    \caption{The TAUSQ-PG algorithm}
    \label{alg: 1}
    \LinesNumbered
    \KwIn{A quantitative sequential database \( \mathcal{D} \); A threshold parameter, \( \xi \); A query sequence, \( T \).} 
    \KwOut{TAUSPs in \( \mathcal{D} \).}
        scan \( \mathcal{D} \) to:\\
        \quad 1) Filter out redundant sequences to construct the filtered database \( \mathcal{D}_T \); Constructing the \textit{q}-matrix of each \textit{q}-sequence in \( \mathcal{D}_T \)\\
        \quad 2) For each \textit{q}-sequence \( QS \) in \( \mathcal{D}_T \): Construct its the \textit{q}-matrix\\
        \quad 3) Calculate the utility of \( \mathcal{D}_T \), and the utility value and \textit{SRAU} of each \textit{1}-sequence in \( \mathcal{D}_T \)\\
 	\quad 4) Construct the \textit{LI}-Table based on \( T \) and the projected databases of all \textit{1}-sequences\\
	\For {\rm \textbf{each} \( s \) \( \in \) \textit{1}-sequences}{
		\If{{\( \textit{au}(s) \geqslant \xi \times u(\mathcal{D}_T) \)} \( \wedge \) {\( \textit{IIMatch} = |T| \)}}{
    		update TAUSPs \( \leftarrow \) {TAUSPs \( \bigcup \) \( s \)}\\
		}
		\If{\( \textit{SRAU}(s, T) \geqslant \xi \times u(\mathcal{D}_T) \)}{
    		call PGrowth(\( s \), \( \textit{proDB}(s) \), TAUSPs) \\
		} 
	}
	\textbf{return} TAUSPs
\end{algorithm}

In the first part of the algorithm, the original quantitative sequential database \( \mathcal{D} \) is scanned, and a projection database, \textit{proDB}, is constructed based on the filtered database \( \mathcal{D_T} \). This step follows strategy \ref{str: STR_1} for filtering and serves as the foundation for the following procedures, where all the distinct items and their utility and \textit{ru} are stored in the \textit{proDB}, which is then used in the next stage, the \textit{PGrowth} procedure. The construction of the projection database involves organizing the sequence data into a more accessible form for efficient mining in the subsequent steps.

\begin{algorithm}[ht]
    \small
    \caption{The \textit{PGrowth} algorithm}
    \label{alg: 2}
    \LinesNumbered
    \KwIn{A projected database, \textit{proDB}(\( S \)); A prefix of pattern, \( S \).} 
    \KwOut{TAUSPs in \( \mathcal{D} \).}
        \For {\rm \textbf{each} targeted list \( \textit{tL} \in \textit{proDB} \)}{
	scan \textit{proDB}(\( S \)) to get the \textit{q}-matrix associated with the targeted list\\
	\quad 1) get the collection of \textit{I}-Extension items for \( S \), \textit{iList}\\
	\quad 2) get the collection of \textit{S}-Extension items for \( S \), \textit{sList}\\
         }
	\For {\rm \textbf{each} item \( i \in \textit{iList} \)}{
		\If{\( \textit{TDAU}(S, T) < \xi \times u(\mathcal{D}_T) \)}{
            continue
		}
		call \textit{AUCalcu}(\( S \oplus i \), \( \textit{proDB}(S) \), TAUSPs)\\
	}
	\For {\rm \textbf{each} item \( i \in \textit{sList} \)}{
		\If{\( \textit{TDAU}(S, T) < \xi \times u(\mathcal{D}_T) \)}{
            continue
		}
		call \textit{AUCalcu}(\( S \otimes i \), \( \textit{proDB}(S) \), TAUSPs)\\
	}
\end{algorithm}

The second procedure, \textit{PGrowth}, commences by constructing candidate extension item lists. To facilitate the process efficiently, the algorithm utilizes the \textit{LI}-Table data structure and Strategy \ref{str: STR_2} for filtering. Between Lines 6 and 17, it filters these lists by comparing the minimum acceptable average utility against the value of \textit{TDAU} and applying pruning strategies to eliminate unpromising candidates. Next, the algorithm generates new candidate sequences via either \textit{I}-Extension or \textit{S}-Extension. For each candidate, the \textit{AUCalcu} procedure computes its actual average utility and the value of \textit{SRAU}, which are used to identify sequences likely to form high-utility patterns.

\begin{algorithm}[ht]
    \small
    \caption{The \textit{AUCalcu} algorithm}
    \label{alg: 3}
    \LinesNumbered
    \KwIn{A projected database, \textit{proDB}(\( S \)); A sequence extended by the prefix \( S \), \( {S'} \).} 
    \KwOut{TAUSPs in \( \mathcal{D} \).}
	\( \textit{proDB}({S'}) \) \( \leftarrow \) \{{\textit{proDB} of \( {S'} \)}|{\( {S'} \sqsubseteq QS \) \( \wedge \) \( QS \in \textit{proDB}(S) \)}\}\\
	calculate \( \textit{au}({S'}) \) and \( \textit{SRAU}({S'}) \)\\
	\If {{\( \textit{au}({S'}) \geqslant \xi \times u(\mathcal{D}_T) \)} \( \wedge \) {\( \textit{IIMatch} = |T| \)}}{
		update TAUSPs \( \leftarrow \) {TAUSPs \( \bigcup \) \( {S'} \)}\\
	}
	\If{\( \textit{SRAU}(S', T) \geqslant \xi \times u(\mathcal{D}_T) \)}{
		call PGrowth(\( {S'} \), \( \textit{proDB}({S'}) \), TAUSPs)\\
	} 
\end{algorithm}

As shown in the final procedure, \textit{AUCalcu}, operates by first creating a new projection database for the candidate sequence \( {S'} \). The average utility and \textit{SRAU} of \( {S'} \) are calculated, and these two evaluation parameters are respectively compared against a predefined utility threshold \( \xi \times u(\mathcal{D}_T) \). Once the actual average utility of \( {S'} \) meets or exceeds the prespecified threshold and the flag \textit{IIMatch} attains the value \( |T| \), the sequence qualifies as a TAUSP. Additionally, when the \textit{SRAU} of \( {S'} \) satisfies the predefined threshold \( \xi \times u(\mathcal{D}_T) \), the corresponding sequence will be identified as a potential prefix of a TAUSP. The algorithm continues to generate candidate sequences by extending the current prefix and calling the \textit{PGrowth} procedure recursively. Upon the generation of all candidate sequences, the algorithm returns the set of TAUSPs and terminates.

\subsection{Complexity Analysis} \label{sec: complexity}

Suppose the quantitative sequential database is composed of \( |\mathcal{D}| \) \textit{q}-sequences. There are \( |\mathcal{D}_T| \) \textit{q}-sequences that are contained \( T \) in this database. Assume that the average number of items in \textit{q}-sequence \( QS \) is \( |QS| \). This value in \( \mathcal{D}_T \) is \( |{QS}_T| \). Let \( |I| \) be the number of distinct items in the original database \( |\mathcal{D}| \), then we have the number of distinct items in \( |\mathcal{D}_T| \) is denoted as \( |I_T| \). First of all, starting with the first scanning for the original database, the first step takes \( \mathcal{O}(|\mathcal{D}| \times |QS|) \). The memory complexity is also \( \mathcal{O}(|\mathcal{D}| \times |QS|) \) to construct a \textit{q-matrix} and the corresponding \textit{LI}-Table. Then, the function \textit{PGrowth} is called recursively, and the set of TAUSPs is returned. 

In the second function, \textit{PGrowth}, all items in the filtered projection database are read, and the \textit{iList} and \textit{sList} are built at first. Then, it takes \( \mathcal{O}(|\mathcal{D}_T| \times |{QS}_T|) \) to calculate the \textit{TDAU} of each extension item, and to remove the unpromising ones with low \textit{TDAU}. After the item is appended to the prefix, the next function, AUCalcu, is called to calculate the average utility of the generated candidate sequence. In the function \textit{AUCalcu}, the \textit{RSAU} and the average utility of the generated sequence are calculated for each appending item. Thus, it takes \( \mathcal{O}(|\mathcal{D}_T| \times |{QS}_T|) + \mathcal{O}(|\mathcal{D}_T| \times |{QS}_T|) \), which equals \( \mathcal{O}(|\mathcal{D}_T| \times |{QS}_T|) \). In this step, its memory complexity is \( \mathcal{O}(1) \).

Let \( |L_T| \) be the longest generated sequence length in \( |\mathcal{D}_T| \). During the recursive call of Algorithm 2, the maximum depth and the number of times of recursively calling are \( |L_T| \) and \( |I_T|^{|L_T|} \). During the process of prefix expansion before Algorithm 3 is called, each item in \textit{iList} and \textit{sList} is appended to the prefix. At this time, in the worst case, none of them can be removed. The corresponding time complexity is the sum of all the time complexities of the calling processing, and the memory complexity is \( \mathcal{O}(|I_T|) \). The maximum number of recursive calls of AUCalcu is \( |I_T| \). Therefore, the memory complexity and time complexity of function \textit{PGrowth} are \( \mathcal{O}(|\mathcal{D}_T| \times |{QS}_T| + |I_T|) \) and \( \mathcal{O}(|\mathcal{D}_T| \times |{QS}_T| + |I_T| \times |\mathcal{D}_T| \times |{QS}_T|) \).

Based on the above, the time complexity of TAUSQ is \( \mathcal{O}(|\mathcal{D}| \times |QS|) + |I_T|^{|L_T|} \mathcal{O}(|\mathcal{D}_T| \times |{QS}_T| + |I_T| \times |\mathcal{D}_T| \times |{QS}_T|) \), equivalent to \( \mathcal{O}(|\mathcal{D}||QS| + |I_T|^{|L_T|}|\mathcal{D}_T||{QS}_T|) \). The memory complexity of HAUSP-PG is \( \mathcal{O}(|\mathcal{D}| \times |QS|) + |L_T|\mathcal{O}(|\mathcal{D}_T| \times |{QS}_T| + |I_T|) \), equivalent to \( \mathcal{O}(|\mathcal{D}||QS| + |L_T||\mathcal{D}_T||{QS}_T| + |L_T||I_T|) \). Since \( |{QS}_T| \leqslant |L_T| \), and in the worst case of TAUSQ task — where all \textit{q}-sequences contain the query sequence \( T \) — the maximum time and memory complexities are respectively \( \mathcal{O}(|I|^{|L|}|\mathcal{D}||L|) \) and \( \mathcal{O}(|L|^{2}|\mathcal{D}| + |L||I|) \).

\section{Experiments} \label{sec: experiment}

The performance of the proposed algorithm is assessed with the results of the experiment in this section. The experimental design consists of three parts:
\begin{itemize}
	\item  Comparative experiments are conducted to demonstrate the effectiveness of the targeted querying approach and the efficiency of the proposed algorithm in the context of TAUSPM.

	\item  Based on the ablation experimental results, we analyze how the proposed variants of upper bound models contribute to performance optimization.
	
	\item  Based on the experimental results, we further evaluate the performance of algorithms under varying target sequence lengths.
\end{itemize}	

All algorithms are implemented in Java, and the source code is available at https://github.com/HNUSCS-DMLab/TAUSPM. The experiments are performed on a cloud virtual machine equipped with an AMD EPYC 7542 32-Core CPU and Linux version 5.4.0-166-generic.x86\_64 operating system.

\subsection{Data Description}
\label{sec: data}
\small
\begin{table}[ht]
	\caption{Features of datasets.}  
	\label{table: datasets}
	\centering
    \resizebox{0.76\linewidth}{!}{ 
	\begin{tabular}{@{}ccccccc@{}}
	\toprule
		Dataset			& $|D|$	& $|I|$	& \textit{AvgLen}	& \textit{MaxLen}	& \textit{AvgSeqSize}	& \textit{AvgSetSize} \\
	\midrule
		Bible				& 36369	& 13905	& 21.64			& 100 			& 21.64				& 1.0 \\
		Leviathan			& 5834 	& 9025 	& 33.81			& 100 			& 33.81				& 1.0 \\
		Sign				& 730 	& 267 	& 51.99			& 94 				& 51.99				& 1.0 \\
		{Kosarak}\_{10K}	& 10000 	& 10094 	& 8.14 			& 608 			& 8.14				& 1.0 \\
		{SynDataset}\_{40K}	& 40000 	& 7584 	& 26.85 			& 18 				& 6.20				& 4.33 \\
		{SynDataset}\_{80K}	& 79718 	& 7584 	& 26.80 			& 18 				& 6.19 				& 4.32 \\	\bottomrule
	\end{tabular}}
\end{table} 

For the experiments, we utilize four real-world and two synthetic datasets, all accessible for download from SPMF (http://www.philippe-fournier-viger.com/spmf/). Table \ref{table: datasets} outlines the key features of these datasets, where $|D|$ and $|I|$ signify, respectively, the count of \textit{q}-sequences and the number of distinct items in the original dataset. \textit{AvgLen} represents the average length of \textit{q}-sequence. \textit{MaxLen} is the maximal length of \textit{q}-sequence in the original dataset. \textit{AvgSetSize} and \textit{AvgSeqSize} indicate the average number of \textit{q}-items in one \textit{q}-itemset and the average count of \textit{q}-itemsets in one \textit{q}-sequence respectively.

The Bible and Leviathan datasets are both transformed text datasets, constructed from portions of the books \textit{The Bible} and \textit{Leviathan}, respectively. In these datasets, each sequence corresponds to a sentence, while each item represents a word. The sequence lengths are moderately distributed. The Sign dataset is a sign language dataset, and the version used in this study is derived from the original American Sign Language (ASL) data created by a research at Boston University. This dataset is characterized by relatively long sequences. The Kosarak dataset, a typical clickstream dataset, originates from a Hungarian online news portal. Its most notable feature is the presence of extremely long sequences. In addition, two synthetic datasets, \( {SynDataset}\_{40K} \) and \( {SynDataset}\_{80K} \), are used in the experiments. They contain 40,000 and 79,718 sequences, respectively, with the former being a complete subset of the latter. The experiments conducted on the six datasets provide a comprehensive evaluation of the proposed algorithm's performance in TAUSPM.

\subsection{Speed Performance and Efficiency Analysis}
\label{sec: runtime}

EHAUSM is recognized as the first algorithm designed for mining HAUSPs in the general case \cite{truong2020ehausm}. Based on this algorithm, two baselines, \( {\rm EHAUSM}^+  \) and \( {\rm EHAUSM}^- \), are designed for comparative purposes. Specifically, \( {\rm EHAUSM}^+  \) follows the same recursive querying method as the proposed algorithm \textit{\rm TAUSQ-PG}, where target queries are repeatedly processed during recursion.  In contrast, \( {\rm EHAUSM}^- \) applies a filtering process to the original database based on the target sequence using Strategy \ref{str: STR_1}, performed only at the initial stage. It is worth noting that the original EHAUSM algorithm \cite{truong2020ehausm} performs a preliminary pruning using the AMUB upper bound model before applying its designed tighter upper bounds. This initial filtering proves especially effective for certain datasets, such as the Kosarak dataset. To better highlight the effect of filtering strategies for TPM, the implementation of both baselines in this experiment omits this preliminary pruning. This modification has a negligible impact on most datasets and does not compromise the validity of comparative results or the experimental objective. The target sequences for the six datasets are set to <{356},{10},{10},{10}>, <{8},{17},{8}>, <{8},{9}>, <{11},{218},{6},{148}>, <{1857,4250}>, and <{1857,4250}>, respectively.

\begin{figure}[ht]
	\centering
	\begin{minipage}[t]{0.98\textwidth}
			\includegraphics[clip,scale=0.17]{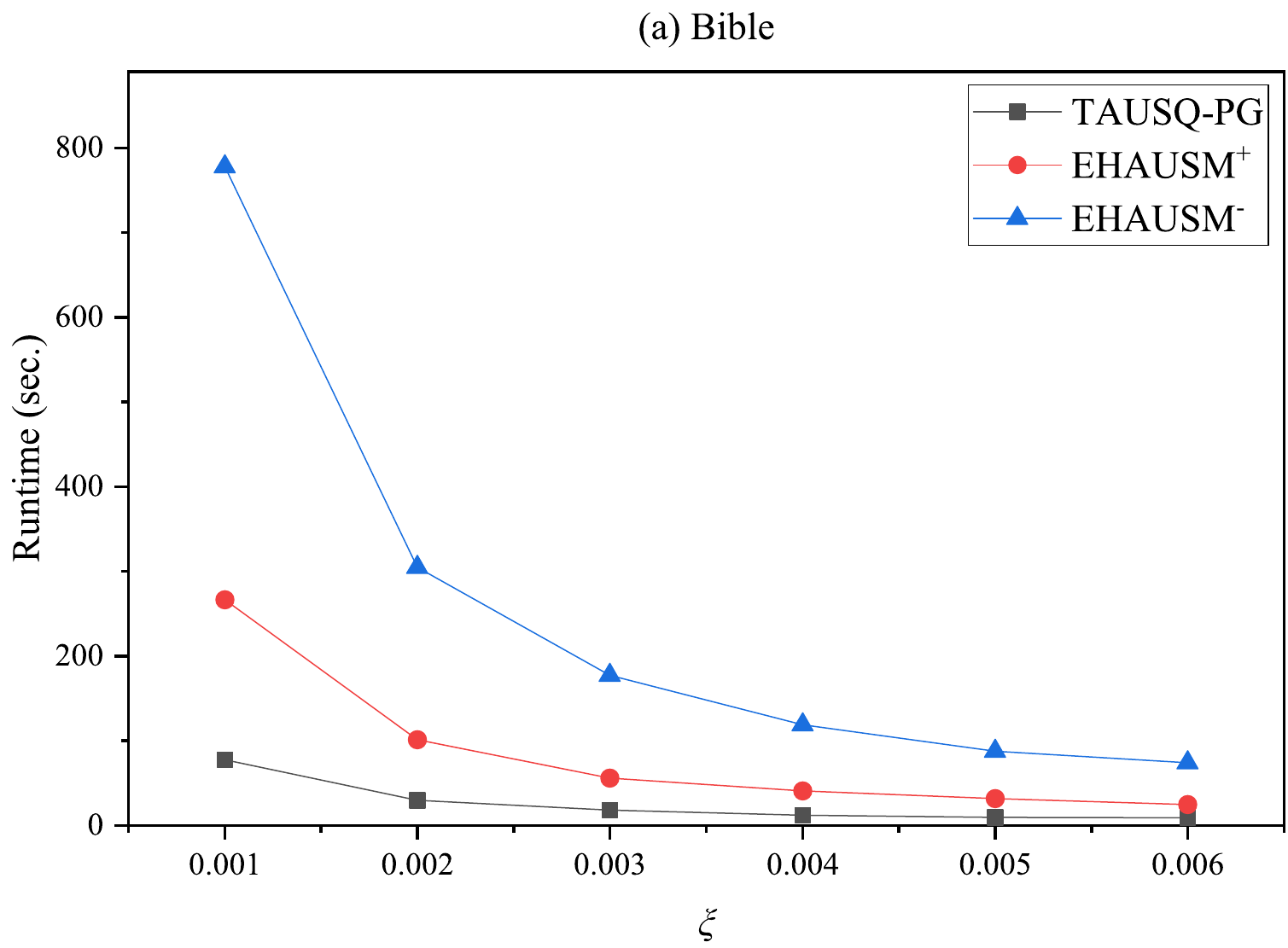}
			\label{fig: 51a}
			\includegraphics[clip,scale=0.17]{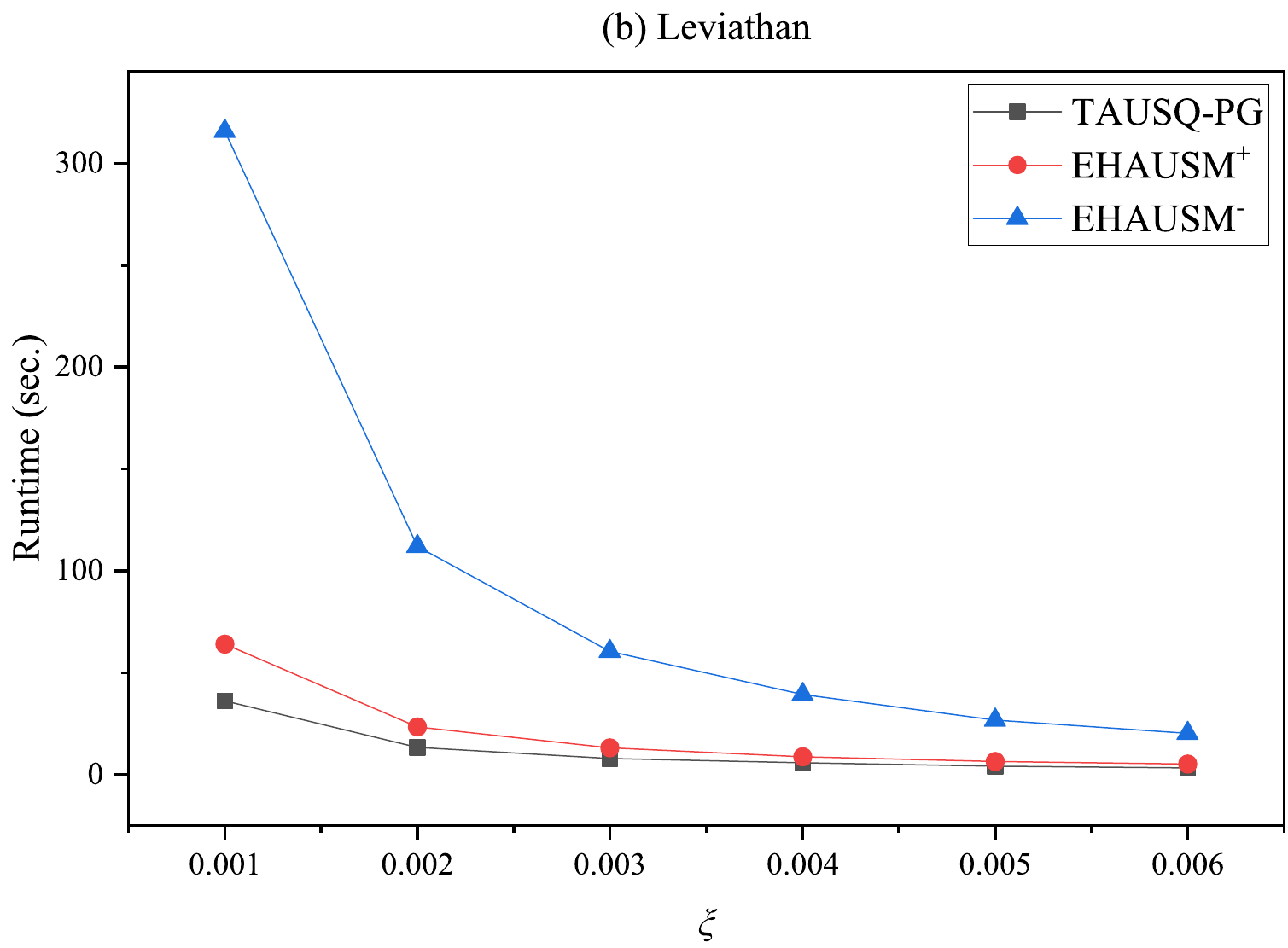}
			\label{fig: 51b}
			\includegraphics[clip,scale=0.17]{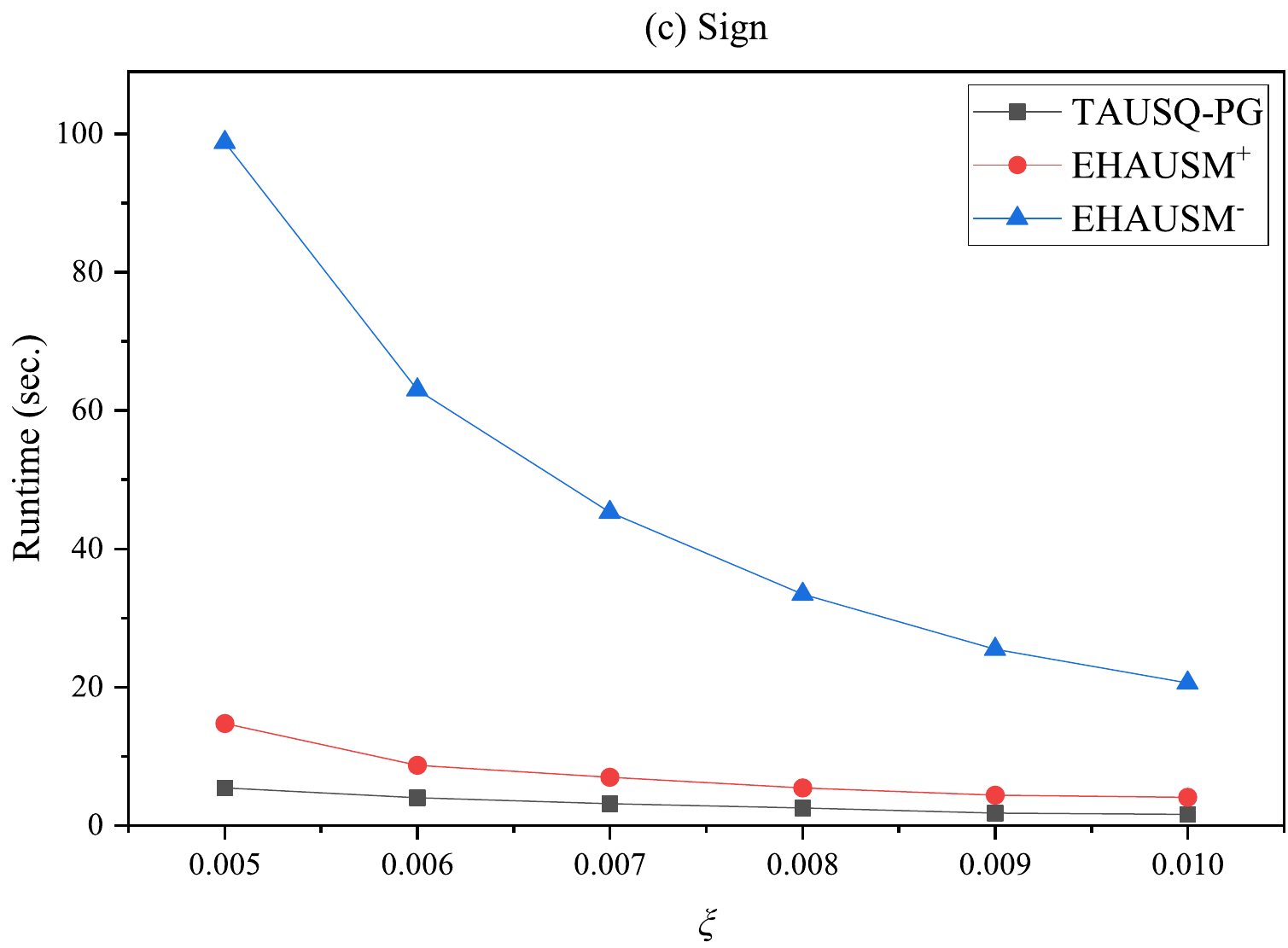}
			\label{fig: 51c}
	\end{minipage}
	\begin{minipage}[t]{0.98\textwidth}
			\includegraphics[clip,scale=0.17]{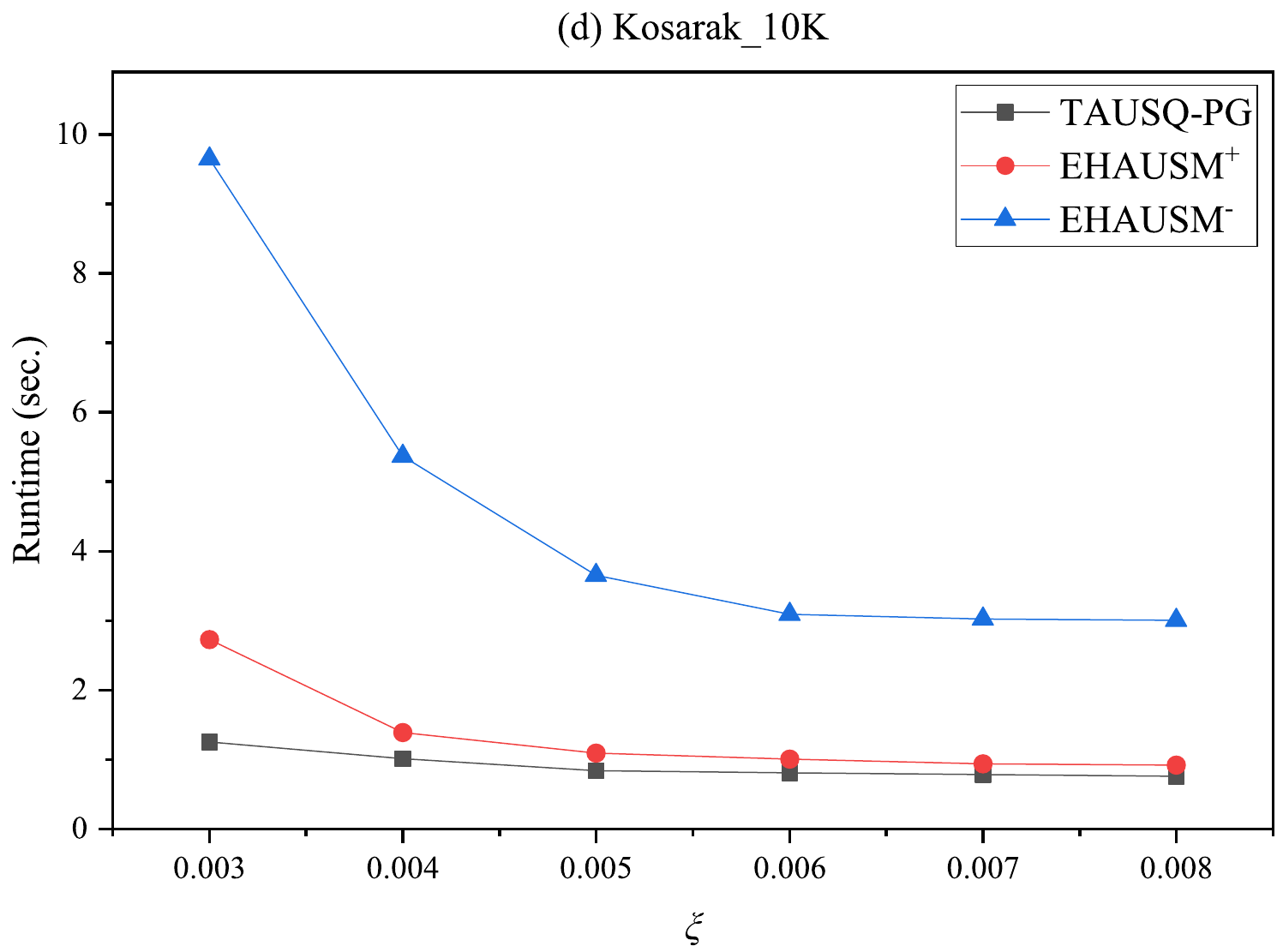}
			\label{fig: 51d}
			\includegraphics[clip,scale=0.17]{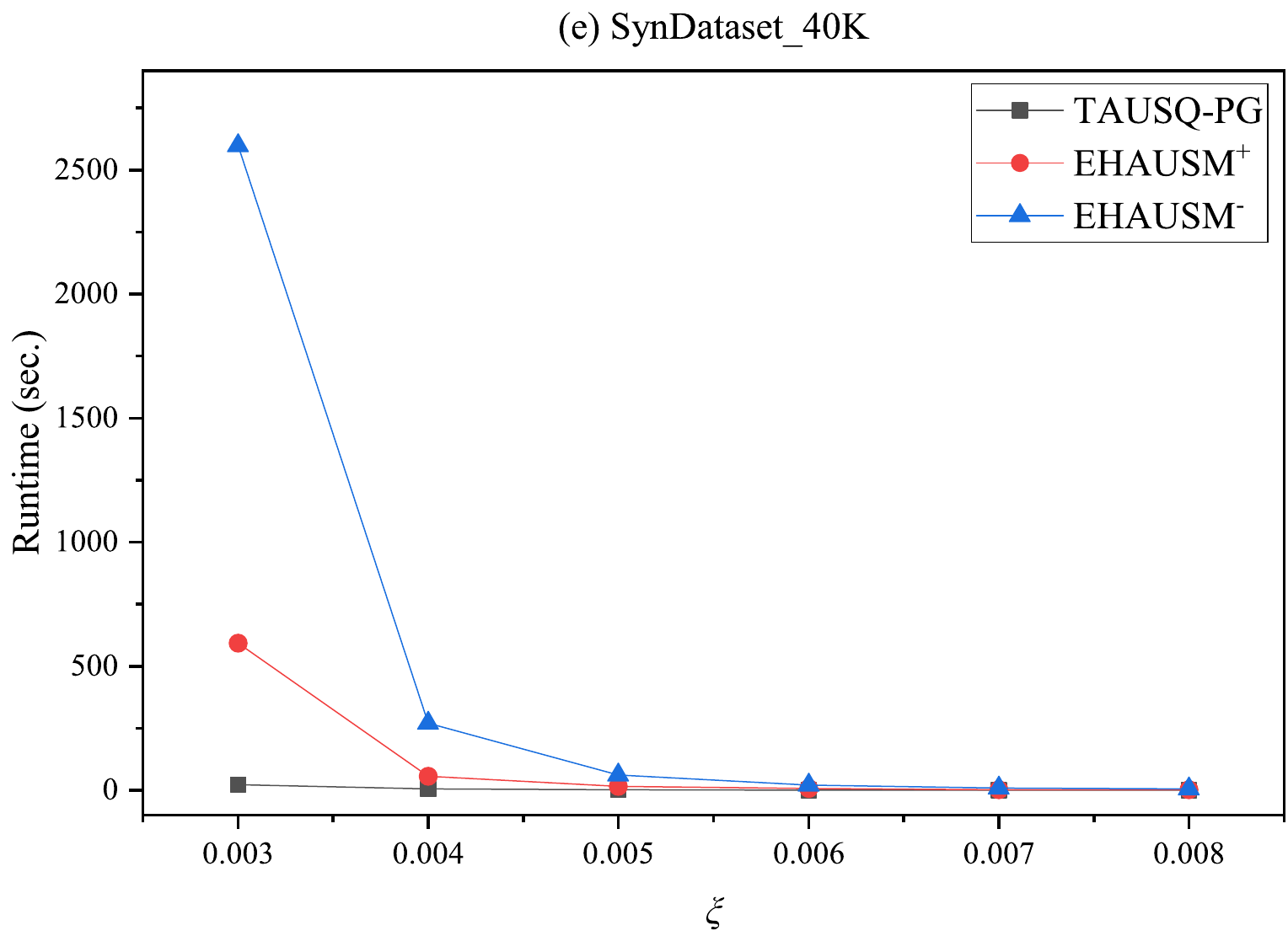}
			\label{fig: 51e}
			\includegraphics[clip,scale=0.17]{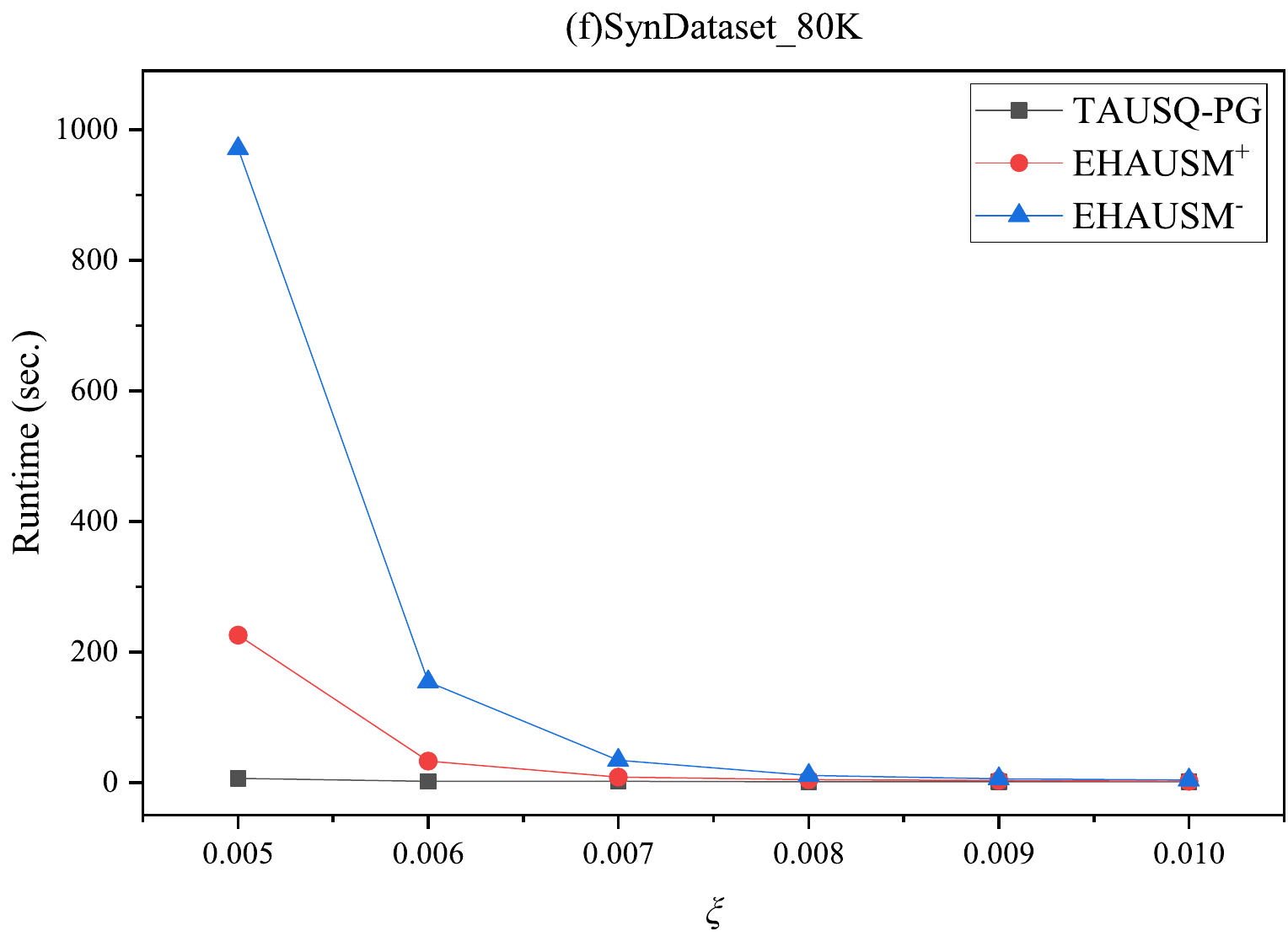}
			\label{fig: 51f}
	\end{minipage}
	\caption{Runtime for various thresholds.}
	\label{fig: 51}
\end{figure}

As shown in Fig. \ref{fig: 51}, increasing the value of \( \xi \) raises the average utility threshold, thereby reducing the runtime for all algorithms. However, \( {\rm EHAUSM}^- \) consistently incurs the highest runtime across all settings, highlighting its inefficiency due to the absence of recursive target filtering. \( {\rm EHAUSM}^+  \) demonstrates improved performance, but still underperforms the proposed TAUSQ-PG. Comparing Fig. \ref{fig: 51}(a) and Fig. \ref{fig: 51}(b), which represent datasets with moderate sequence lengths and similar characteristics, the proposed algorithm shows slightly lower runtime on the larger-scale Bible dataset (Fig. \ref{fig: 51}(a)). For datasets with longer sequences, such as in Fig. \ref{fig: 51}(c), TAUSQ-PG maintains a consistent runtime advantage. This efficiency gap becomes more pronounced in Fig. \ref{fig: 51}(e) and Fig. \ref{fig: 51}(f), where TAUSQ-PG achieves the lowest runtime while the baselines experience a sharp increase as \( \xi \) decreases.

These results demonstrate that the recursive target-querying mechanism adopted in both \( {\rm EHAUSM}^+  \) and TAUSQ-PG is effective in improving mining efficiency. Among them, TAUSQ-PG is particularly well-suited for the TAUSPM task, consistently outperforming the baselines in runtime across diverse datasets.

\subsection{Number of Candidates}
\label{sec: candidate}

The number of candidate sequences is a critical metric for evaluating the search space explored by an algorithm. In the experimental datasets, all three algorithms identify a comparable number of TAUSPs, indicating a consistent level of completeness. However, due to differences in algorithmic strategies, the count of candidate sequences generated by different algorithm varies significantly.

\begin{figure}[ht]
	\centering
	\begin{minipage}[t]{0.98\textwidth}
			\includegraphics[clip,scale=0.15]{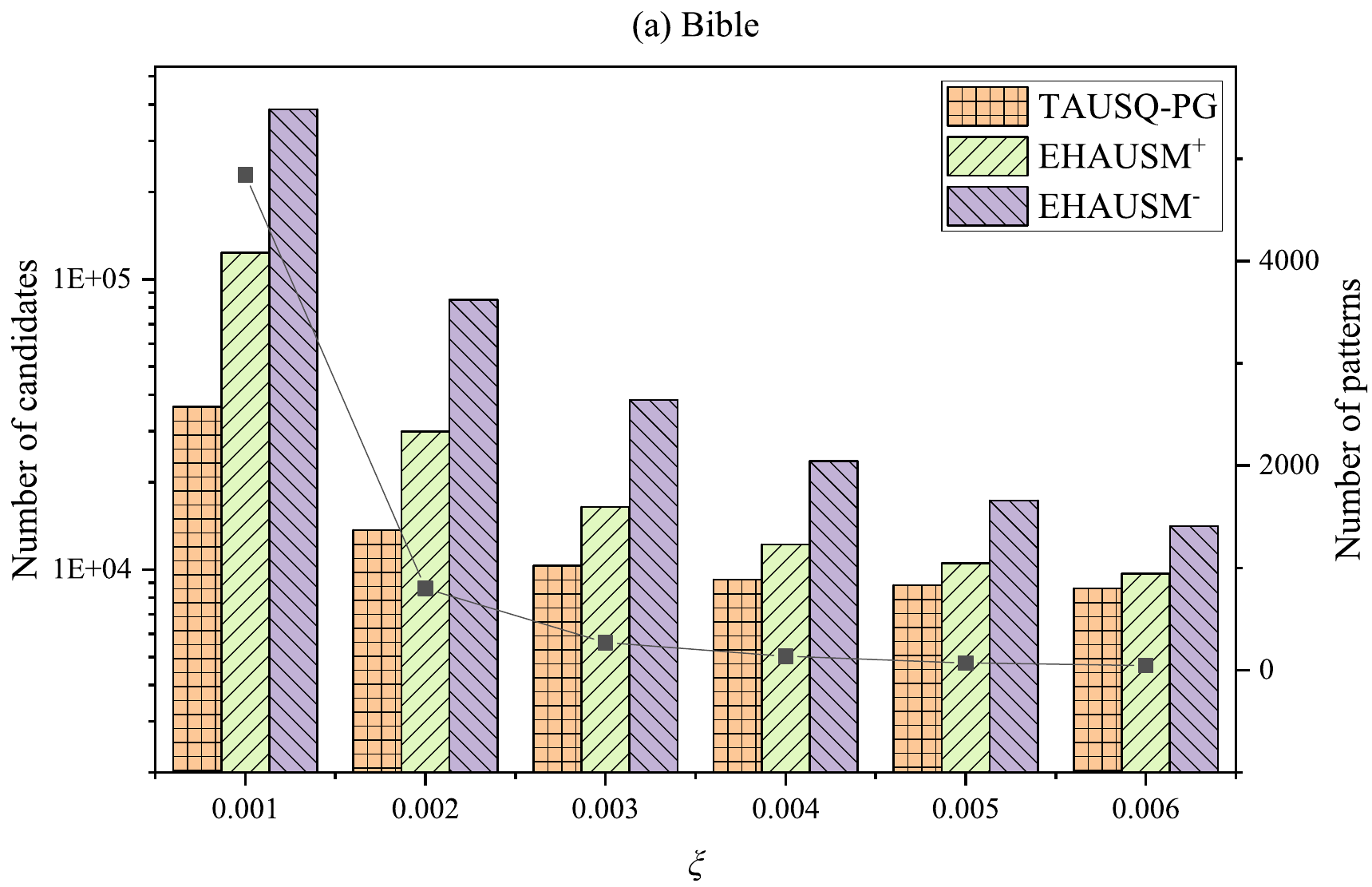}
			\label{fig: 52a}
			\includegraphics[clip,scale=0.15]{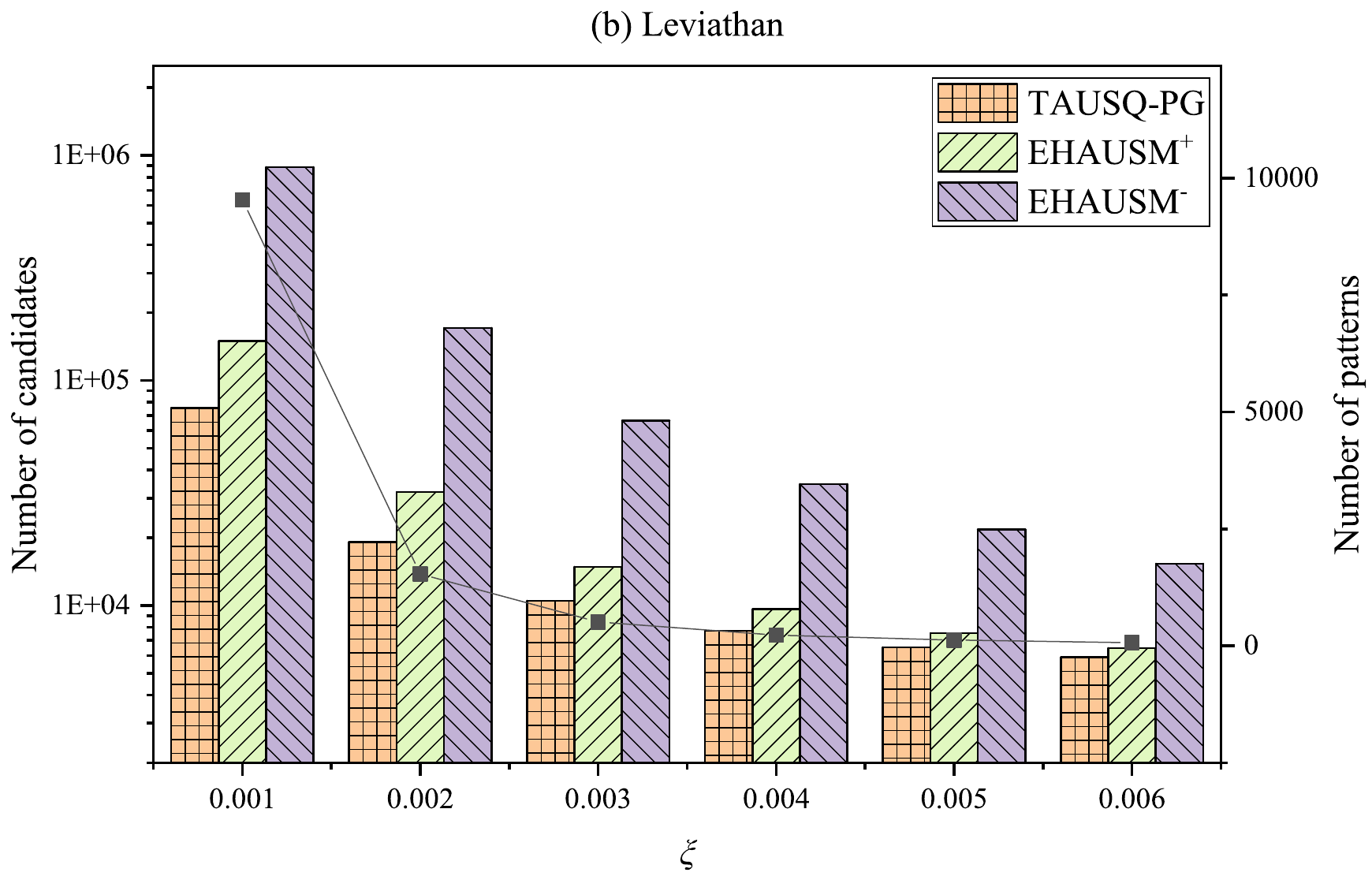}
			\label{fig: 52b}
			\includegraphics[clip,scale=0.15]{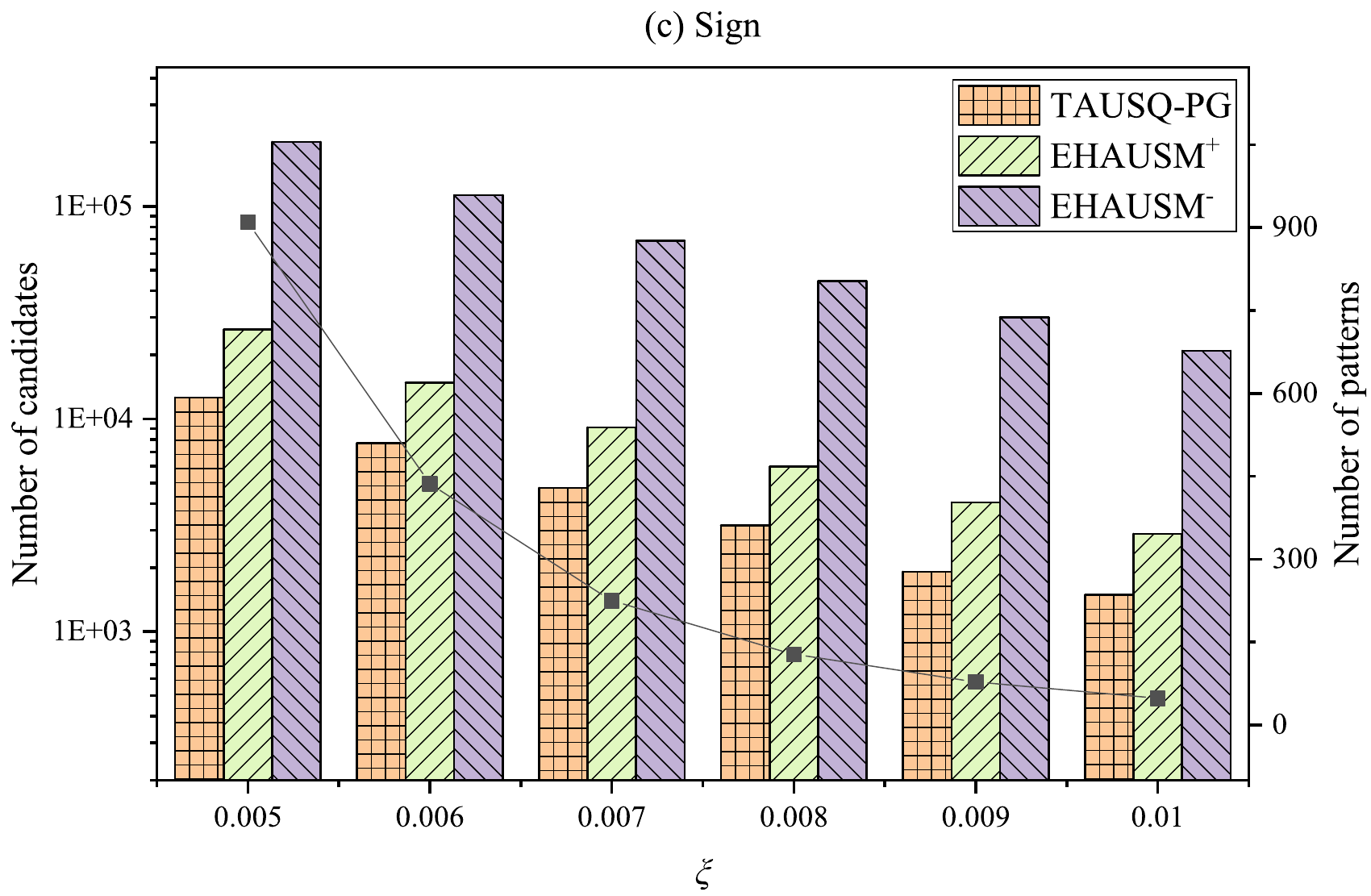}
			\label{fig: 52c}
	\end{minipage}
	\begin{minipage}[t]{0.98\textwidth}
			\includegraphics[clip,scale=0.15]{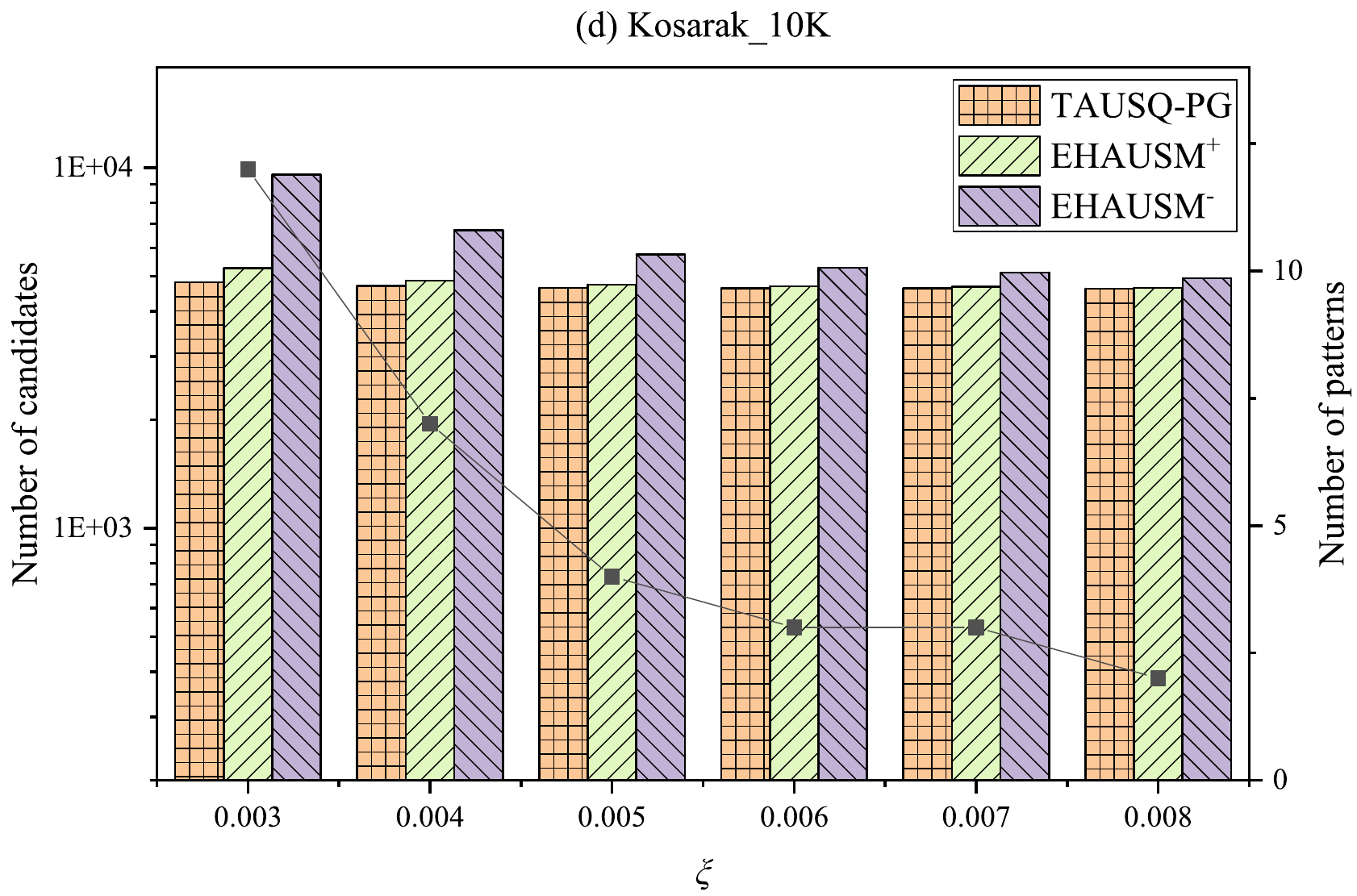}
			\label{fig: 52d}
			\includegraphics[clip,scale=0.15]{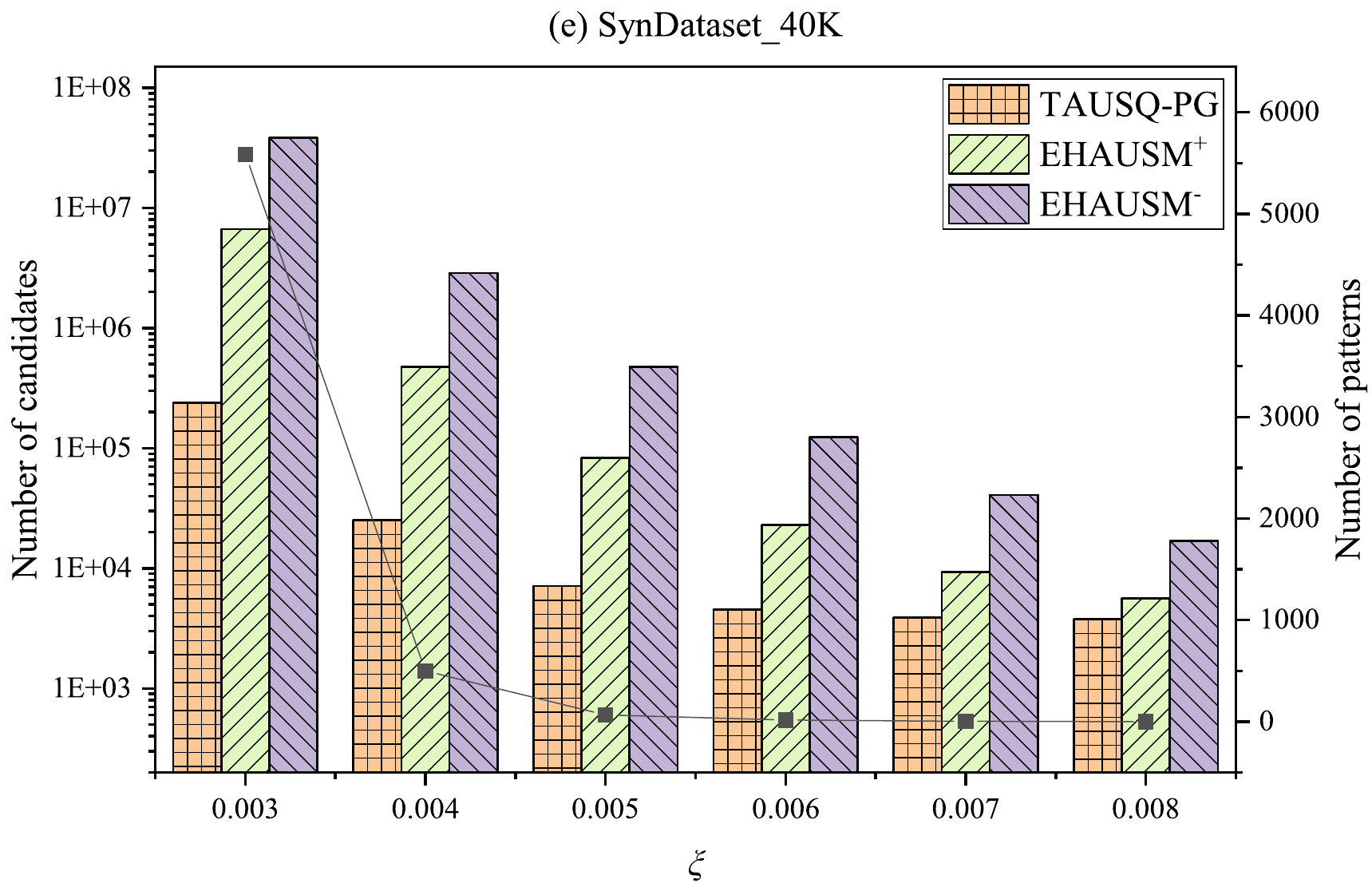}
			\label{fig: 52e}
			\includegraphics[clip,scale=0.15]{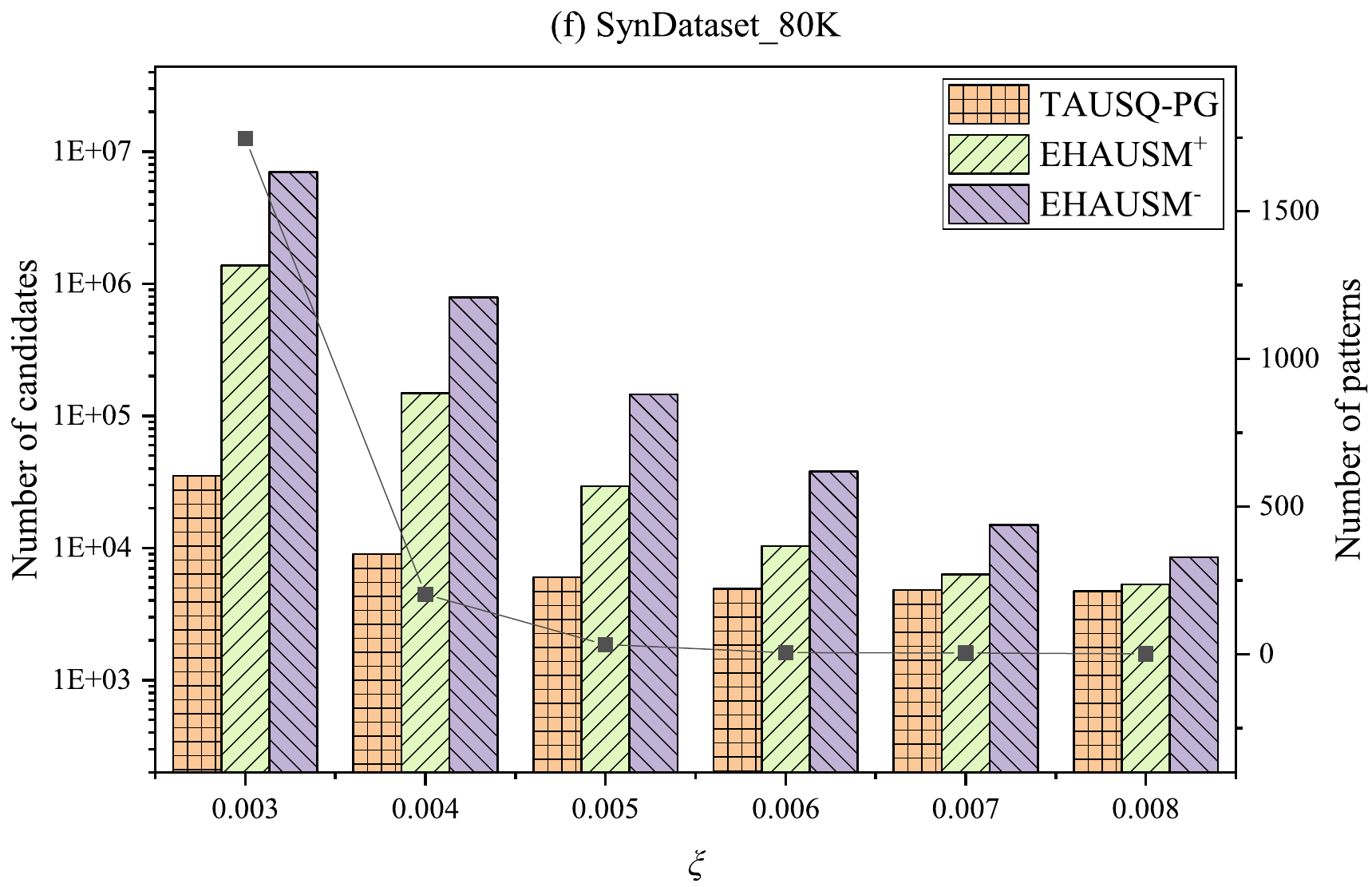}
			\label{fig: 52f}
	\end{minipage}
	\caption{Generated candidate sequences for various thresholds.}
	\label{fig: 52}
\end{figure}

In Fig. \ref{fig: 52}, the TAUSQ-PG consistently generates fewer candidate sequences than both \( {\rm EHAUSM}^+  \) and \( {\rm EHAUSM}^- \) across all datasets. As shown in both Fig. \ref{fig: 52}(a) and Fig. \ref{fig: 52}(b), as the parameter \( \xi \) increases, the number of candidate sequences decreases for both \( {\rm EHAUSM}^+  \) and \( {\rm EHAUSM}^- \). Even as the performance gap narrows, the number of candidates generated by these two baselines remains consistently higher than that of TAUSQ-PG. In Fig. \ref{fig: 52}(b), all three algorithms effectively constrain the search space size, and their performances are relatively close. This advantage is particularly evident in Fig. \ref{fig: 52}(c), where the dataset has a relatively small total volume but contains many long sequences. Similar benefits are observed in synthetic datasets shown in Fig. \ref{fig: 52}(e) and Fig. \ref{fig: 52}(f), which have a much larger number of sequences and the \textit{AvgSetSize} exceeds 1.0.

Although all three algorithms incorporate preprocessing to filter the original dataset, their varying approaches to target querying and the high average utility sequence mining task lead to different levels of effectiveness in reducing the search space. The proposed TAUSQ-PG leverages a pattern growth framework integrated with a tighter variant of the upper bound model. This design enables it to dynamically and efficiently prune unpromising candidate sequences during the mining process. The comparison of candidate sequence generation aligns well with the runtime results discussed above, further highlighting the advantages of the proposed algorithm in addressing the TAUSPM task.

\subsection{Memory Overhead Evaluation}
\label{sec: memory}

Memory usage is a critical metric for evaluating the resource efficiency of pattern mining algorithms. As shown in the experimental results in Fig. \ref{fig: 53}, memory consumption generally increases with the value of \( \xi \) before stabilizing. Across all datasets, the proposed algorithm TAUSQ-PG consistently demonstrates lower memory usage compared to both \( {\rm EHAUSM}^+  \) and \( {\rm EHAUSM}^- \).

\begin{figure}[ht]
	\centering
	\begin{minipage}[t]{0.98\textwidth}
			\includegraphics[clip,scale=0.17]{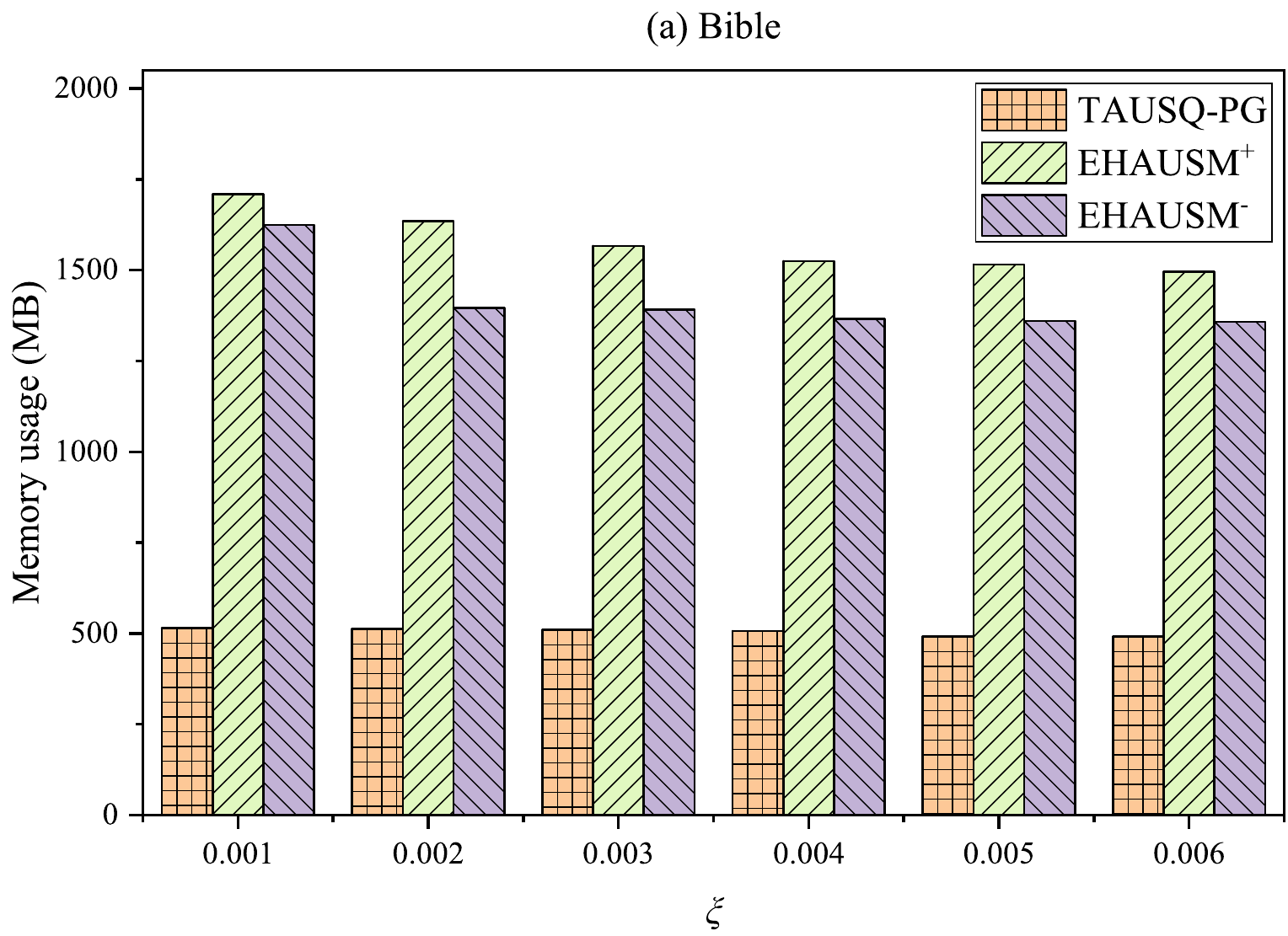}
			\label{fig: 53a}
			\includegraphics[clip,scale=0.17]{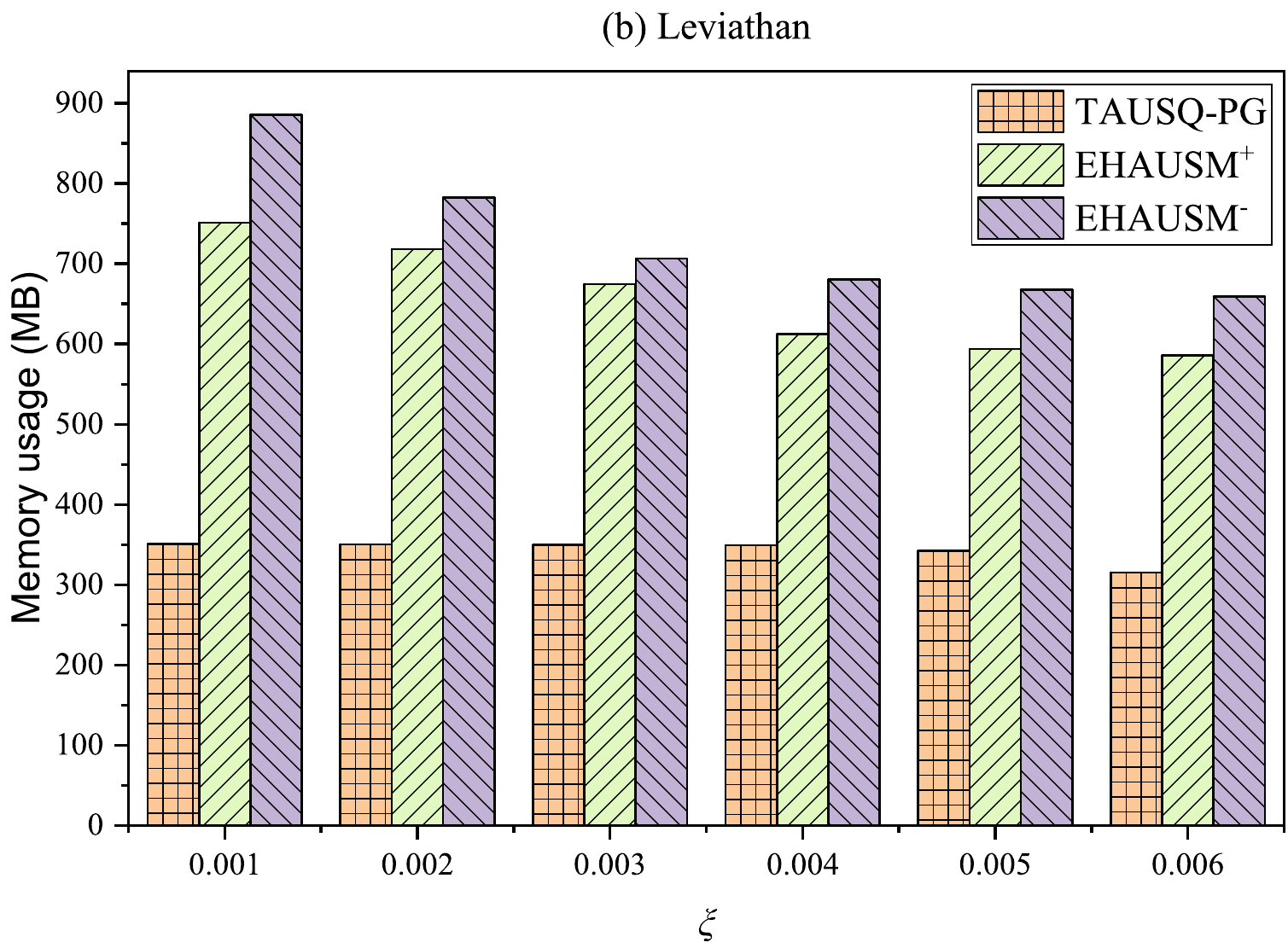}
			\label{fig: 53b}
			\includegraphics[clip,scale=0.17]{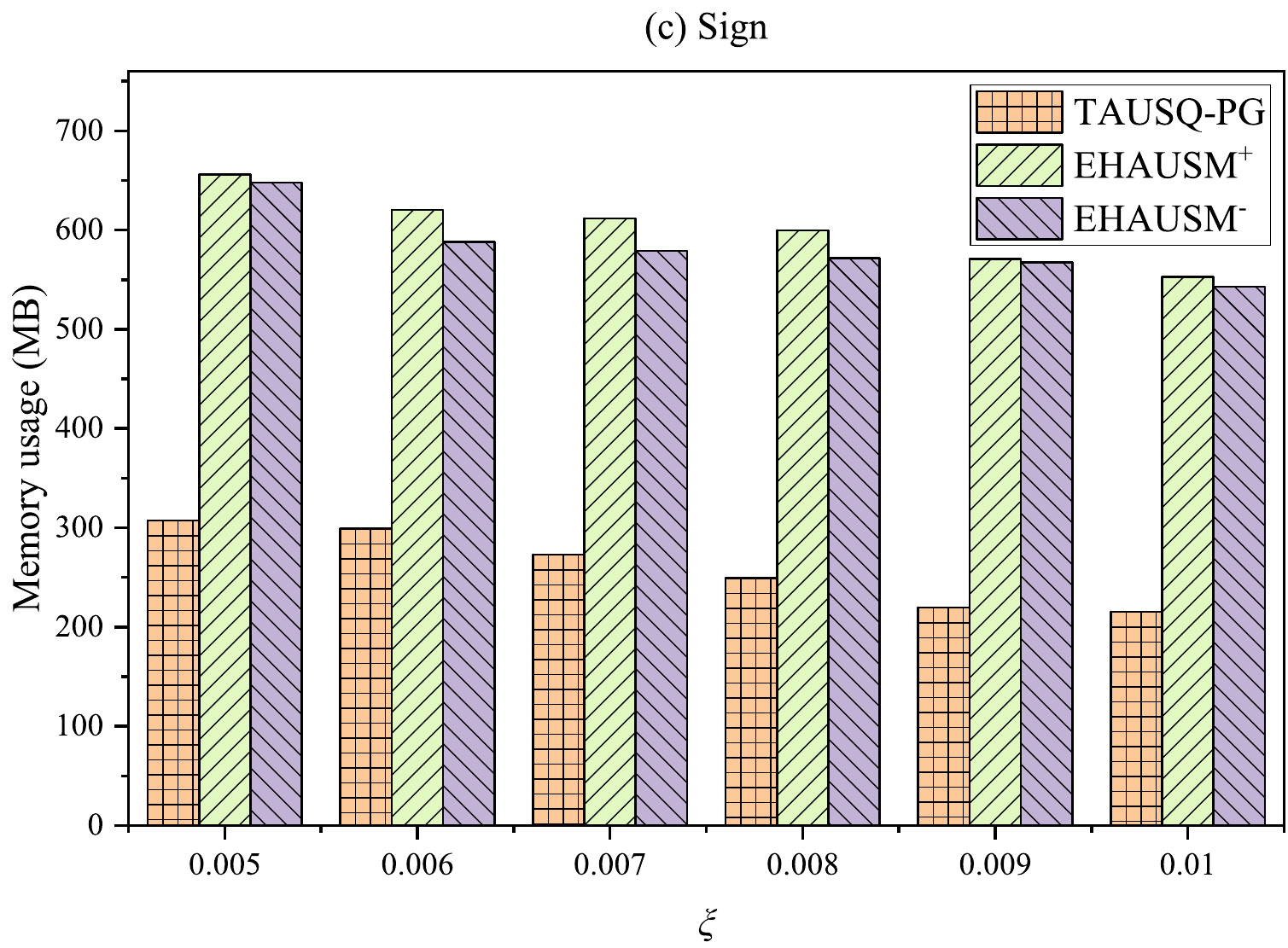}
			\label{fig: 53c}
	\end{minipage}
	\begin{minipage}[t]{0.98\textwidth}
			\includegraphics[clip,scale=0.17]{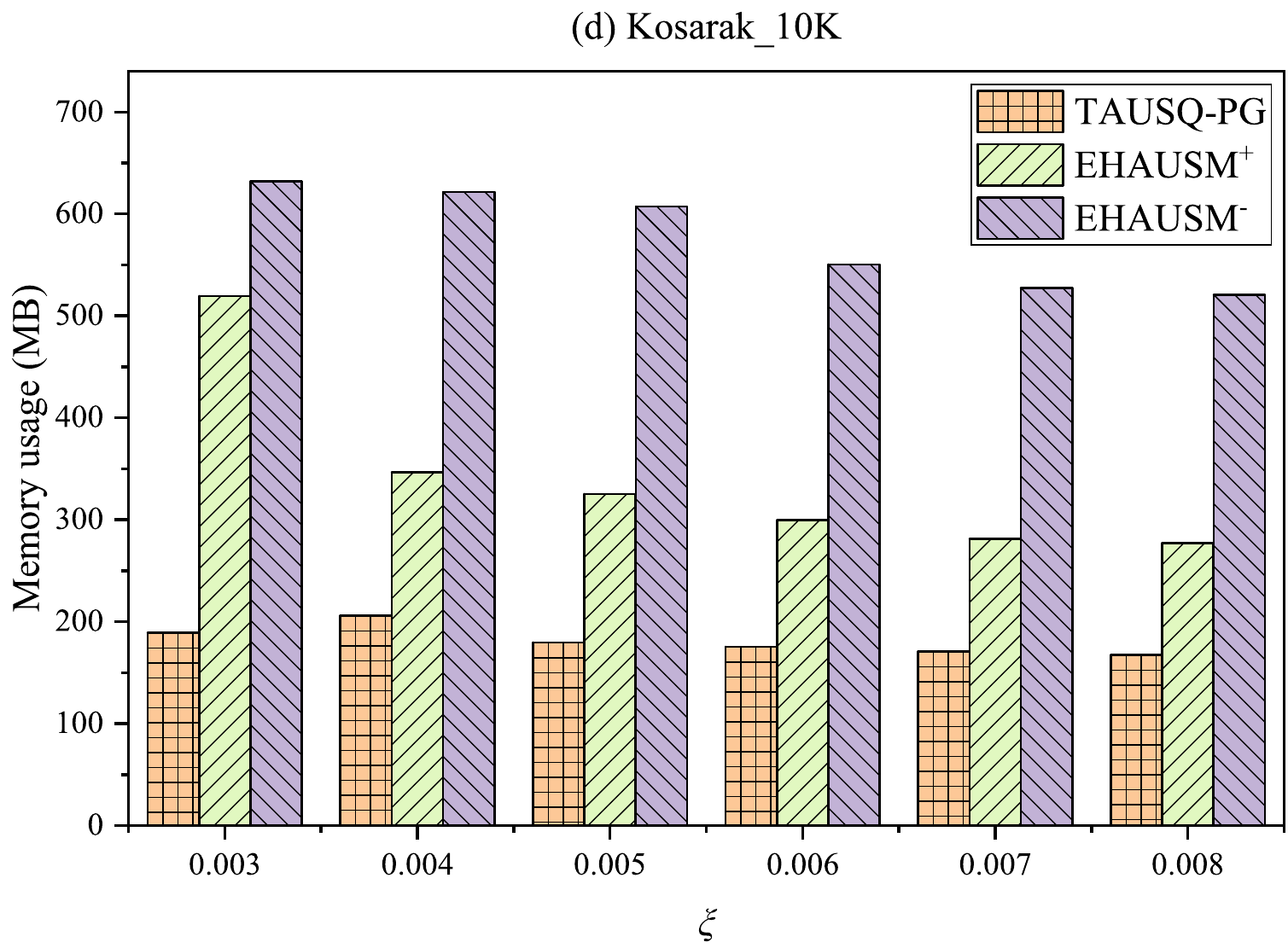}
			\label{fig: 53d}
			\includegraphics[clip,scale=0.17]{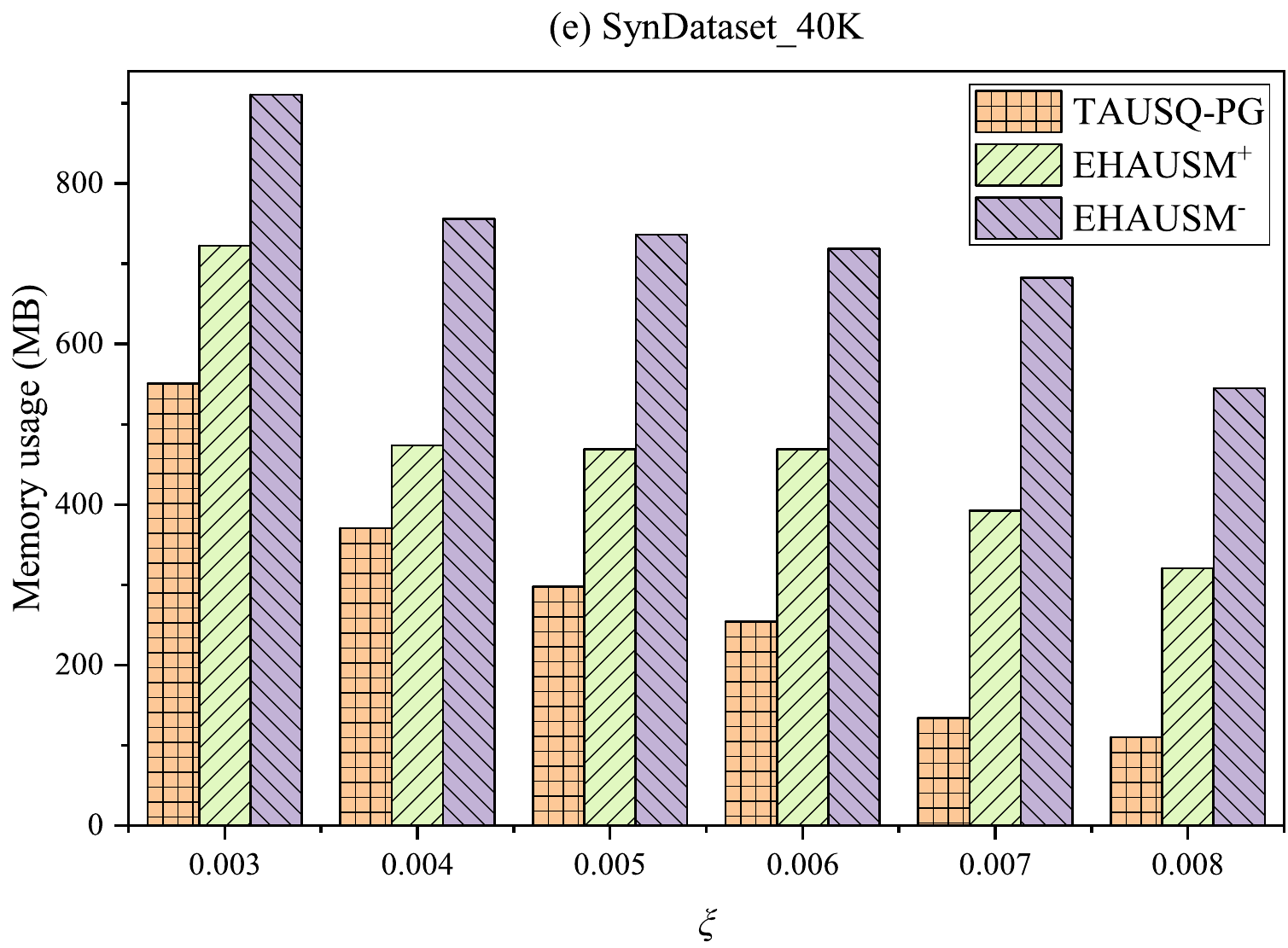}
			\label{fig: 53e}
			\includegraphics[clip,scale=0.17]{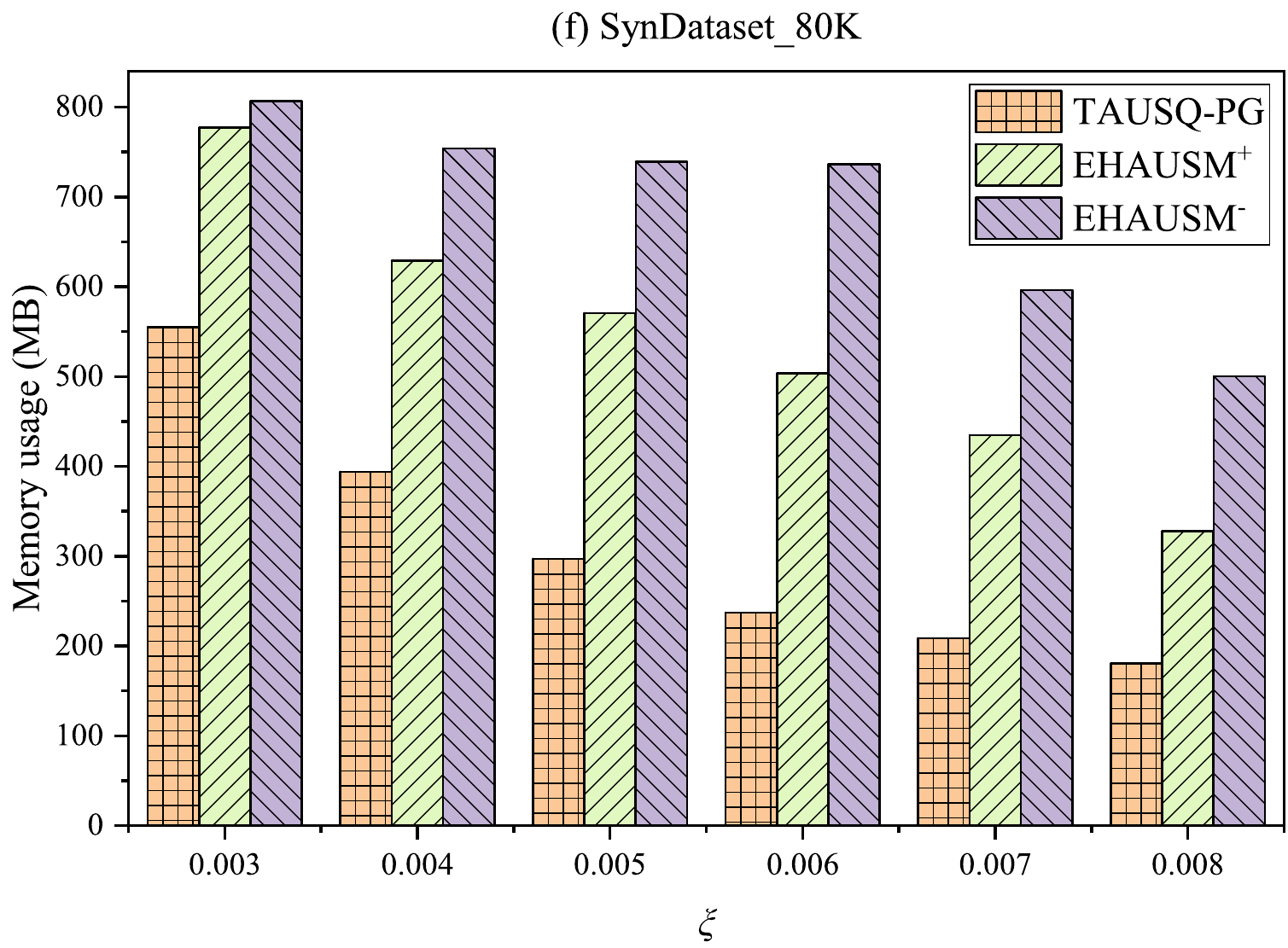}
			\label{fig: 53f}
	\end{minipage}
	\caption{Memory usage for various thresholds.}
	\label{fig: 53}
\end{figure}

In Fig. \ref{fig: 53}(a) and Fig. \ref{fig: 53}(b), memory usage remains relatively stable across the tested parameter range for all three algorithms. Nevertheless, TAUSQ-PG maintains a clear advantage in memory efficiency, consuming less memory in all cases. Interestingly, in Fig. \ref{fig: 53}(a) and Fig. \ref{fig: 53}(c), although TAUSQ-PG still shows the lowest memory usage overall, \( {\rm EHAUSM}^+  \) consumes slightly more memory than \( {\rm EHAUSM}^- \), with a non-negligible gap. This is somewhat counterintuitive, as \( {\rm EHAUSM}^+  \) generally demonstrates better control over search space size, as previously shown in Fig. \ref{fig: 52}. Another notable observation arises from Fig. \ref{fig: 52}(d), where the differences in the number of candidate sequences among the algorithms are relatively minor. In contrast, Fig. \ref{fig: 53}(d) reveals more pronounced differences in memory consumption.

These results suggest that although limiting the number of candidate sequences is generally effective in reducing memory usage, the additional memory overhead introduced by strategies for TPM may become significant. In cases where the search space is relatively small, this overhead remains minimal. However, in larger search spaces, stronger pruning mechanisms are necessary to offset the extra memory cost. Therefore, in Fig. \ref{fig: 53}(e) and Fig. \ref{fig: 53}(f), the memory efficiency of different methods becomes even more evident. The proposed TAUSQ-PG has a noticeable reduction in memory usage compared to the baselines.

\subsection{Ablation Analysis of Upper Bound Models}
\label{sec: ablationAnalysis}

In this subsection, we conduct an ablation study by varying the upper bound models used in the proposed algorithm to evaluate their effectiveness and necessity. The goal is to understand how different upper bound modeling strategies influence the performance of TAUSQ-PG under the average utility framework. Inspired by prior work in utility-oriented research \cite{zhang2022tusq}, where ablation analysis has been employed to assess the impact of different pruning strategies, we apply a similar methodology in the context of average-utility-based mining. As a necessary extension, we further evaluate how incorporating different lengths into the upper-bound models affects pruning effectiveness and overall runtime.

Based on the proposed algorithm TAUSQ-PG, we design three baseline variants for comparison: \( {\rm TAUSQ}_{\rm rrs} \), \( {\rm TAUSQ}_{\rm qSuf} \), and \( {\rm TAUSQ}_{\rm none} \). The variant \( {\rm TAUSQ}_{\rm rrs} \) considers only the length of the prefix and \textit{rrs} subsequence in the upper bound estimation, while \( {\rm TAUSQ}_{\rm qSuf} \) incorporates only the length of the prefix and \textit{qSuf}. In contrast, \( {\rm TAUSQ}_{\rm none} \) includes only the length of the prefix and disregards both \textit{rrs} and \textit{qSuf}.

\begin{figure}[ht]
	\centering
	\begin{minipage}[t]{0.98\textwidth}
			\includegraphics[clip,scale=0.17]{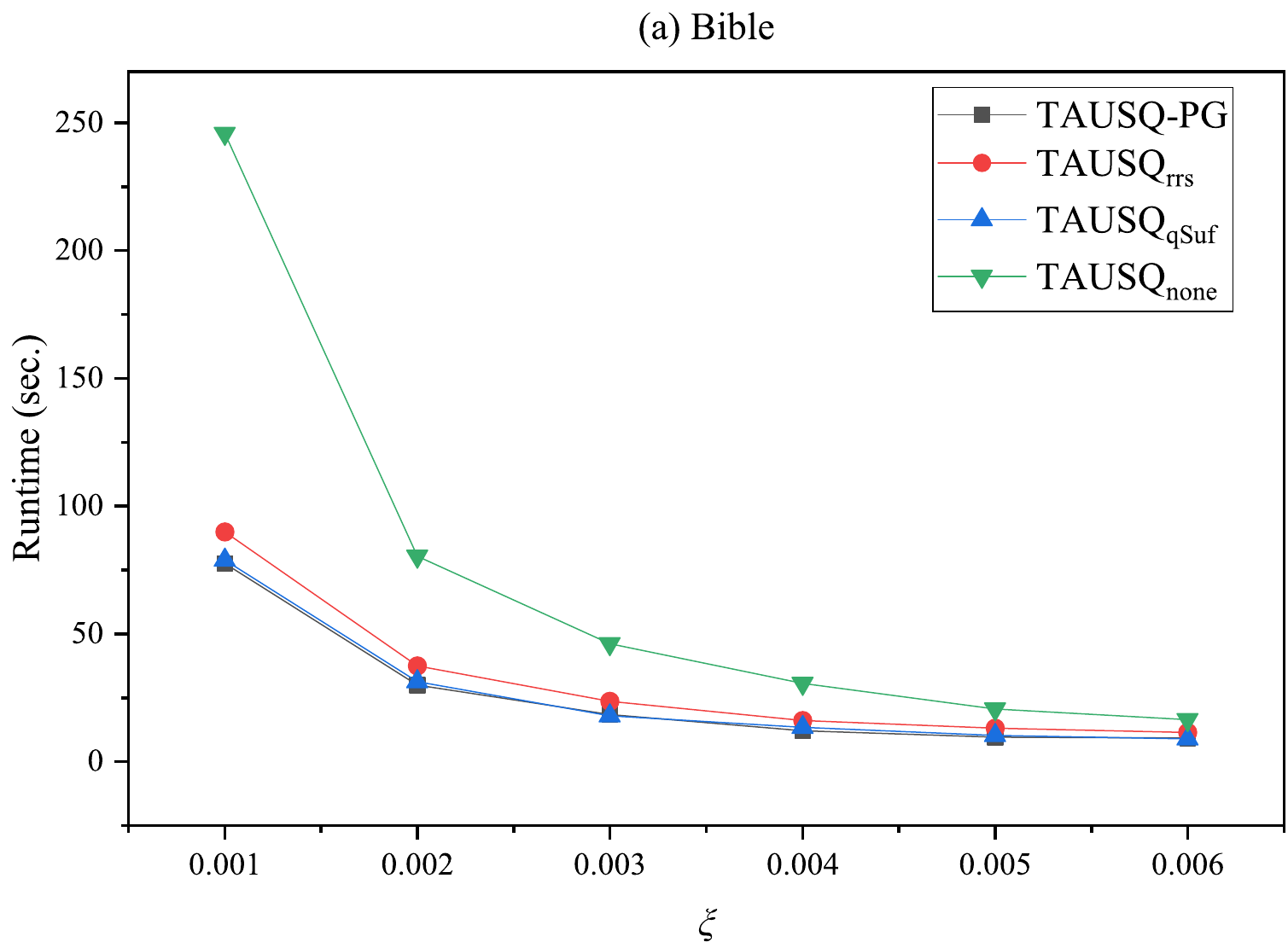}
			\label{fig: 54a}
			\includegraphics[clip,scale=0.17]{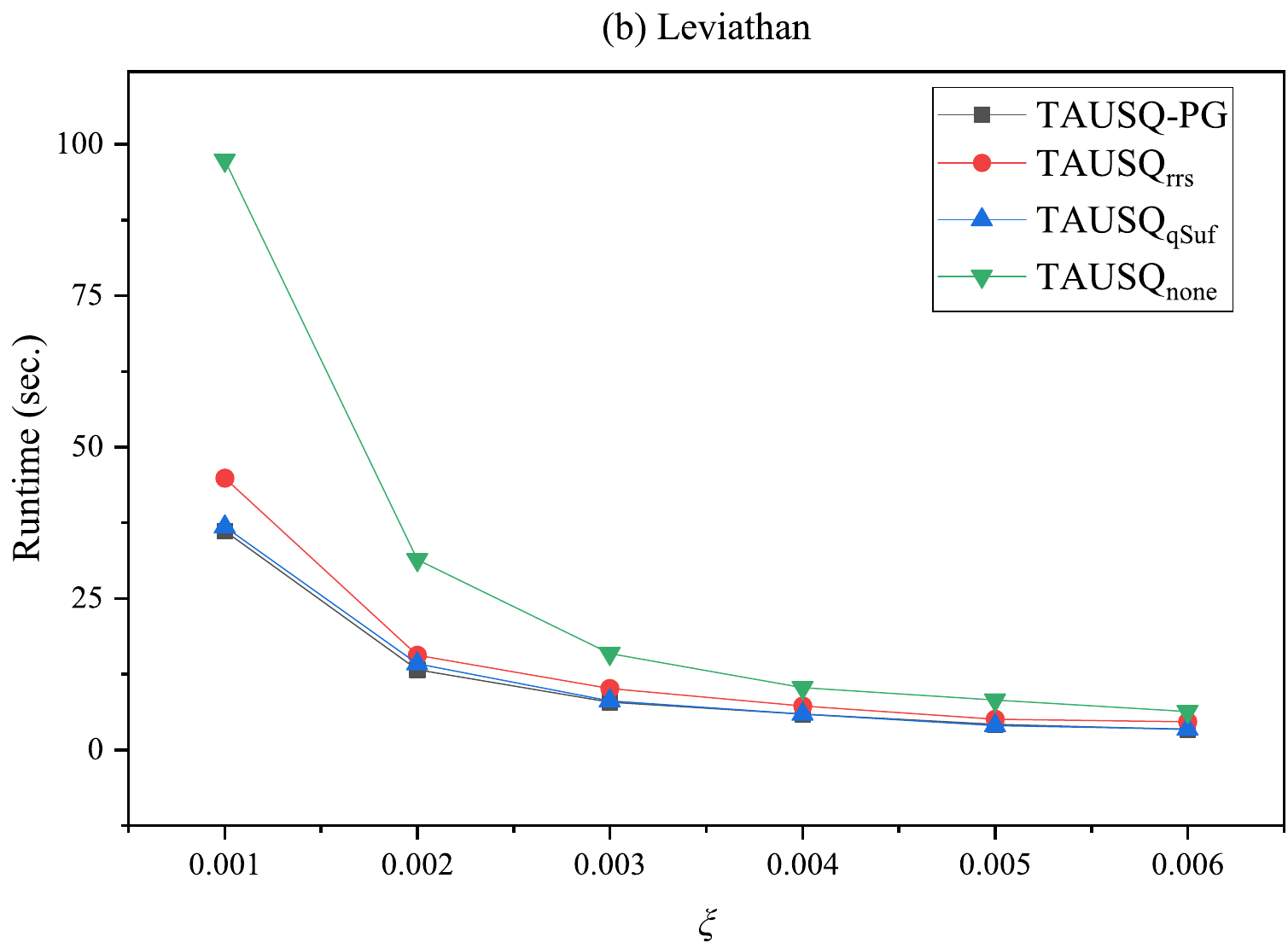}
			\label{fig: 54b}
			\includegraphics[clip,scale=0.17]{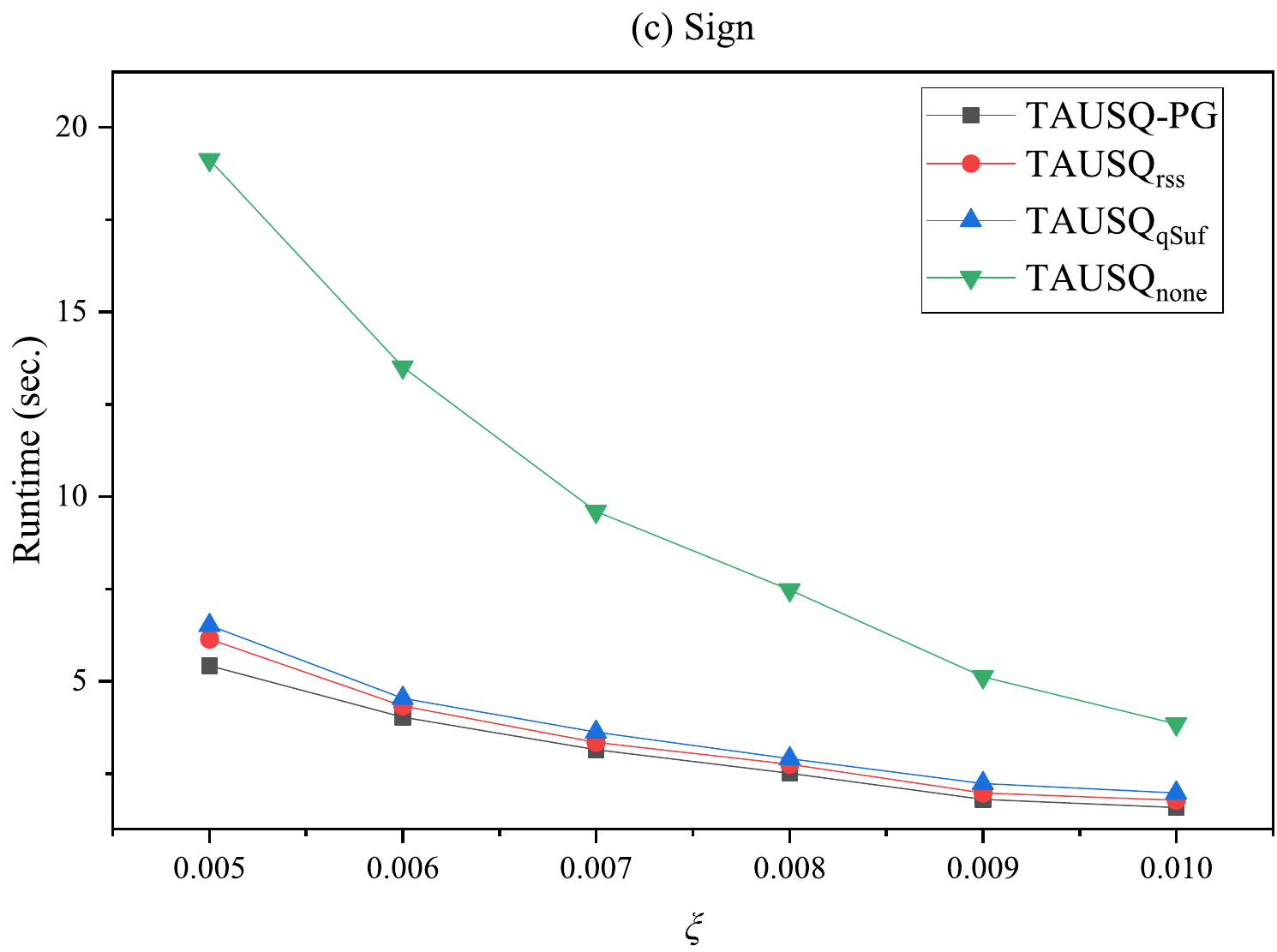}
			\label{fig: 54c}
	\end{minipage}
	\begin{minipage}[t]{0.98\textwidth}
			\includegraphics[clip,scale=0.17]{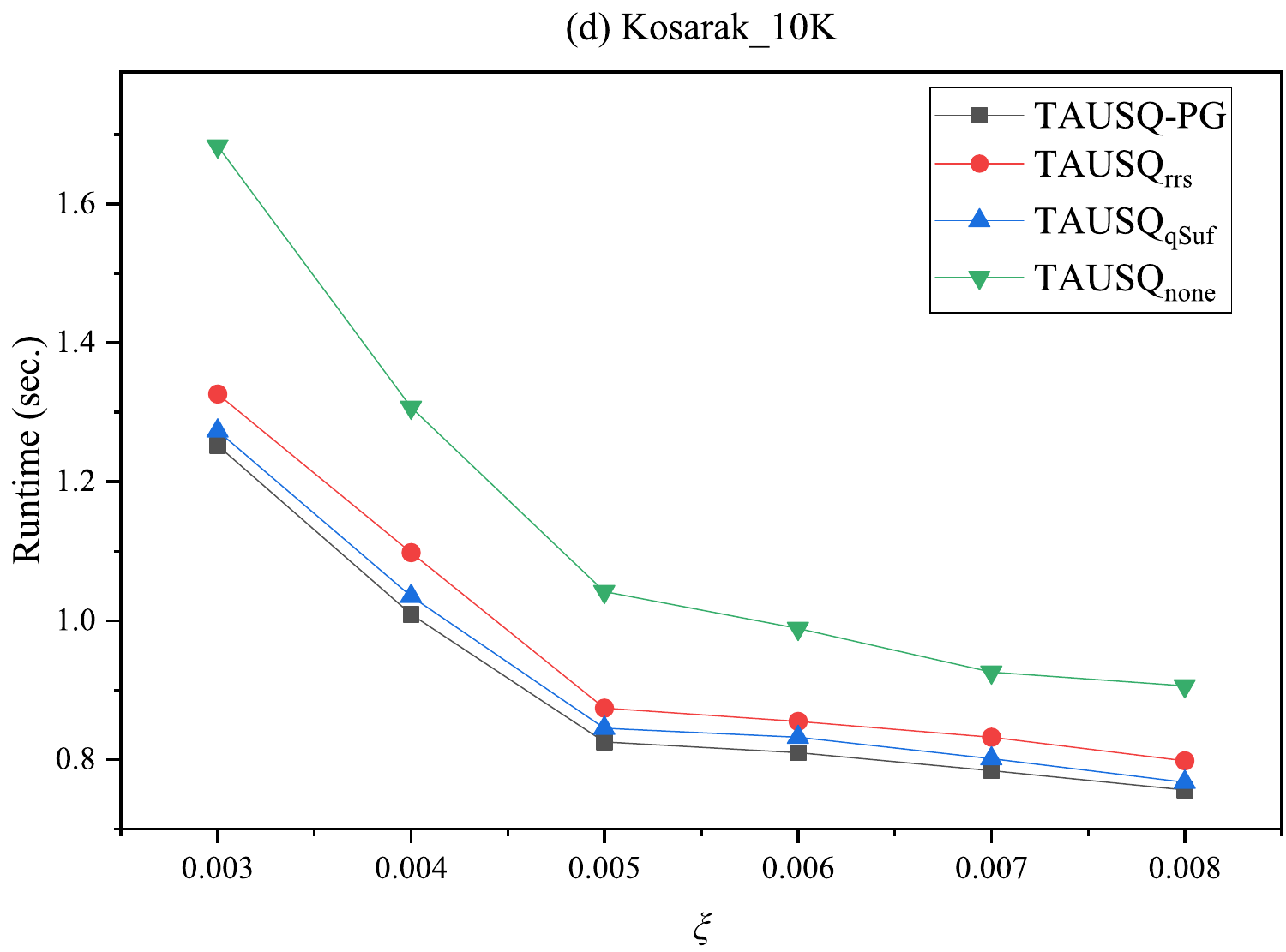}
			\label{fig: 54d}
			\includegraphics[clip,scale=0.17]{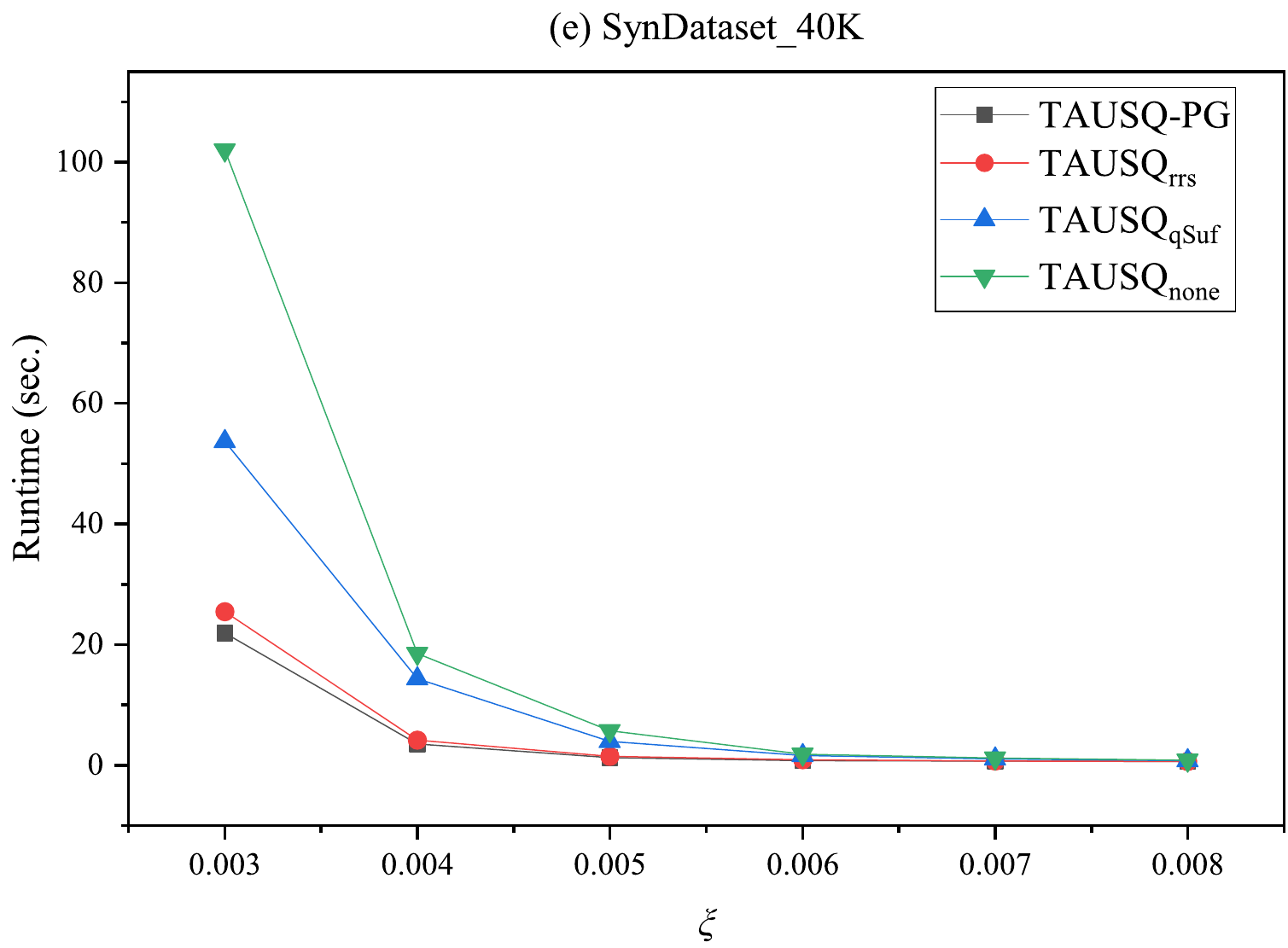}
			\label{fig: 54e}
			\includegraphics[clip,scale=0.17]{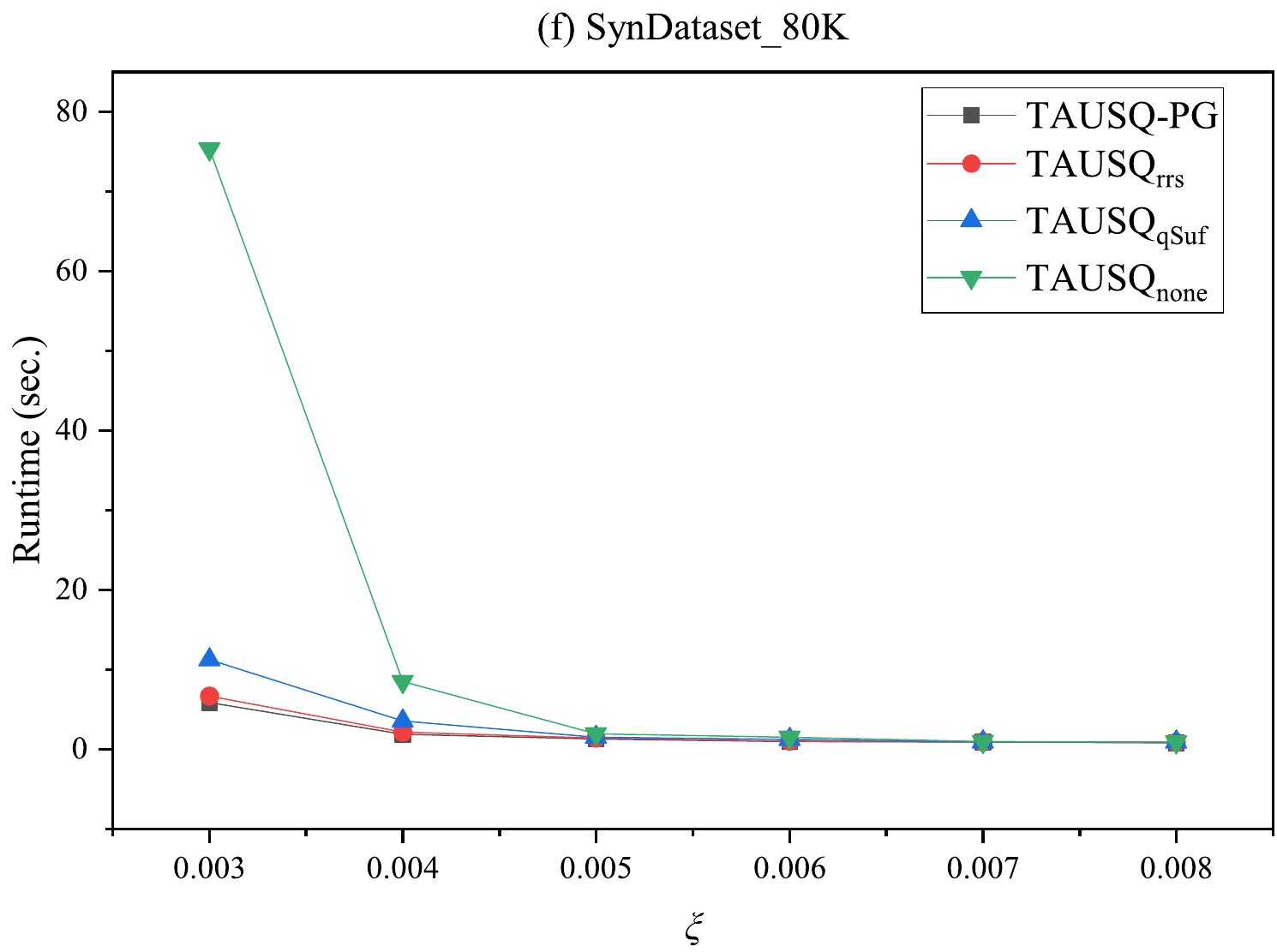}
			\label{fig: 54f}
	\end{minipage}
	\caption{Runtime for various upper bound models.}
	\label{fig: 54}
\end{figure}

\begin{figure}[ht]
	\centering
	\begin{minipage}[t]{0.98\textwidth}
			\includegraphics[clip,scale=0.15]{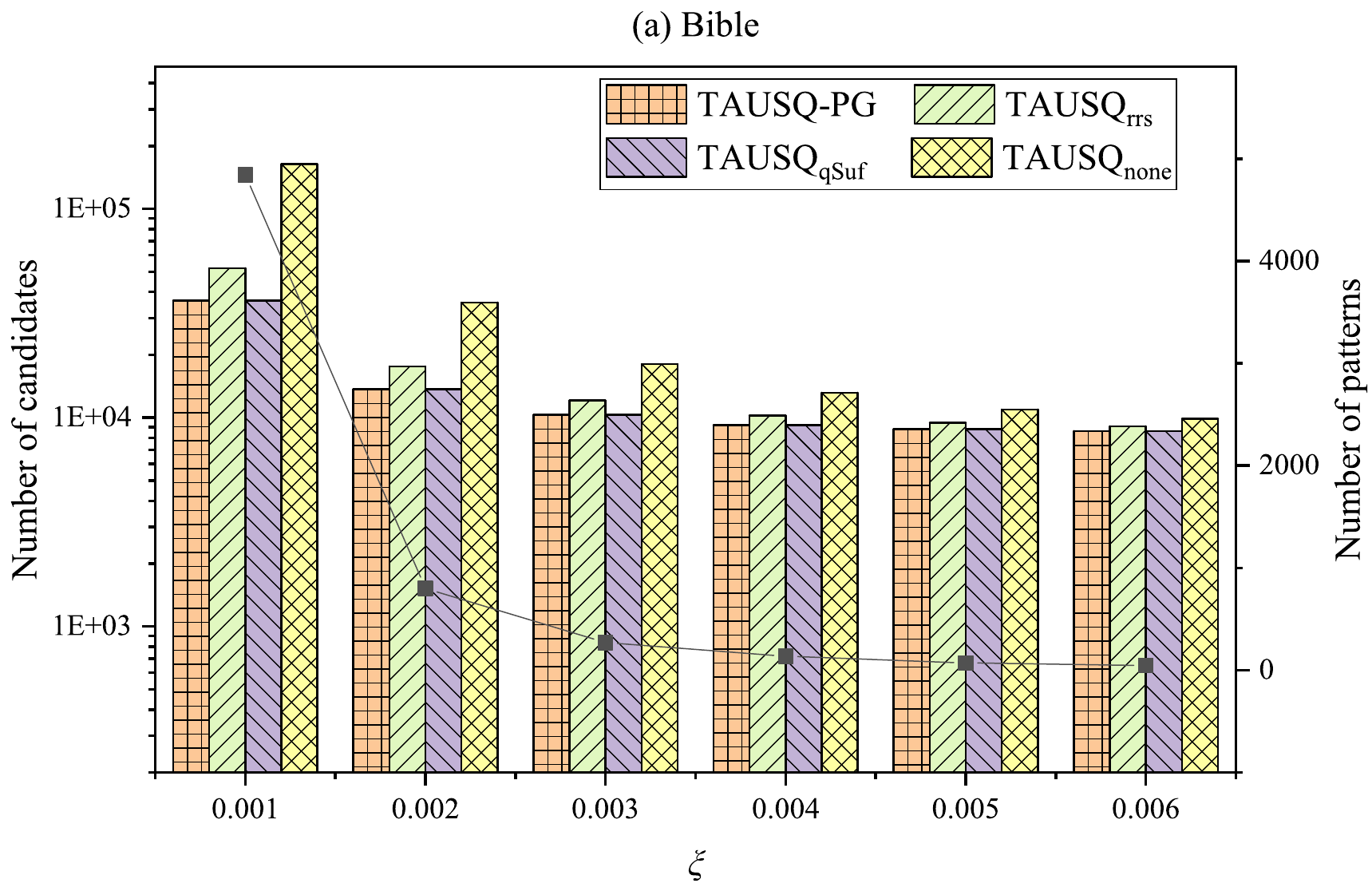}
			\label{fig: 55a}
			\includegraphics[clip,scale=0.15]{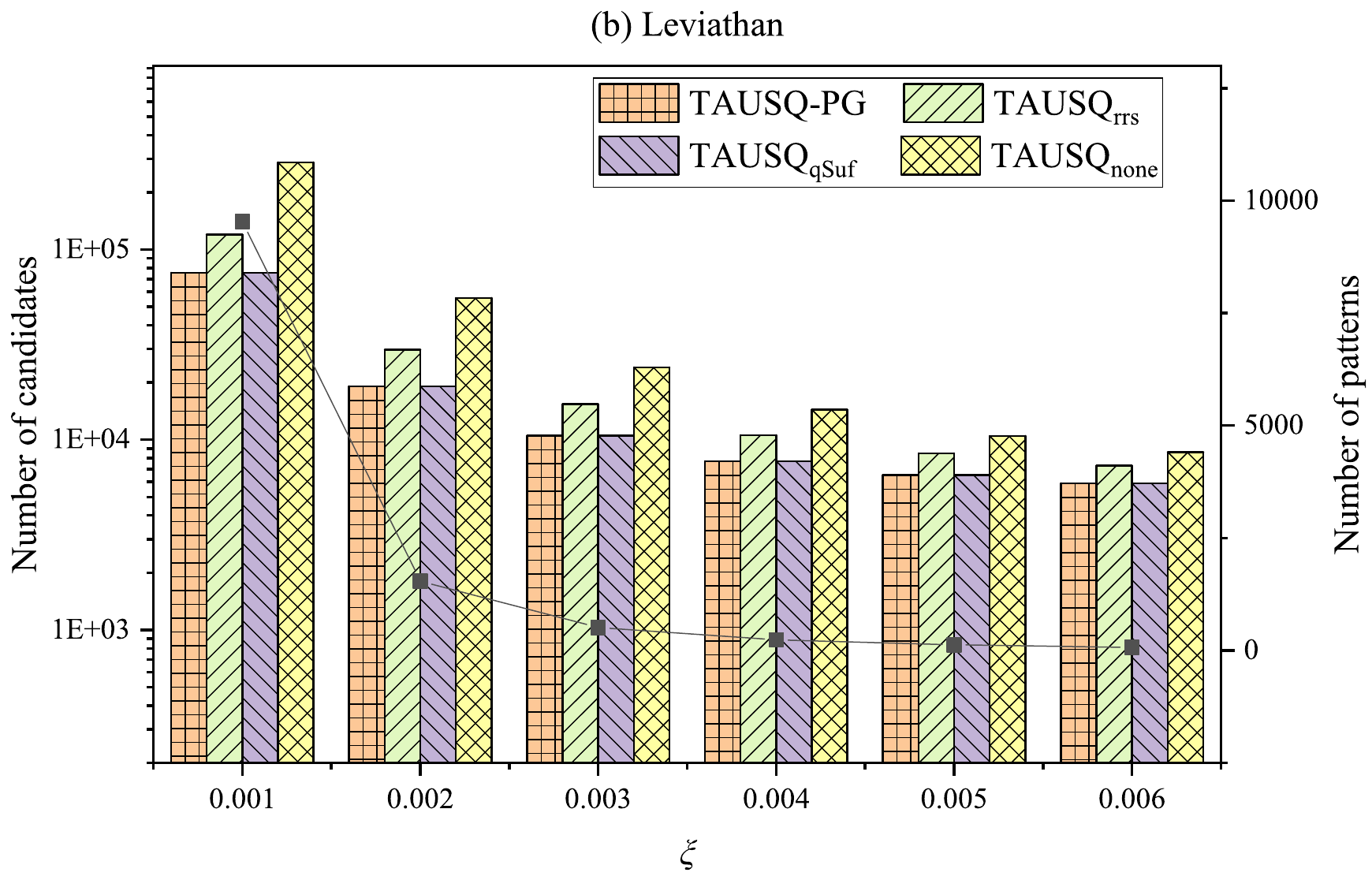}
			\label{fig: 55b}
			\includegraphics[clip,scale=0.15]{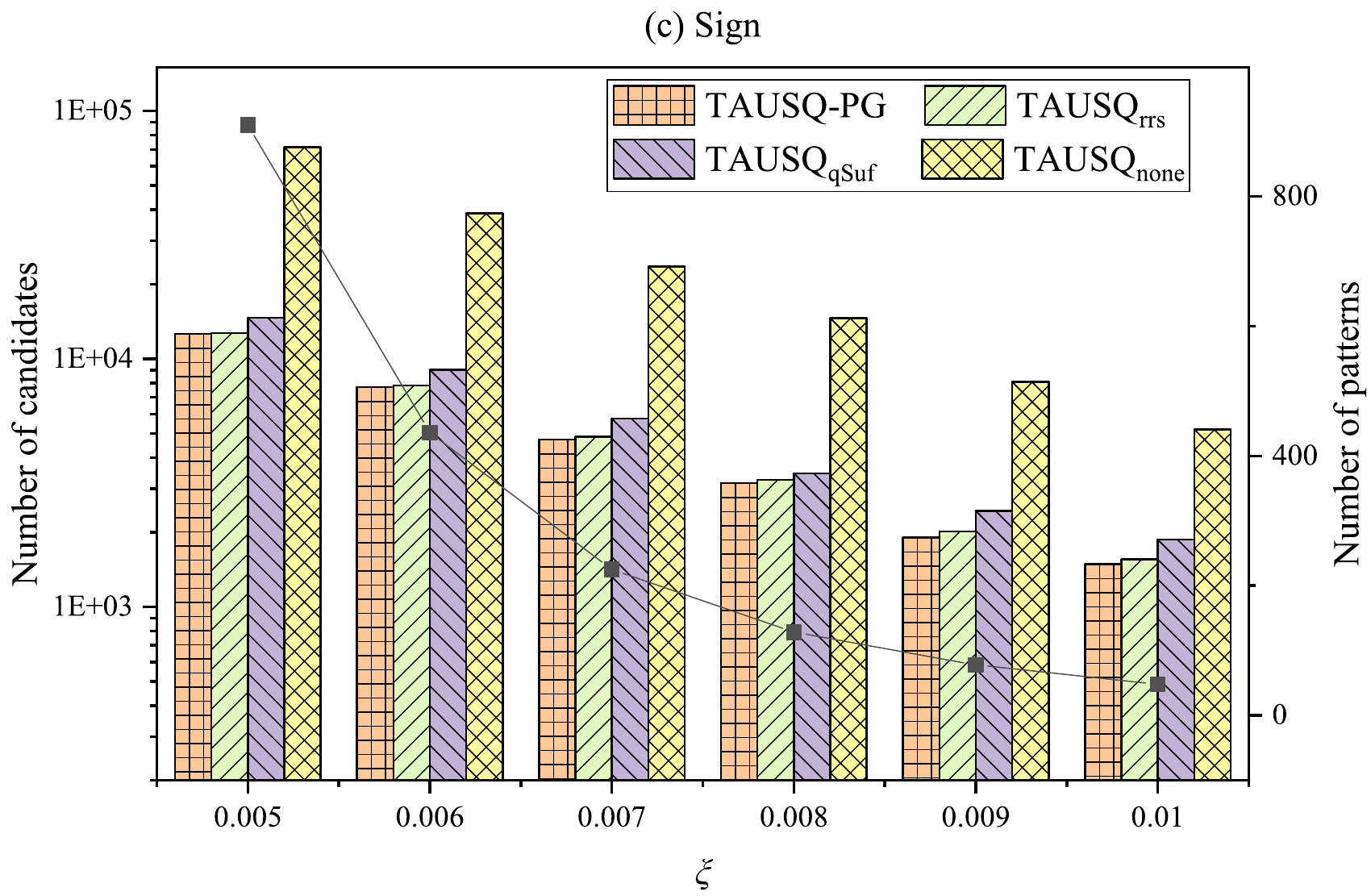}
			\label{fig: 55c}
	\end{minipage}
	\begin{minipage}[t]{0.98\textwidth}
			\includegraphics[clip,scale=0.15]{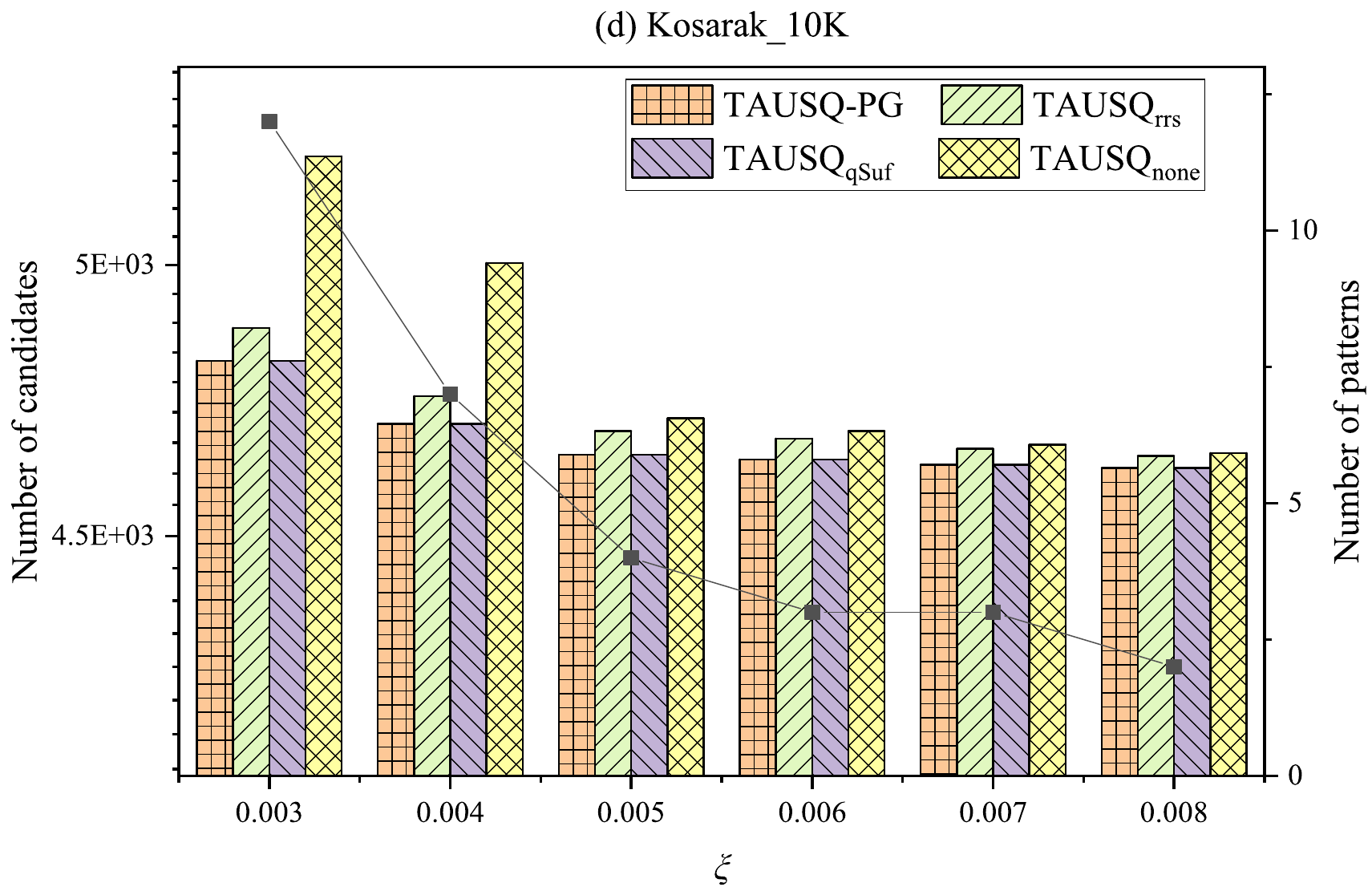}
			\label{fig: 55d}
			\includegraphics[clip,scale=0.15]{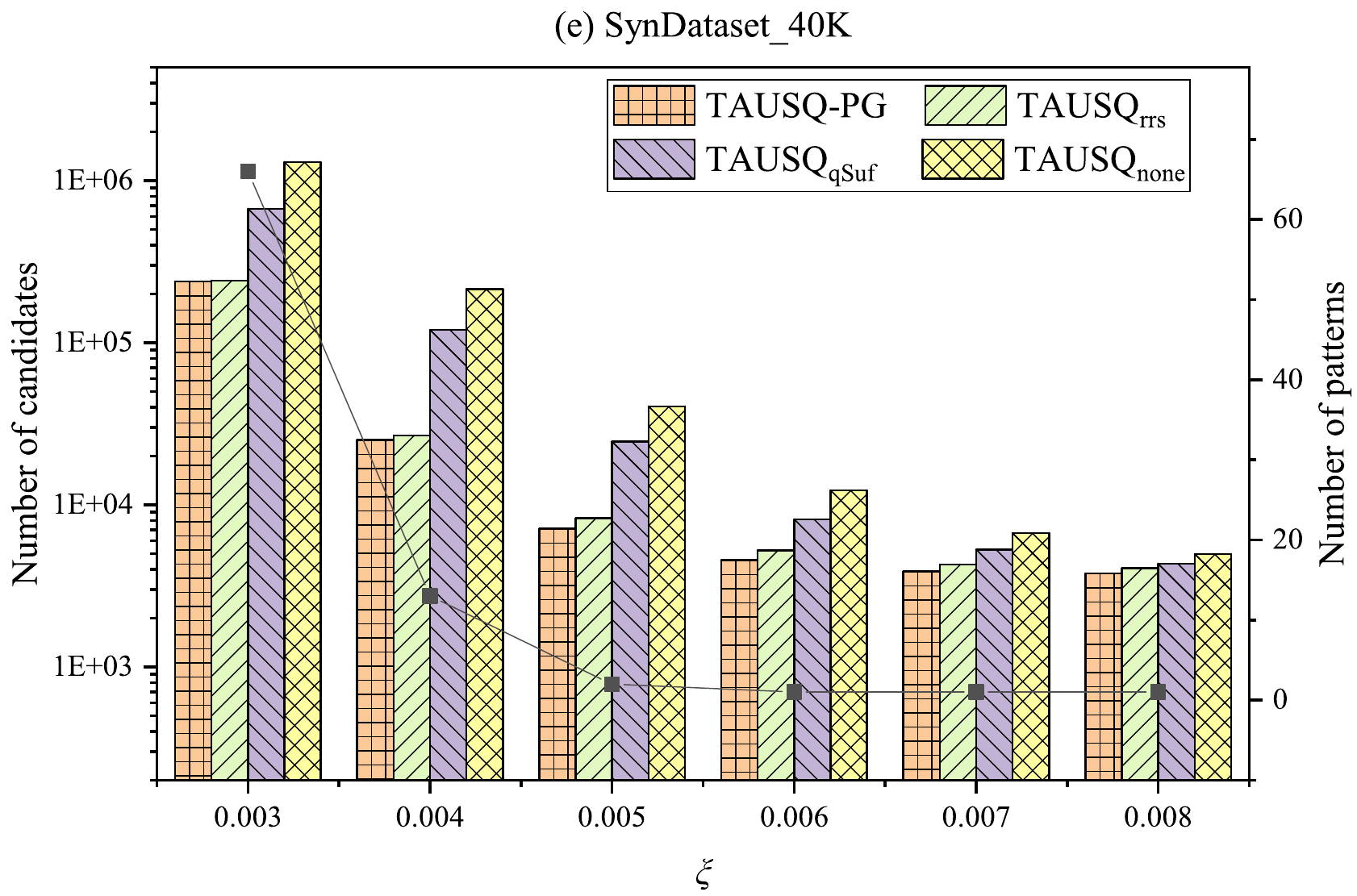}
			\label{fig: 55e}
			\includegraphics[clip,scale=0.15]{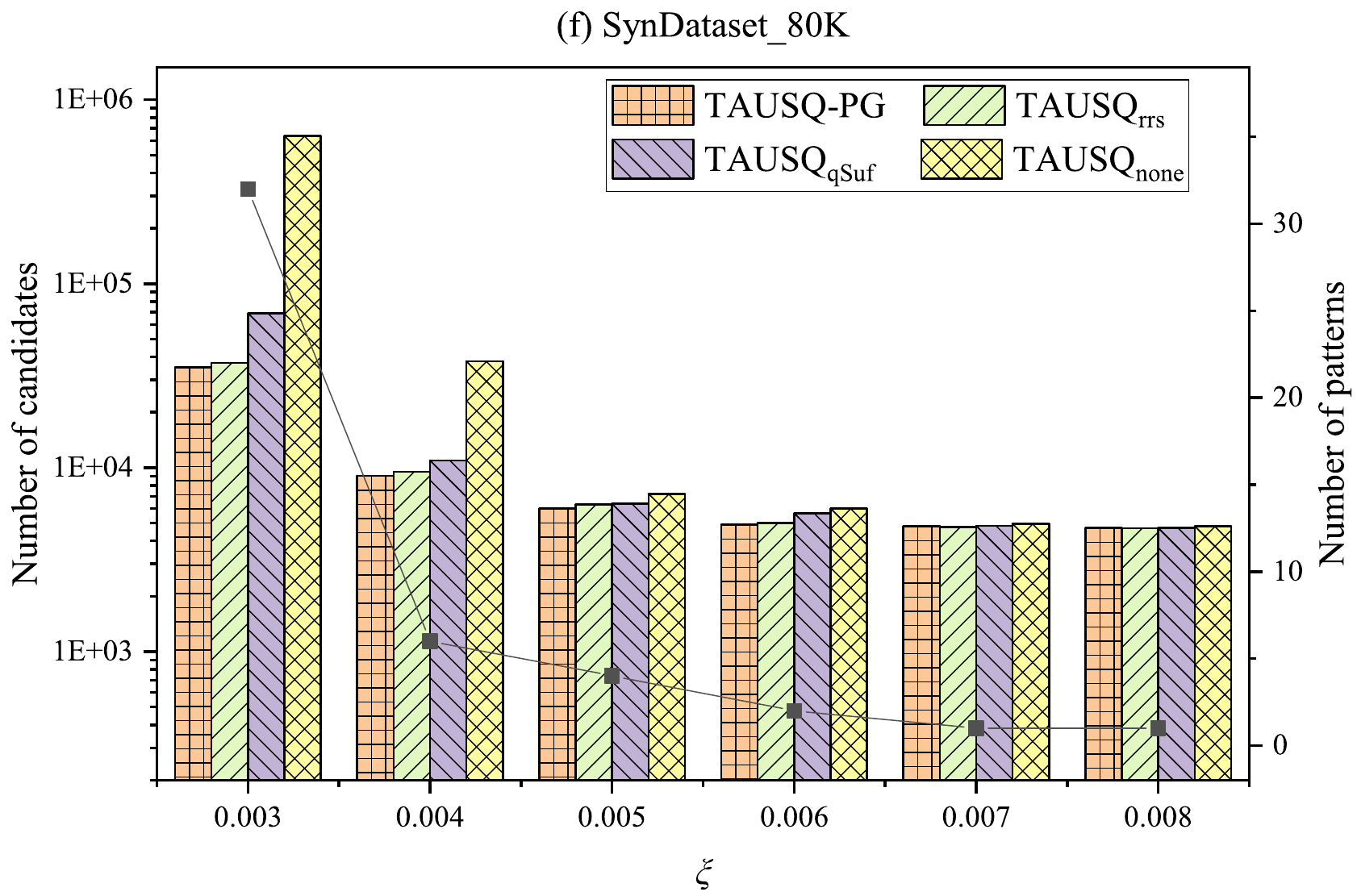}
			\label{fig: 55f}
	\end{minipage}
	\caption{Generated candidate sequences for various upper bound models.}
	\label{fig: 55}
\end{figure}

The outcomes of the experiments are illustrated in Fig. \ref{fig: 54} and Fig. \ref{fig: 55}. From these results, we observe that the most effective upper bound model varies across datasets. In Fig. \ref{fig: 54}(a), Fig. \ref{fig: 54}(b), and Fig. \ref{fig: 54}(d), algorithmic efficiency is primarily improved by incorporating \textit{qSuf}, whereas in Fig. \ref{fig: 54}(c), Fig. \ref{fig: 54}(e), and Fig. \ref{fig: 54}(f), the inclusion of \textit{rrs} plays a more critical role. These differences indicate that the key factors influencing algorithm performance differ by dataset and target sequence characteristics. Therefore, adopting a flexible upper bound modeling strategy that dynamically considers both \textit{rrs} and \textit{qSuf} is essential to maintain consistent performance across diverse data scenarios.

\subsection{Evaluation of the Impact of Varying Target Sequence Lengths}
\label{sec: scalability4Length}

\begin{figure}[ht]
	\centering
	\begin{minipage}[t]{0.98\textwidth}
			\includegraphics[clip,scale=0.17]{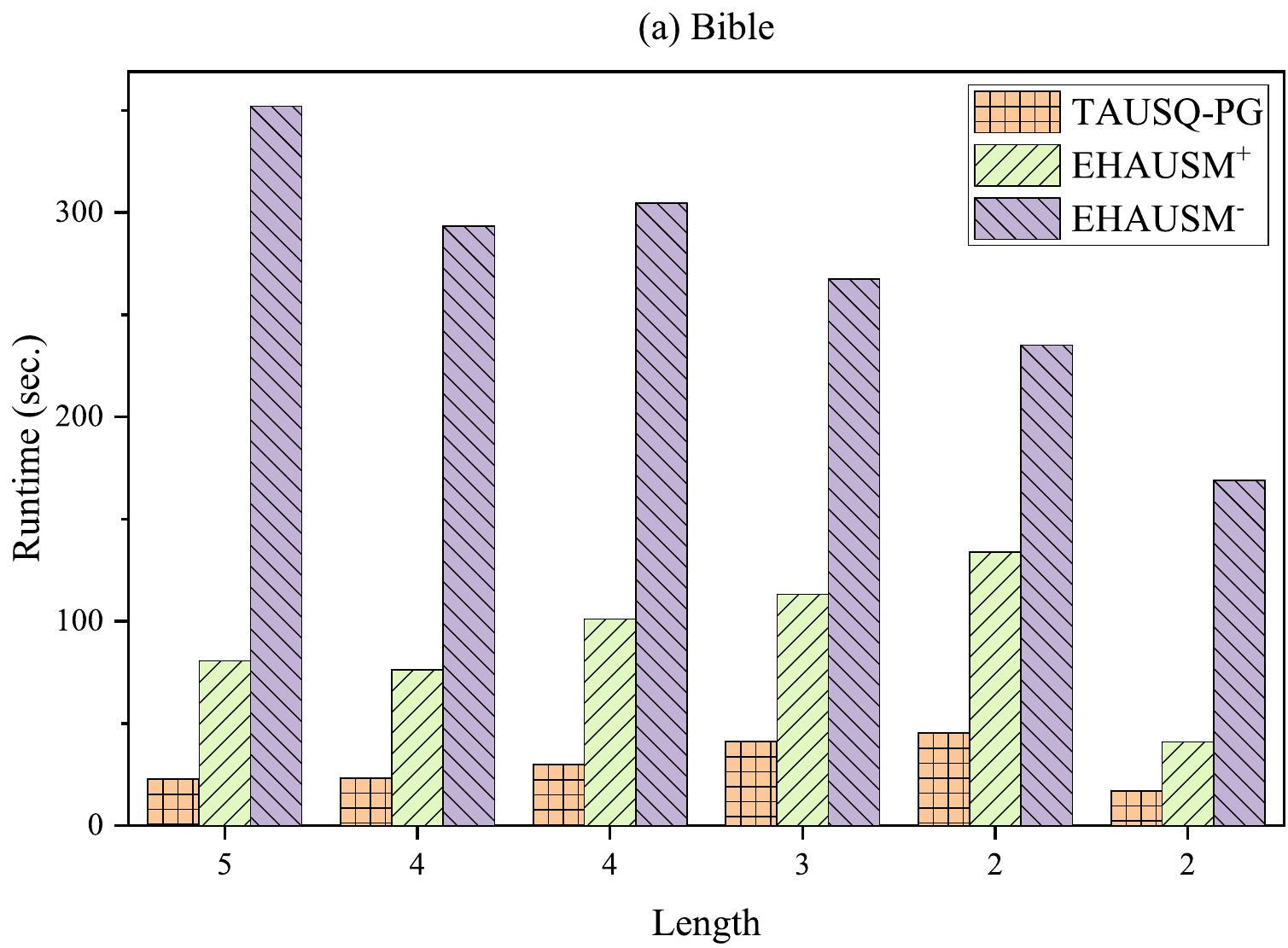}
			\label{fig: 57a}
			\includegraphics[clip,scale=0.17]{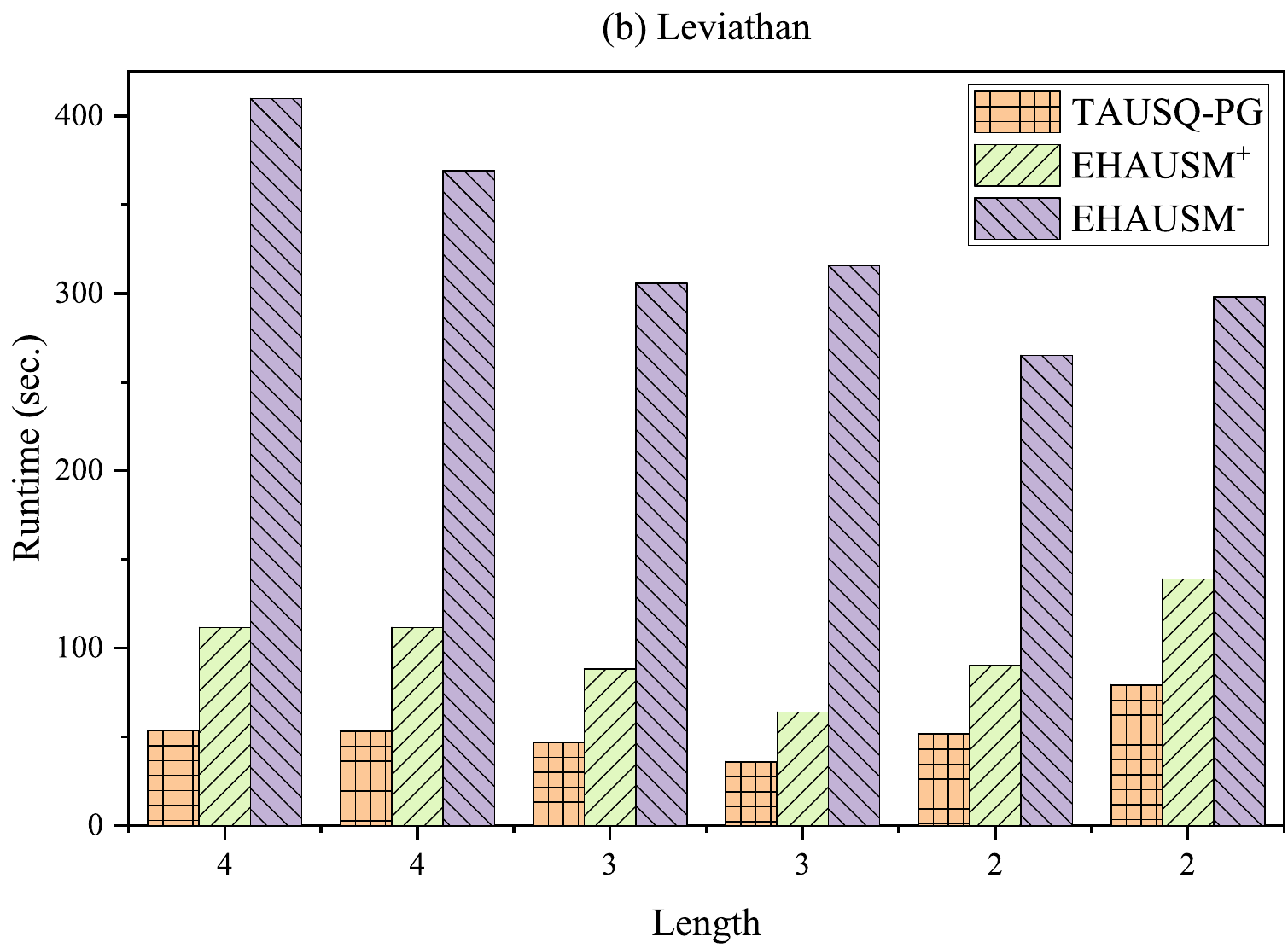}
			\label{fig: 57b}
			\includegraphics[clip,scale=0.17]{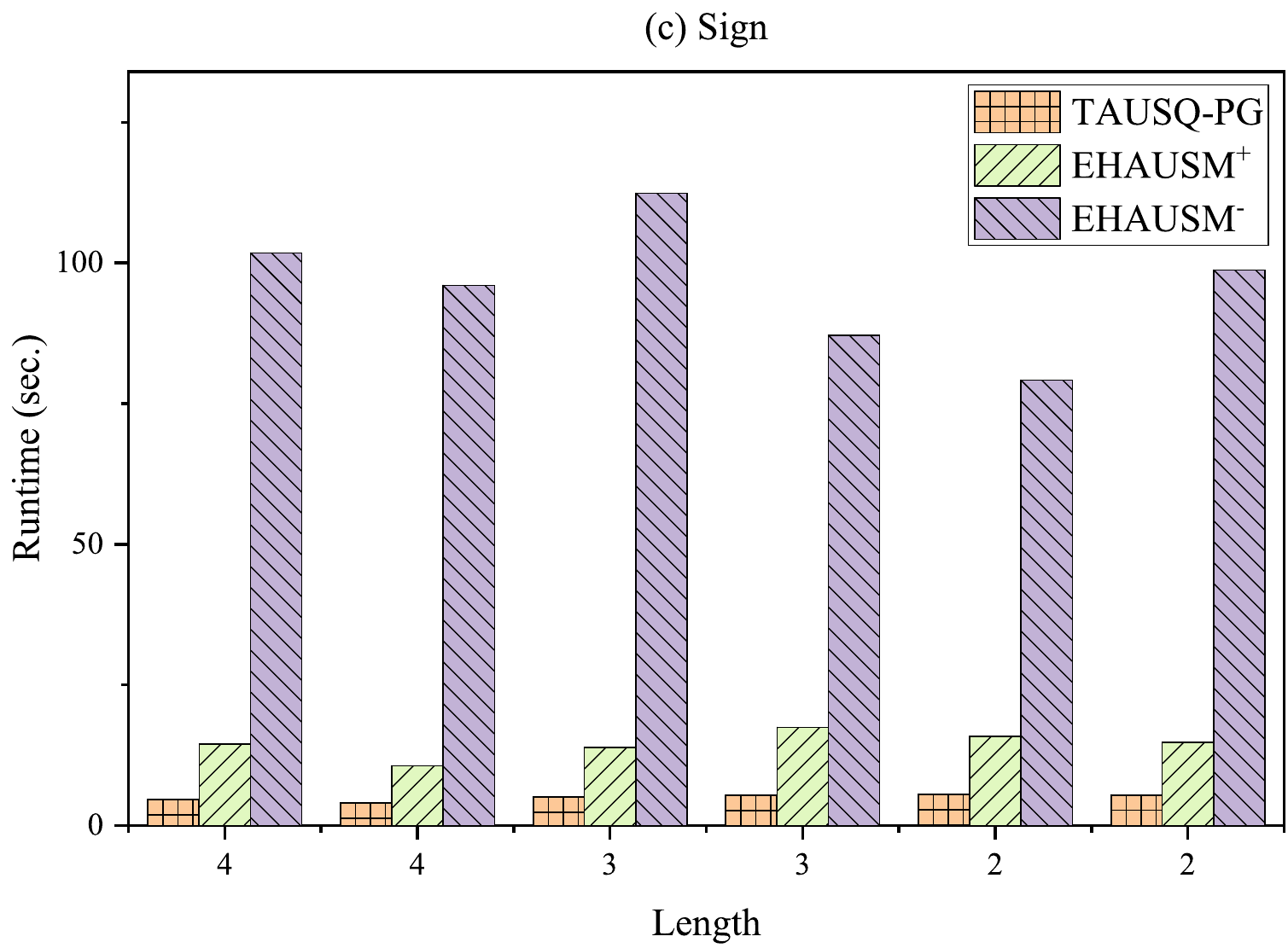}
			\label{fig: 57c}
	\end{minipage}
	\begin{minipage}[t]{0.98\textwidth}
			\includegraphics[clip,scale=0.17]{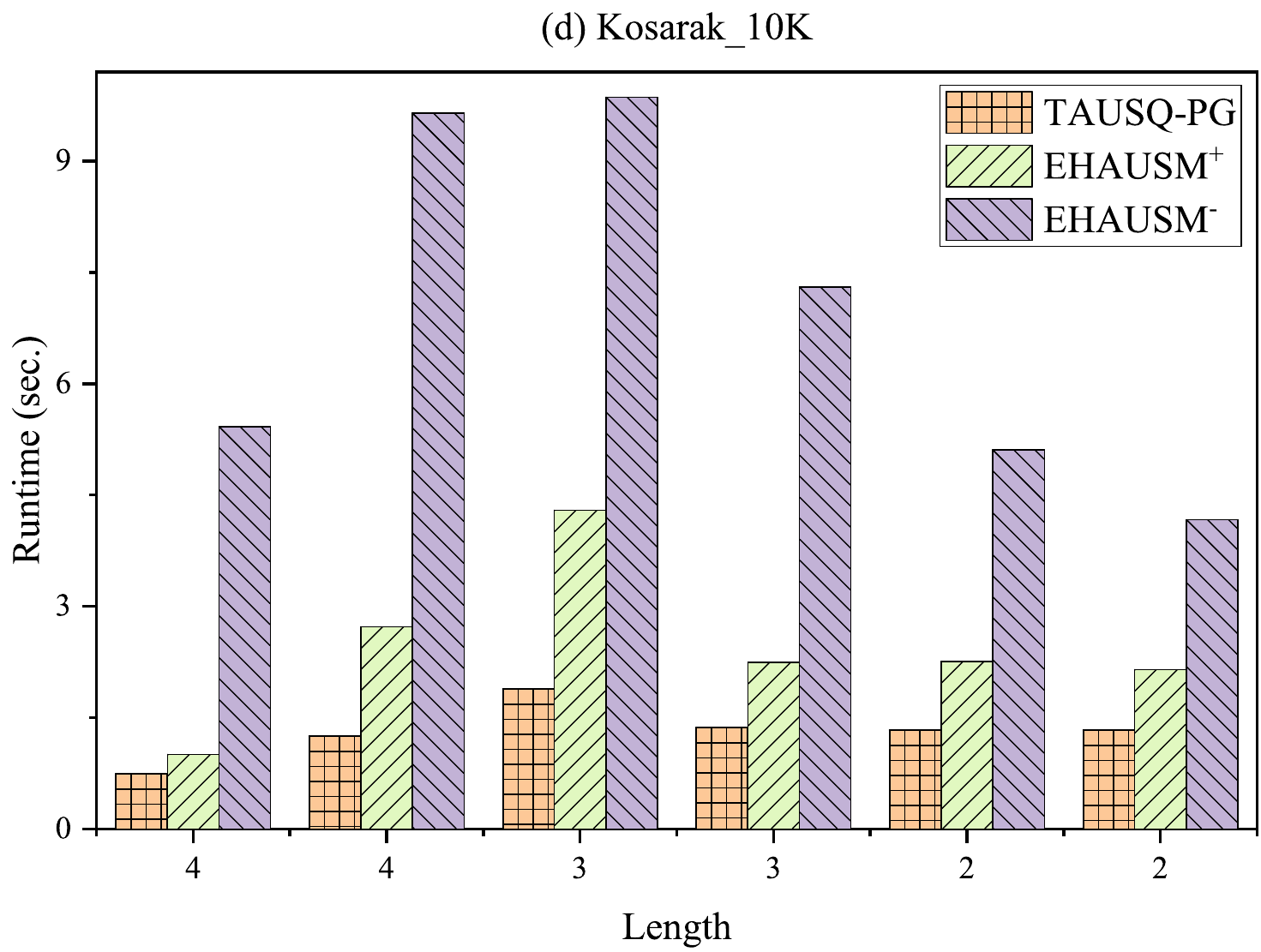}
			\label{fig: 57d}
			\includegraphics[clip,scale=0.17]{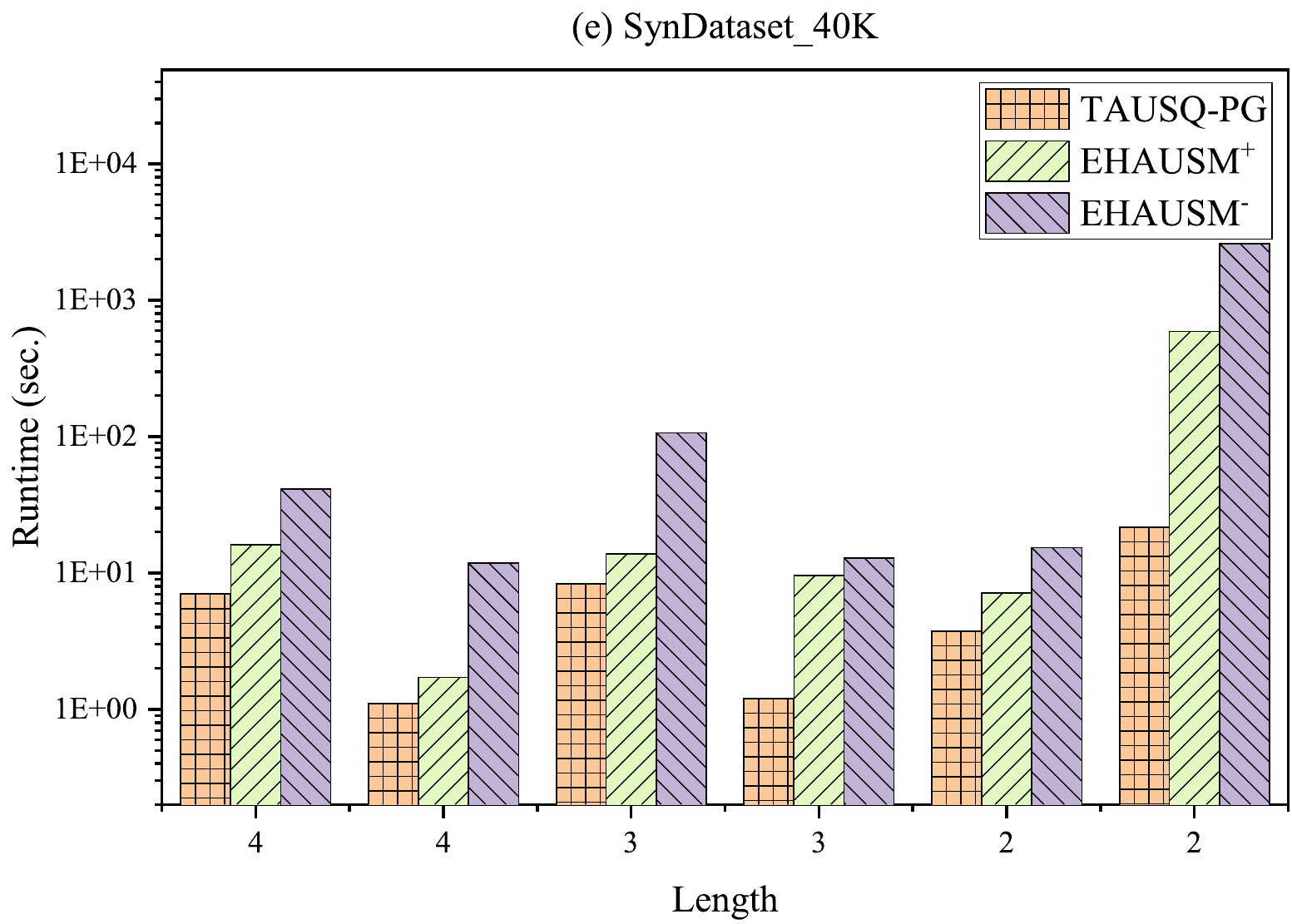}
			\label{fig: 57e}
			\includegraphics[clip,scale=0.17]{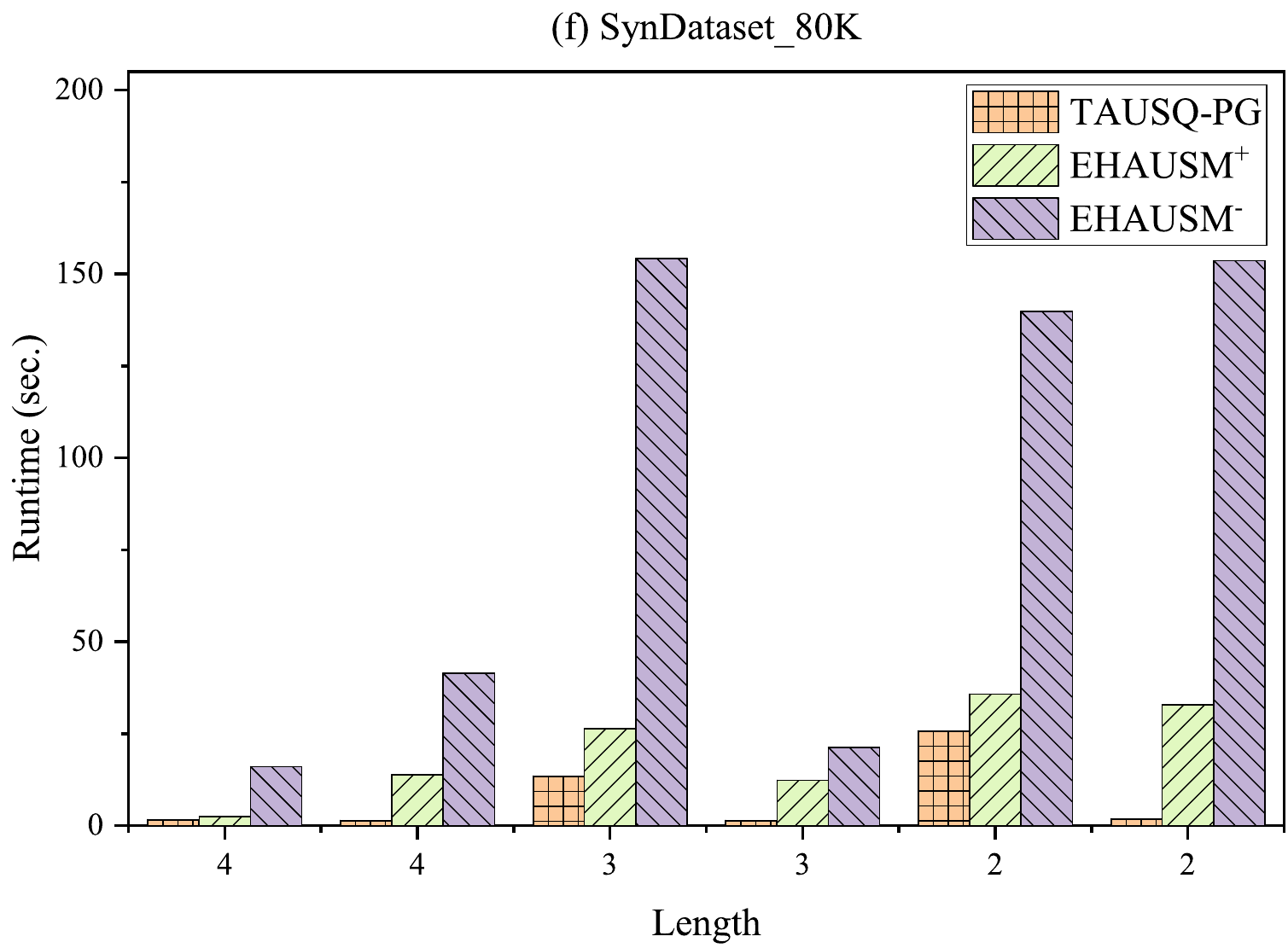}
			\label{fig: 57f}
	\end{minipage}
	\caption{Runtime for various target sequences.}
	\label{fig: 57}
\end{figure}

\begin{figure}[ht]
	\centering
	\begin{minipage}[t]{0.98\textwidth}
			\includegraphics[clip,scale=0.17]{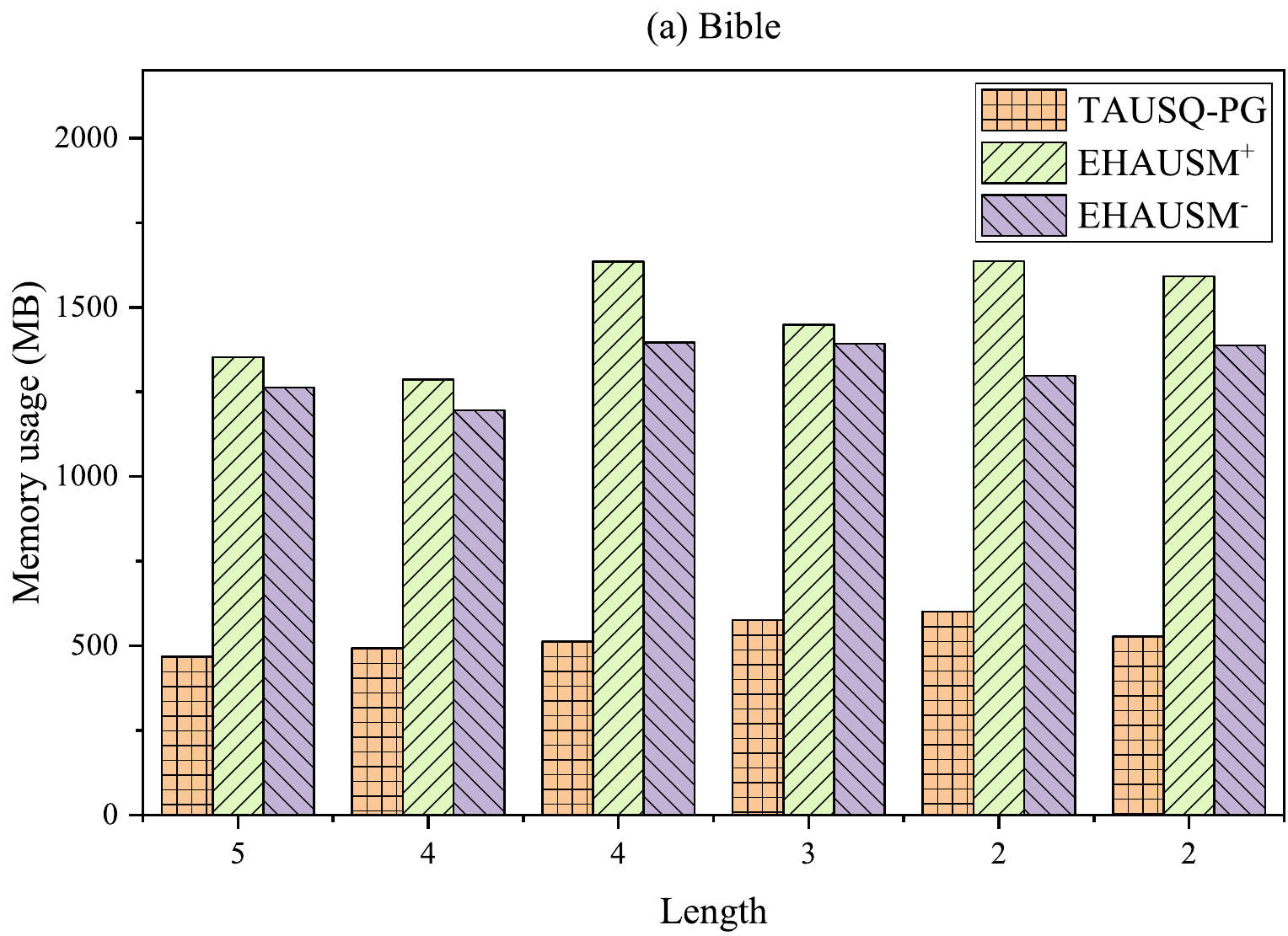}
			\label{fig: 58a}
			\includegraphics[clip,scale=0.17]{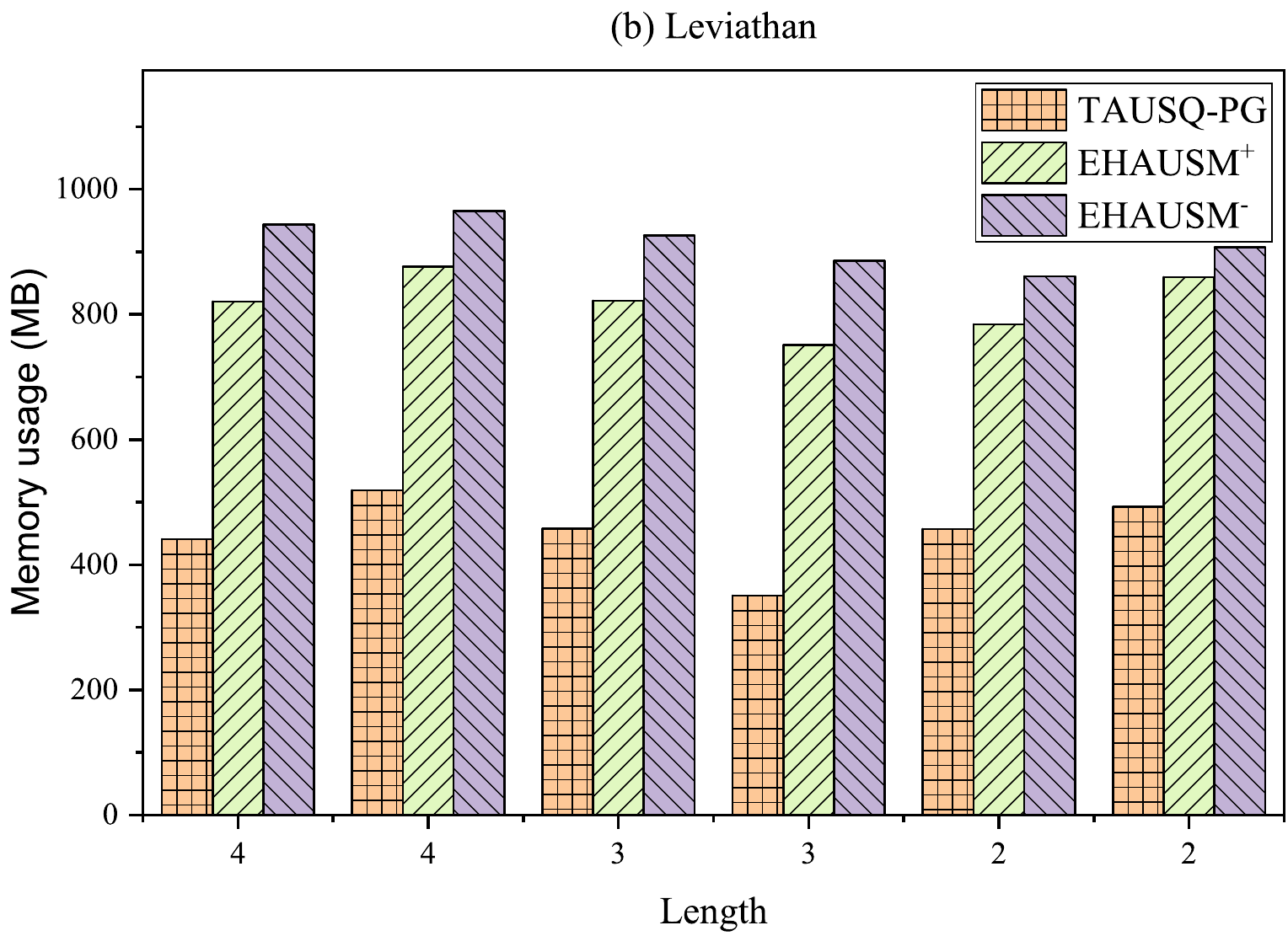}
			\label{fig: 58b}
			\includegraphics[clip,scale=0.17]{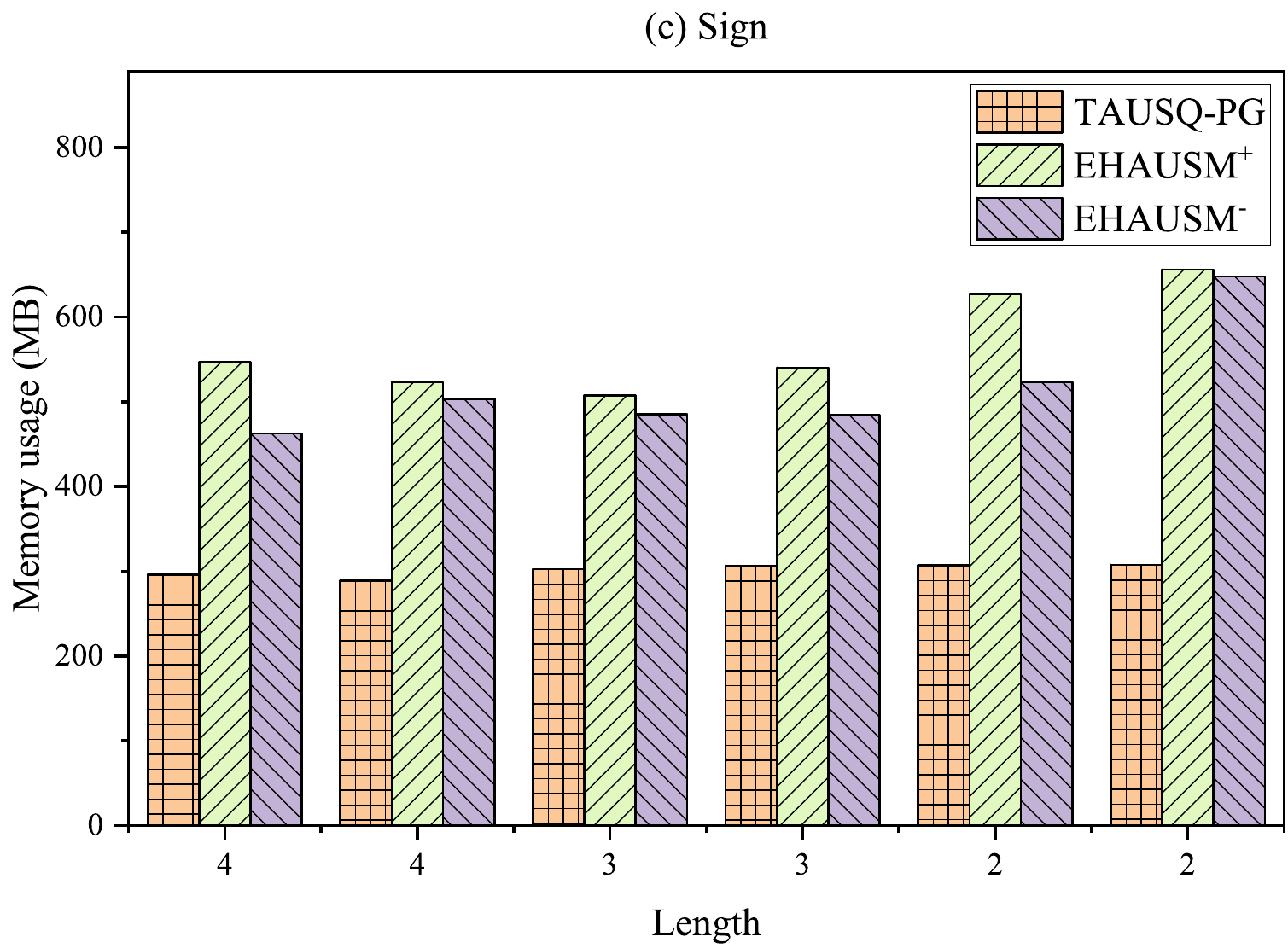}
			\label{fig: 58c}
	\end{minipage}
	\begin{minipage}[t]{0.98\textwidth}
			\includegraphics[clip,scale=0.17]{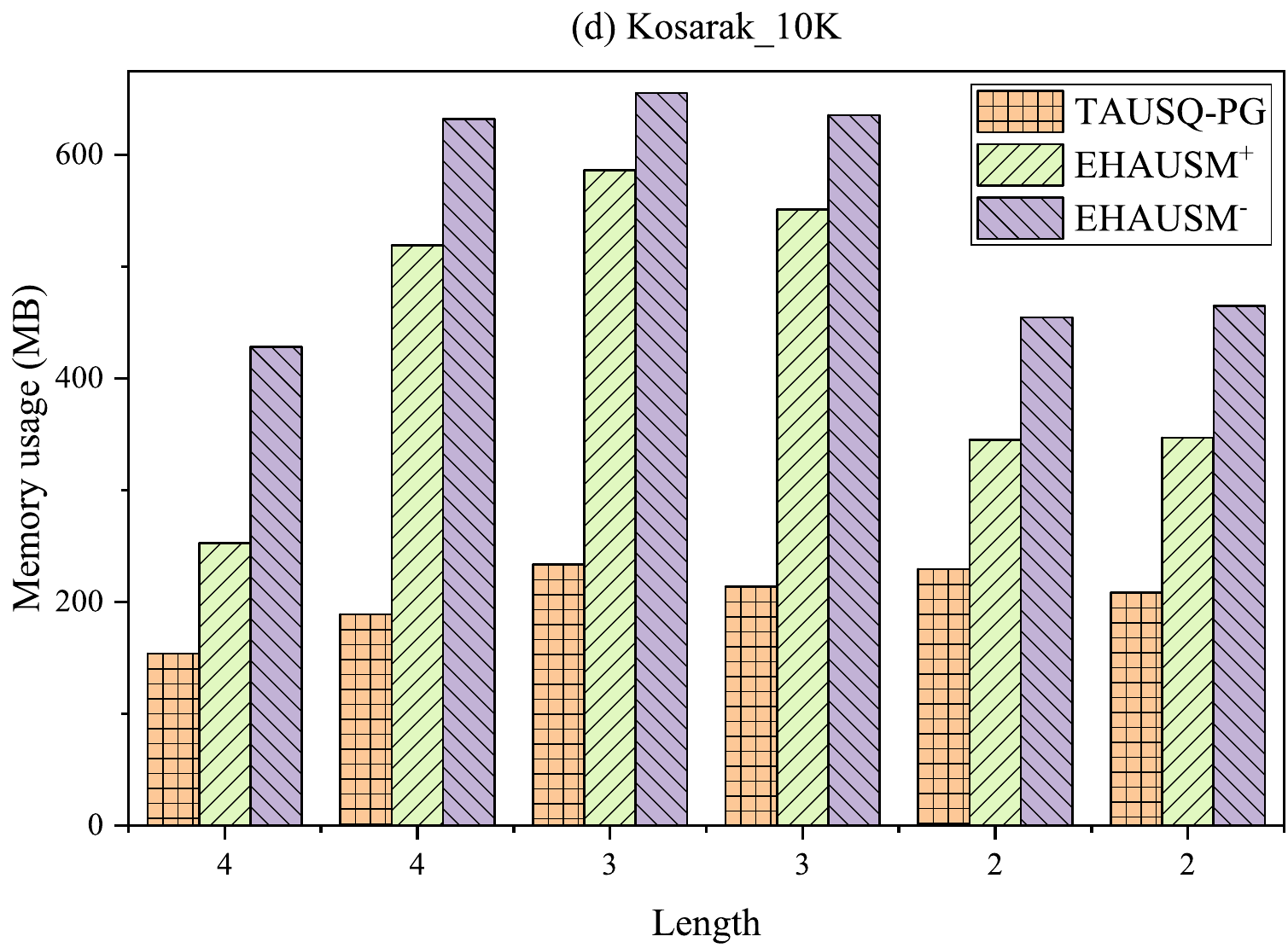}
			\label{fig: 58d}
			\includegraphics[clip,scale=0.17]{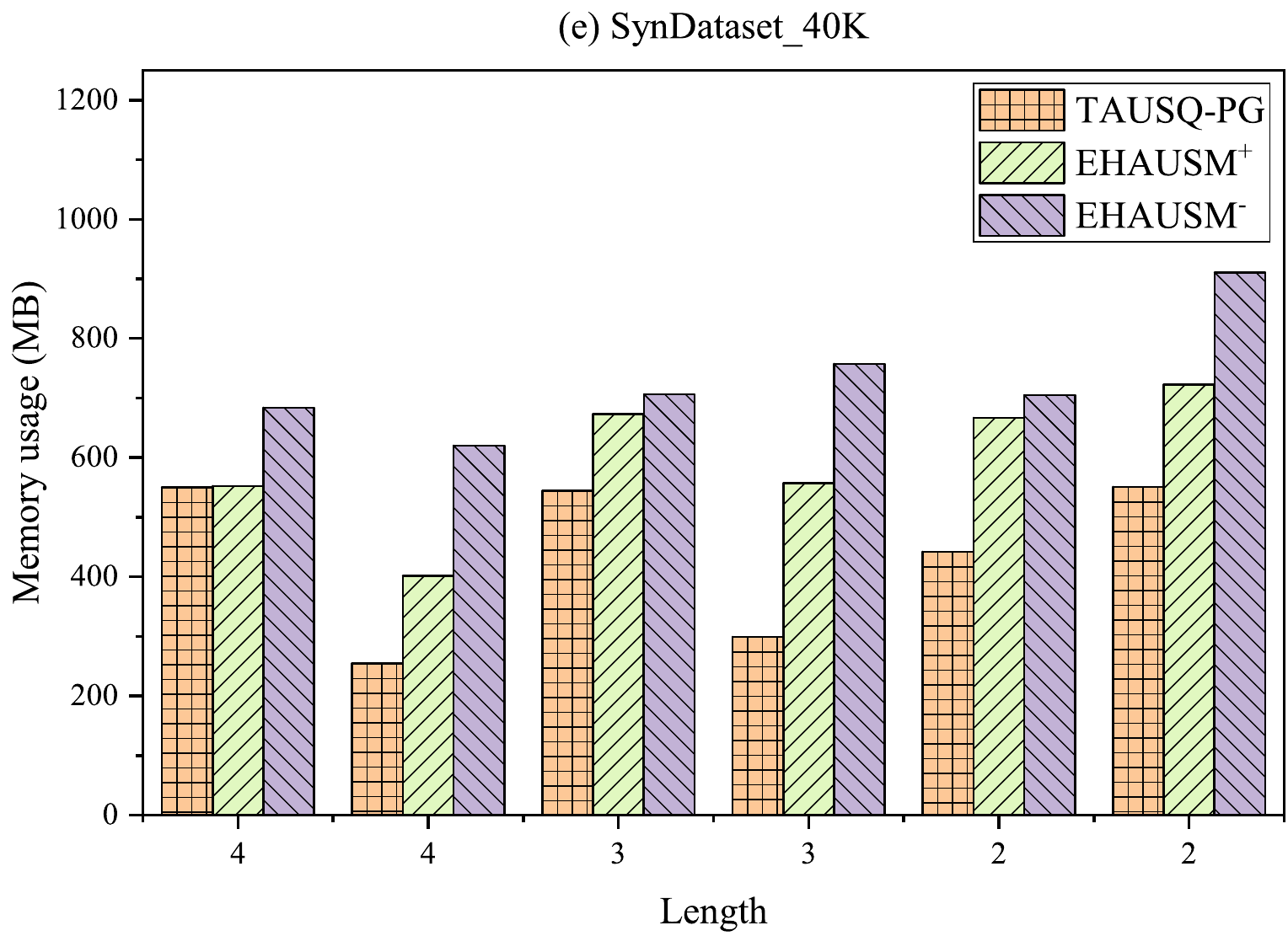}
			\label{fig: 58e}
			\includegraphics[clip,scale=0.17]{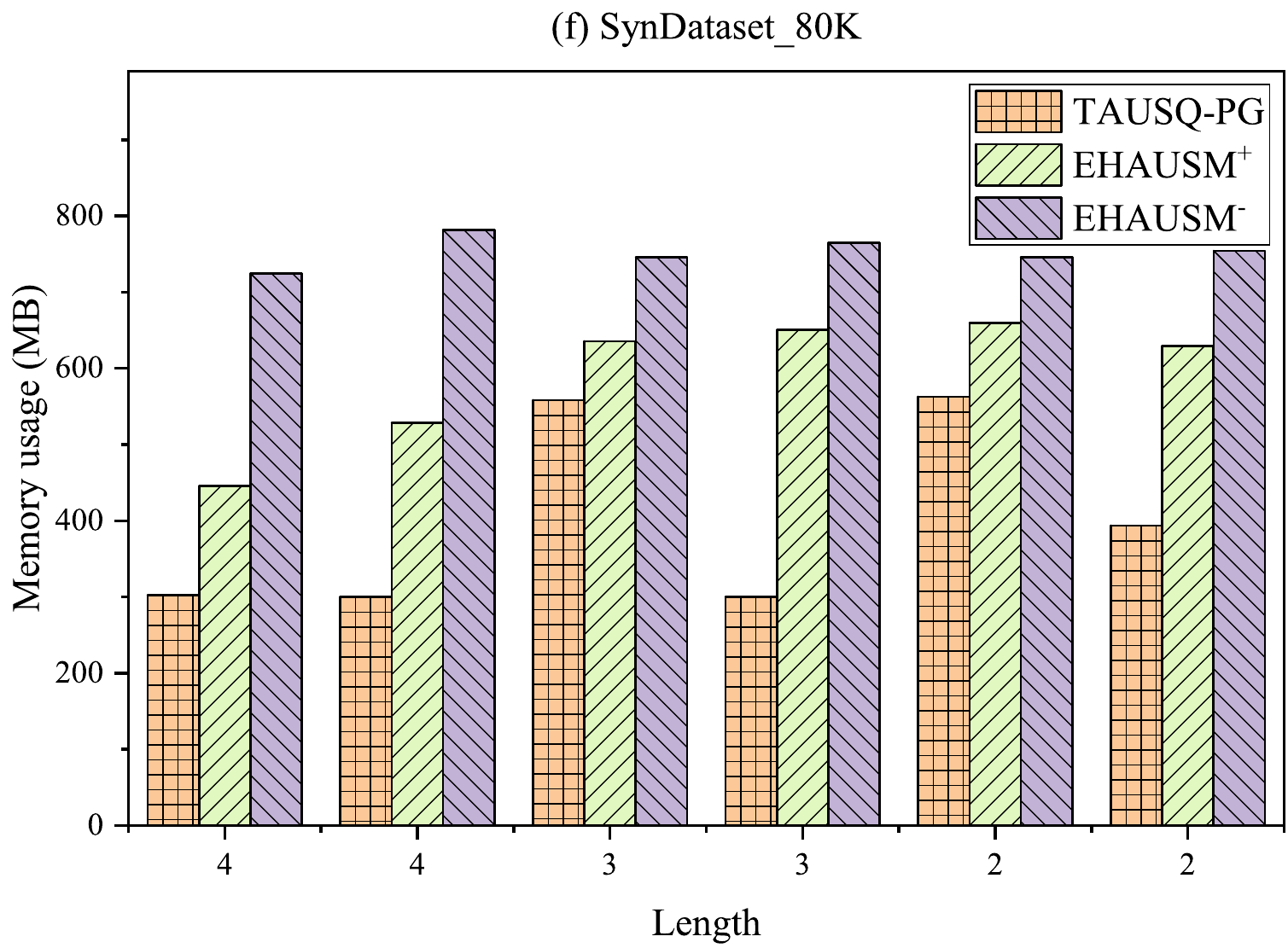}
			\label{fig: 58f}
	\end{minipage}
	\caption{Memory usage for various target sequences.}
	\label{fig: 58}
\end{figure}

In this series of comparative tests, we evaluate the performance of different algorithms across six datasets under varying target sequence lengths. For each dataset, we randomly select target sequences from the top 100 frequent patterns, with selection constrained to match the specified lengths. The threshold parameters for the six datasets are fixed at 0.2\%, 0.1\%, 0.5\%, 0.3\%, 0.3\% and 0.4\%, respectively. The experiments clearly demonstrate that TAUSQ-PG significantly outperforms both \( {\rm EHAUSM}^+  \) and \( {\rm EHAUSM}^- \). In terms of runtime, as illustrated in Fig. \ref{fig: 57}, TAUSQ-PG consistently achieves superior efficiency across all datasets and target sequence lengths. The performance advantage is particularly notable on synthetic datasets such as \( {SynDataset}\_{40K} \) and \( {SynDataset}\_{80K} \). Regarding memory consumption, TAUSQ-PG also demonstrates greater efficiency than the other two methods. As shown in Fig. \ref{fig: 58}, on \( {SynDataset}\_{40K} \), even in the worst-case scenario, the memory consumption of TAUSQ-PG stays below \( {\rm EHAUSM}^+  \). In summary, these experimental findings confirm that TAUSQ-PG offers superior overall efficiency in both memory usage and runtime, even as the length and complexity of target sequences vary. These advantages make TAUSQ-PG particularly well-suited for target-sequence-driven pattern mining tasks across diverse datasets.

\section{Conclusion} \label{sec: conclusion}

The introduction of the average utility concept not only addresses certain limitations of traditional utility-based pattern mining but also provides a fairer and more insightful evaluation criterion. However, many of the generated patterns may still lack practical relevance or fail to meet specific user interests. To address this challenge, this study integrates average utility with TPM, thereby defining the problem of TAUSPM. Herein, we introduce a new algorithm, TAUSQ-PG, which employs a compact data structure specifically optimized for average utility mining. To further improve the efficiency of sequential pattern querying, two matching query flags combined with the position comparison method are introduced. Moreover, the algorithm employs tighter variants of UBs and pruning strategies tailored specifically for the TAUSPM task to further improve efficiency. Experimental findings demonstrate that the proposed algorithm significantly enhances the effectiveness and efficiency of TAUSPM, especially in scenarios involving large-scale datasets with long sequences. Future research will explore several directions. One goal is to continue refining the TAUSQ framework and applying it to real-world applications. Another objective is to explore advanced topics in average utility mining, as we believe this line of research has strong potential for uncovering patterns of higher interest. We also plan to extend our methods to support more diverse task requirements and complex data characteristics. Specifically, these include constraints such as contiguous patterns, uncertain or noisy data, and datasets with negative utility items.

\bibliographystyle{ACM-Reference-Format}

\bibliography{main.bib}

\end{document}